\documentclass[a4paper, 11pt]{article} 

\usepackage[T1]{fontenc}
\usepackage[utf8]{inputenc}
\usepackage{float}
\usepackage{graphicx}
\usepackage{bm}
\usepackage{titling}
\usepackage{subcaption} 
\usepackage[thinlines]{easytable}
\usepackage[hidelinks]{hyperref}
\usepackage{wrapfig}
\usepackage{dingbat}
\usepackage{color}
\usepackage{refcount}
\usepackage[table]{xcolor}
\usepackage{fancyhdr} 
\usepackage[toc,page]{appendix}
\usepackage{bbm}
\usepackage{mathrsfs} 
\usepackage{comment}
\usepackage{amsmath,amssymb,amsthm}
\usepackage{mathtools}

\usepackage[english]{babel}

\usepackage{natbib}

\usepackage{xcolor}

\makeatletter
\@addtoreset{equation}{section}

\makeatother

\usepackage[margin=3cm]{geometry}

\usepackage{setspace} %
\theoremstyle{Theorem}
\newtheorem{theo}{Theorem}
\theoremstyle{Lemma}
\newtheorem{lem}[theo]{Lemma}
\newtheorem{prop}[theo]{Proposition}
\newtheorem{rem}[theo]{Remark}
\newtheorem{assumption}[theo]{Assumption}

\theoremstyle{Definition}
\newtheorem{defi}[theo]{Definition}
\theoremstyle{Corollary}
\newtheorem{cor}[theo]{Corollary}
\newtheorem{example}[theo]{Example}

\DeclareFontFamily{U}{matha}{\hyphenchar\font45}
\DeclareFontShape{U}{matha}{m}{n}{
      <5> <6> <7> <8> <9> <10> gen * matha
      <10.95> matha10 <12> <14.4> <17.28> <20.74> <24.88> matha12
      }{}
\DeclareSymbolFont{matha}{U}{matha}{m}{n}
\DeclareFontSubstitution{U}{matha}{m}{n}

\DeclareFontFamily{U}{mathx}{\hyphenchar\font45}
\DeclareFontShape{U}{mathx}{m}{n}{
      <5> <6> <7> <8> <9> <10>
      <10.95> <12> <14.4> <17.28> <20.74> <24.88>
      mathx10
      }{}
\DeclareSymbolFont{mathx}{U}{mathx}{m}{n}
\DeclareFontSubstitution{U}{mathx}{m}{n}

\DeclareMathDelimiter{\vvvert}{0}{matha}{"7E}{mathx}{"17}
\DeclareMathAlphabet{\mathpzc}{OT1}{pzc}{m}{it}

\newcommand{\E}{{\mathbb E}}

\newcommand{\N}{{\mathbb N}}
\newcommand{\Z}{{\mathbb Z}}

\newcommand{\R}{{\mathbb R}}

\newcommand{\Tr}{\textnormal{Tr}}

\newcommand{\LR}{\lambda}

\newcommand{\bsec}{\pmb{\beta}}

\newcommand{\be}{\begin{equation}} 
\newcommand{\ee}{\end{equation}}

\numberwithin{theo}{section}

\newcommand{\TK}{\textcolor{red}}

\allowdisplaybreaks

\begin{document}

\setlength{\droptitle}{-4em} 

\title{ \textbf{ Validating Approximate Slope Homogeneity in Large Panels}}


\author{
{\small  Tim Kutta\thanks{Correspondences may be addressed to Tim Kutta (tim.kutta@rub.de).}, Holger Dette} \\
{\small Colorado State University,  Ruhr-Universit\"{a}t Bochum} 
}

\maketitle
\begin{abstract}
The analysis of large data panels is a cornerstone of modern economics. An important benefit of panel analysis is the possibility to reduce noise by intersectional pooling. In order to control the resulting heterogeneity bias, various works have proposed homogeneity tests - in the context of linear panels for slope homogeneity. However, such tests risk detecting even minute deviations from perfect homogeneity, discouraging pooling even when practically beneficial. To address this problem, we introduce the hypothesis of approximate slope homogeneity and propose a test for this more pragmatic hypothesis. In contrast to most of the related literature, our test statistic is asymptotically pivotal and still valid in the presence of simultaneous temporal and intersectional dependence (even for large panels). We also demonstrate uniform consistency against classes of local alternatives.  A simulation study and a data example underline the usefulness of our approach.
 \end{abstract}


\bigskip

%
%
%
%
%
%
%
%
%

\section{Introduction}

%
%
%
%
%
%
%
%
%

The analysis of large data panels is a cornerstone of modern economics. Key benefits of data panels derive from intersectional pooling, i.e. combining observations across individuals, which boosts sample size and consequently reduces variance of statistical estimates. Enhanced stability of pooled estimates compared with their individual counterparts and the associated gains in interpretability, have been testified by a number of classical studies, such as \cite{Baltagi1997} and \cite{Baltagi2003}. Moreover, pooled estimates are attractive from the standpoint of statistical inference, as they circumvent  the problem of multiple testing (for each individual) and allow broader statements about the underlying population.
Nevertheless, it is not true that “anything goes” in terms of intersectional pooling: If the observed individuals are so heterogeneous that no representative trend exists among them,  pooled estimators might pay a smaller variance, by an immense bias (see e.g. \cite{Hsiao1997, Hsiao}). 

In linear panels, the question of homogeneity concerns the individual regression slopes. 
As in most cases slope homogeneity cannot be determined a priori, many authors have developed 
goodness of fit tests  to investigate if this assumption is reasonable  (see  \cite{Zellner, Swamy, Pesaran1996, Phillips, Pesaran,  Blomquist2013b, breitung2016} among others). 
These tests have come to encompass evermore practical scenarios, such as panels with large intersections
and to some extent dependent observations. However, the use of homogeneity tests has not been met with unequivocal approval. Critics  have pointed out that tests tend to reject homogeneity, even when practically pooling is advisable (see e.g. \cite{Baltagi2008, Blomquist}). Moreover and more fundamentally, the hypothesis of perfect homogeneity is usually implausible in high dimensional panels, such that it seems conceptually questionable to test
it in the first place.
These insights have spurred recent trends of searching for alternative poolability criteria (see \cite{Campello}) or to subdivide heterogeneous individuals into more homogeneous groups (\cite{Lin, Blomquist, Sarafidis2015,  Su2016, Wang}).

In this paper,  we contribute to this discussion. After a brief review of 
some related literature on slope homogeneity tests,
which is relevant in this context, we propose in  Section \ref{Section2}
the pragmatic notion of \textit{approximate slope homogeneity}. 
This hypothesis is rigorously defined in Section 
\ref{sec_22}. It
expresses the belief, that there is always some amount of individual heterogeneity in the data 
and the real question is whether it is small enough to permit a pooled analysis or not.
In Section \ref{Section3} we present a test for approximate slope homogeneity,  which keeps 
its preassigned nominal level for large sample sizes and cross-sections. Moreover, we establish uniform
 consistency of this test against classes of local alternatives. 
In contrast to most works on slope homogeneity tests, which impose strong assumptions
either on the dependence structure of the model errors (such as  temporal and intersectional independence in \cite{Pesaran}), or  on the panel itself (such as  small intersection and long time frame in \cite{Blomquist}), our approach works for panels with large intersections, even in the presence of simultaneous intersectional and temporal dependence. We achieve this by constructing a \textit{self-normalized} statistic, that automatically cancels out the long-run variance and is thus asymptotically distribution free. This
makes our method easy to implement and  interpret. Moreover, compared to alternatives such as bootstrap,  self-normalizations are computationally parsimonious. 
Mathematically, our approach  rests on a weak invariance principle for a time-sequential dispersion measure of the slopes. 
Finally, we investigate finite sample properties in Section \ref{Section4} by means of a simulation study and demonstrate practical applicability by virtue of a data example. Proofs and technical details are deferred to the supplement.

\section{Approximate Slope Homogeneity} \label{Section2}

We  consider the linear panel model
\begin{equation} \label{model_1}
y_{i,t} =x_{i,t}' \beta_i+\varepsilon_{i,t} \quad \quad t=1,...,T, \,\, i=1,...,N,
\end{equation}
where $x_{i,t}$ is a $K$-dimensional vector of regressors, $\beta_i$ is a $K$-dimensional vector of slope coefficients and $\varepsilon_{i,t}$ a centered, real model error with unknown distribution.  We call $t$ the \textit{time dimension} of the panel and $i$ the \textit{individual} or \textit{intersectional component}.  We  collect all equations concerning the $i$th individual in the model
\begin{equation} \label{model_2}
y_{i} =X_{i} \beta_i+\varepsilon_{i} \quad \quad  i=1,...,N,
\end{equation}

where $X_{i}=(x_{i,1},...,x_{i,T})' $ is a regression matrix of dimension $T \times K$ and $y_{i}=(y_{i,1},...,y_{i,T})'$ and $\varepsilon_{i}=(\varepsilon_{i,1},...,\varepsilon_{i,T})'$ are  $T-$dimensional vectors. Sometimes we  refer to all regressors collectively and therefore define the compounded matrix $\mathbf{X}=(X_1,...,X_N)$.  
Often panel models comprise constant, individual specific intercepts, which are omitted in model \eqref{model_2} for simplicity (intercepts can be eliminated by subtracting the individual mean over the temporal observations, see Remark \ref{rem_intercepts} for details).

The analysis of a data panel usually begins by assessing the variability between the  slopes $\beta_1,...,\beta_N$. This can be done by a slope homogeneity test, which composes a summary statistic to decide whether all slopes are equal (the hypothesis) or not. 
In Section  \ref{sec_21} we review some literature  on these  tests as well as recent alternatives. Afterwards, in Section \ref{sec_22}, we introduce and discuss the related, but weaker notion of {\it approximate slope homogeneity.}

\subsection{Testing Slope Homogeneity} \label{sec_21}

The literature on slope homogeneity tests is extensive and therefore we confine ourselves to some important works most closely related to our own. 
The hypothesis of (exact) \textit{slope homogeneity} in model \eqref{model_2} proposes equality of all individual slopes, i.e.
\begin{equation}
    \label{det1}
H_0^{exact}: \beta:=\beta_1=...=\beta_N.
\end{equation}
Roughly speaking, $H_0^{exact}$ implies the same impact of the regressors on each individual - a requirement that is usually not \textit{a priori} clear, explaining why validation by tests is deemed necessary. However, if $H_0^{exact}$ does hold, it opens the door to a profound data analysis, that is otherwise infeasible. For instance, it is possible to analyze in detail the  regressors' influence on the dependent variable, by considering the entries of $\beta$ (positive or negative effects), test for the overall explanatory power of the model ($\beta=0$) or check whether $\beta$ conforms to an otherwise predetermined model. In contrast, individual analysis is notoriously unstable and inference is plagued by the multiple testing problem. 

A simple way to test $H_0^{exact}$ is provided by the standard  $F$-test (see e.g. \cite{Pesaran}). The $F$-test is asymptotically valid for $T \to \infty$ and fixed $N$,  if the model errors are centered,  independent, homoscedastic and satisfy exogeneity.  While the $F$-test is well-known and still prevalent in many applications, the requirement of $N<T$ is too restrictive for many economic panels, where the number of individuals usually exceeds that of temporal observations. Additionally, the assumption of homoscedasticity is often difficult to justify and may itself require pre-testing. 

Dispersion type statistics as developed by \cite{Swamy} are a way to tackle  both  problems. 
Swamy's test statistic 
\begin{equation} \label{Swamy}
\hat S_{Swam} := \frac{1}{N} \sum_{i=1}^N \frac{\|X_i(\hat \beta_i - \hat \beta_{pool})\|^2}{\|y_i-X_i \hat \beta_i\|^2/(T-K)},
\end{equation}
divides by individual variance estimators, thus eliminating the effect of heteroscedasticity. 
Here  $\| \cdot \|$ denotes the euclidean norm, $\hat \beta_i := [X_i' X_i]^{-1} X_i y_i$ the 
ordinary least squares  (OLS) estimator 
of the  individual  parameter $\beta_i$ 
and 
\begin{equation} \label{beta_pool}
\hat \beta_{pool}:= \Big( \sum_{i=1}^N \frac{X_i' X_i}{\|y_i-X_i \hat \beta_i\|^2}\Big)^{-1} \sum_{i=1}^N \frac{X_i' y_i}{\|y_i-X_i \hat \beta_i\|^2}
\end{equation}
the pooled version. 
In \cite{Pesaran} a rigorous asymptotic theory for Swamy's original statistic as well as modified versions for large intersectional scenarios has been developed. In particular, it is  
shown that under $H_0^{exact}$  the statistic $\sqrt{N} (\hat S_{Swam}-K)/\sqrt{2K}$ converges weakly to a standard normal  distribution, 
 if $N,T  \to \infty $ and  $N/T^2 \to 0$. Notice that in this case $N$ can have a larger order than  $T$, allowing 
 for a  wider range of applications than the traditional $F$-test. 
This assumption on the ratio of intersectional and temporal  observations can be further relaxed, if the estimator  $\hat \beta^{pool}$ in \eqref{Swamy}
is replaced by a carefully reweighted version and the division by $\|y_i-X_i \hat \beta_i\|^2$ is replaced by $\|y_i-X_i \hat \beta_{pool}\|^2$. Notice that  under $H_0^{exact}$ the latter yields a more accurate variance estimate. With these modifications,  weak convergence 
under the null hypothesis \eqref{det1} holds if  $N/T^4 \to 0$. However, we emphasize that this improvement rests on perfect homogeneity of all slopes.
 Other approaches to test slope homogeneity, such as \cite{Hausman1978} tests have been considered in the literature (see \cite{Pesaran1996, Phillips}) and we refer the interested reader to  \cite{Pesaran}  for a discussion on the applicability and problems of Hausman tests. \\
One important weakness of the aforementioned tests is their assumption of (conditional) independence of the errors, which is often invalid in applications. An early approach to accommodate intersectional dependence is the seemingly unrelated regression equations (SURE) method, which however is best suited to situations, where $N$ is much smaller than $T$ \citep{Zellner}. Different ways on how to incorporate cross-sectional dependence, such as common factor models, are reviewed in \cite{Chudik}. \cite{Blomquist2013b}  discuss a general slope homogeneity test, where cross-sectional independence is assumed, while temporal dependence is permitted.  These authors assume that $N$ is smaller than $T$, while still $N \to \infty$. 
In a follow-up work, \cite{Blomquist} consider a different approach to deal with temporal and simultaneous cross-sectional dependence. In particular, if  the null hypothesis holds, $N$ is fixed  and $T \to \infty$, they   prove (under  mixing assumptions) 
the weak convergence of Swamy's test statistic
to a weighted sum of chi-squared distributions.
 The resulting  test procedure relies on a block bootstrap to simulate the distribution of $\hat S_{Swam}$ under the null hypothesis. 
  We emphasize  that the theory in  this paper (as the theory for the $F$-test, discussed in the previous paragraph) is only  applicable  to  panels with large time frames compared to cross-sections ($N$ fixed). However, simulations strongly suggest that the test is valid  even if  $N$ and $ T$ are of the same order. \\
While slope homogeneity tests remain a staple of panel analysis, recent works have explored alternative poolability criteria. This trend is fueled by the practical observation that homogeneity tests often reject $H_0^{exact}$, even for inconsiderable deviations from the hypothesis (see \cite{Baltagi2008, Sarafidis2020}). We want to illustrate this problem, by means of a small example.

 \begin{figure}[t]
\begin{subfigure}{0.5\textwidth}
\centering
\includegraphics[width=0.75\linewidth,height=190pt]{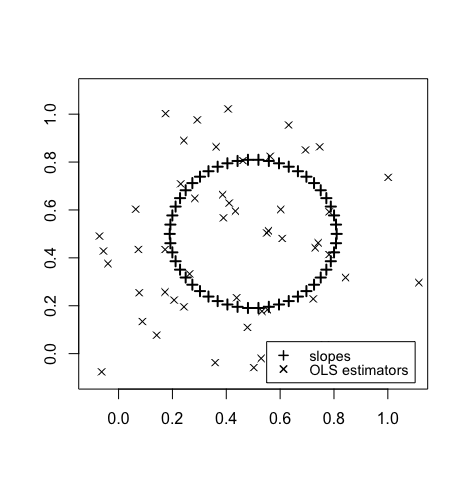}\\[-1ex]
 \end{subfigure}%
 \begin{subfigure}{0.5\textwidth}
 \centering
  \includegraphics[width=0.75\linewidth,height=190pt]{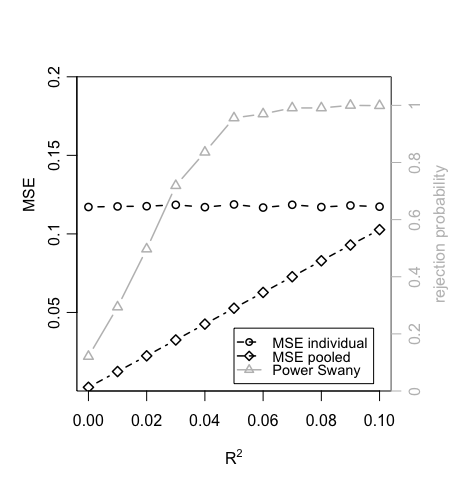}
 \end{subfigure}
 \caption{\doublespacing \label{fig_1} \textit{ Left panel: True slopes (“$+$”) and their estimators (“$\times$”) ($T=20, N=50$). Right panel: Rejection probability of Swamy's test with nominal level $5\%$ (gray), plotted against the individual and pooled MSE (black). The $x$-axis shows the slope heterogeneity, measured by  $R^2$. }}
  \end{figure}
\begin{example} \label{example_1}{\rm
Consider the statistical panel model from \eqref{model_1}, with $T=20$, $N=50$ and $K=2$. The regressors $x_{i,t}$ are generated by  $AR(1)$-processes as described at the very beginning of Section \ref{Section4} and the model errors are i.i.d. standard normally 
distributed random variables, independent of the regressors. The slopes $\beta_1,...,\beta_N \in \mathbb{R}^2$ are equally spaced on the circle with radius $R$ and center $(1/2, 1/2)$. In the left part of Figure \ref{fig_1}, we have plotted the ring of true slopes (“$+$”) and the individual OLS-estimates (“$\times$”) for a radius of $R=0.3$. It is evident that the individual OLS-estimates are subject to extreme fluctuations (a  number of them even have negative entries, which in practice might translate into misguided interpretations). To compare the  performance of pooled and individual analysis, we consider the mean squared error of the estimator $\hat \beta_{pool}$ defined in \eqref{beta_pool} with that of the individual OLS-estimators $\hat \beta_1,...,\hat \beta_N$
$$
MSE_{ind} := \frac{1}{N} \sum_{i=1}^N \E \|\hat \beta_i - \beta_i\|^2 \qquad MSE_{pool} := \frac{1}{N} \sum_{i=1}^N \E \| \hat \beta_{pool} - \beta_i\|^2.
$$
The right part of 
Figure \ref{fig_1}  shows  theses MSEs (approximated by $1000$ simulation runs) as a function of the squared radius
 $R^2$.
We observe 
that $MSE_{ind}$ is constant and not influenced by the degree of slope homogeneity, while $ MSE_{pool} \approx R^2$ (the empirical variance of the slopes). 
 For $R ^2\approx 0.09$ the criterion is indifferent between the two options (even though users would probably  prefer a pooled estimate, for ease of statistical inference).
 We now compare these outcomes with the performance of a classical test for slope homogeneity, say Swamy's test, that we have discussed earlier in this section.   Recall that a user might take rejection by this test as evidence against pooling. The gray line in the right panel of Figure \ref{fig_1} indicates the empirical power of Swamy's test (with nominal level $5\%$), depending on $R^2$. 
As we can see, Swamy's test rejects the hypothesis of slope homogeneity very early, with about $50\%$ probability for $R^2=0.02$, where  $MSE_{pool}/MSE_{ind} \approx 0.2$  and with almost certainty for $R^2 = 0.06$, where still $MSE_{pool}/MSE_{ind} \approx 0.56$. These outcomes indicate, that a decision based on a  test for the 
 hypothesis \eqref{det1} of ``exact'' slope homogeneity 
 can be a rather narrow criterion for poolability.}
  \end{example}

In view of the high power of slope homogeneity tests, it does not come as a surprise that they often reject in practice, where the ideal of perfect homogeneity  is rarely if ever true. This however can be a distraction - as in the previous example - when $H_0^{exact}$ is at least approximately true and pooling is still favorable (see also \cite{Baltagi2008} and \cite{Blomquist}).
Recent works have taken up this insight to formulate more pragmatic criteria for data pooling. One approach is to cluster heterogeneous slopes into (more) homogeneous groups and restrict pooling to these subgroups (see e.g. \cite{Lin, Blomquist, Sarafidis2015, Su2016, Wang}).  Alternatively, it has been suggested to test for “no slope heterogeneity bias”, in order to infer directly whether pooled estimators are appropriate or not \citep{Campello}. 

In the following, we present an alternative approach that combines features of both strands,  slope homogeneity tests, and the novel trend of considering alternative hypotheses than
$H_0^{exact}$. Similar to most traditional works, we believe that a small degree of heterogeneity is a good indicator of poolability, but  we argue that testing for  exact homogeneity can be an analytical straitjacket and often misleading.
  For this reason, we will introduce  the more relaxed notion of  approximate slope homogeneity in the next section.

\subsection{Approximate Slope Homogeneity}\label{sec_22}

We now formulate the hypothesis of \textit{approximate slope homogeneity} for some user determined precision level $\Delta>0$ as follows
\begin{align} \label{hypothesis}
H_0: S_N:=\frac{1}{N} \sum_{i=1}^N \|\beta_i - \bar\beta\|^2 \le \Delta \quad vs. \quad H_1:S_N> \Delta.
\end{align}
Here $\|\cdot\|$ denotes the euclidean norm, $\bar \beta$ is the average over the individual slopes $\beta_1,...,\beta_N$ and   
$S_N$ is the empirical variance of the slopes. Accordingly, if $S_N$ is small,  the individual slopes $\beta_1 , \ldots , \beta_N$ concentrate and $\bar \beta$ is “representative”, while a large value of $S_N$ suggests high diversity such that no real “representative agent” exists. In line with this interpretation, the threshold $\Delta$  expresses how much heterogeneity between  the slopes is still deemed acceptable for the purpose of intersectional pooling. In particular for $\Delta=0$, the hypothesis $H_0$ boils down to the classical version $H_0^{exact}: \beta_1=...=\beta_N$, discussed in the previous section. This corresponds to no tolerance of heterogeneity, regardless how minute - a hypothesis usually justified only for very small $N$. In contrast, a choice $\Delta>0$ means that some degree of heterogeneity is still acceptable for the subsequent analysis, as long as it does not exceed a reasonable level (where the precise choice of $\Delta$ will depend on the problem in hand). 
In this paper we concentrate on the case $\Delta >0$, 
because we argue that in high-dimensional panels it  is rare, and perhaps impossible, to have exact equality of all slopes.



\section{Testing Approximate Slope Homogeneity} \label{Section3}

 The aim of this section is to develop an
 (asymptotically) pivotal  test for the  hypothesis of approximate slope homogeneity defined in  
 \eqref{hypothesis}, when both the sample size $T$ and the panel dimension $N$
 converge to $\infty$. 
 Because this is a demanding task, we proceed in several steps.
 First, in Section \ref{sec31}, we construct an 
 estimator $\hat S_N (\kappa) $ of the heterogeneity measure $S_N$ in \eqref{hypothesis}, which is \textit{ time-sequential }
 in the sense that it is calculated from the data
 $\{ (x_{i,t},y_{i,t}) |  t=1, \ldots  ,\lfloor T \kappa \rfloor , \,\, i=1,\ldots ,N \}$ for $\kappa \in (0,1]$.
 In particular, $\hat S_N(1) $ estimates $S_N$ using all available data, while a consideration of the whole process $\{ \hat S_N (\kappa) \}_\kappa$ allows us to subsequently construct a pivotal test. Second, in Section \ref{sec32}, we collect and discuss the required  assumptions for our theoretical results (this section could be skipped in a first reading). 
 In Section  \ref{sec33} we develop a 
 weak invariance principle for an adequately standardized version of  $\{ \hat S_N (\kappa) \}_\kappa$. 
Fourth, these results are used to define a  self-normalized test-statistic in Section \ref{sec34} and 
to prove its  weak convergence to a pivotal random variable.
As a  result, we obtain a simple (asymptotic) level $\alpha$ test  for the hypothesis \eqref{hypothesis} that is uniformly consistent against classes of local alternatives (see Theorem \ref{theorem_test}). Finally, in Section \ref{sec35}, we discuss an extension of our approach to a different dispersion measure than $S_N$, that assesses slope heterogeneity w.r.t. to prediction. 

\subsection{Estimators} \label{sec31}

Recall that $\hat \beta_i := [X_i' X_i]^{-1} X_i y_i$ is the 
OLS estimator of the individual  parameter $\beta_i$ and define the average
$\hat  \beta := \frac{1}{N} \sum_{i=1}^{N} \hat \beta_i$. A natural 
approximation of $S_N$ is then given by the empirical variance of the estimates $\hat \beta_1,...,\hat \beta_N$, defined as 
\begin{equation*}
    \label{det2} 
    \hat{S}_N := \frac{1}{ N  } \sum_{i=1}^{N} \|\hat \beta_i  - \hat  \beta \|^2.
\end{equation*}
In principle, statistical inference can be based directly on $\hat S_N$. The results presented in Section \ref{sec33} below imply for fixed  $N$ and $T\to \infty $, that the statistic $\sqrt{T} ( \hat S_N  - S_N)$ is asymptotically normal. If we permit $N\to \infty $, an additional bias term, say $B_N$, appears and 
the bias-corrected statistic $\sqrt{NT} ( \hat S_N  - S_N - B_N)$ is asymptotically normal. Theoretically, this suffices to construct a test for the hypothesis \eqref{hypothesis}: Let
$\lambda^2 $ denote 
 the variance of the limiting distribution and
$\hat B_N$ be an estimator of the bias
$B_N$ such that  $\sqrt{NT}  ( \hat B_N   -B_N) =o_P(1)$. 
 Then rejecting the null hypothesis \eqref{hypothesis}, whenever
$  \hat S_N  >  \Delta  +\hat B_N +  \lambda u_{1-\alpha }/  \sqrt{NT} $ yields a 
consistent and asymptotic level $\alpha $ test, where  $u_{1-\alpha}$ is the $(1-\alpha)$
 quantile of the normal distribution. However, a practical implementation requires consistent estimation of the variance  $\lambda^2$.
At this point we do not give  details  on the form of $\lambda^2$, which will be defined rigorously in \eqref{long_run_variance}. We just remark that this 
 long-run variance depends in a  very complicated way 
 on the intersectional and temporal dependence 
 of the data, which makes estimation nearly impossible.

As an alternative, we propose a self-normalization approach, which - roughly speaking - consists in dividing $\hat S_N$
by a further statistic, say $\hat V_N$, such that 
$\sqrt{NT }\hat V_N $ converges weakly to a random variable
proportional to
 $\lambda$. As a consequence, the factor $\lambda$ cancels out in the ratio $T_N= (\hat S_N-\hat B_N -S_N )/ \hat V_N $, making it asymptotically
 distribution free. 
 Self-normalization is an important tool in the analysis of dependent time series, due to its user-friendliness and robustness (see \cite{Mueller2007}).
 Early references include  
 \cite{lobato2001}, \cite{shao2010}, and \cite{shazha2010} (for an overview see \cite{Shao2015}). More recently, self-normalization  has also been used for the analysis of high-dimensional time series (see \cite{wangshao2020,wang2021inference} among others).
 A common feature of all these references 
 is that the proposed methodology is only applicable for testing hypotheses of 
 exact equality (such as formulated 
in  \eqref{det1} in the case of comparing the slopes).  Self-normalization in the context of 
testing hypotheses of approximate equality (as considered in \eqref{hypothesis}) has been introduced in \cite{DetteKokotVolgushev}, but not in a high-dimensional scenario.

 A key feature of self-normalizations is the use of \textit{time-sequential estimators}, which  do not draw on all temporal observations but only on a subset,  specified by a parameter $\kappa$. To be precise, 
 let $p \in (0,1)$ denote a fixed constant and $I =[p,1]$ an indexing set. Then for $\kappa \in I$ 
we consider the observations
 \begin{equation}\label{temp_model}
\{  (x_{i,t},y_{i,t}) |~ t=1, \ldots ,\lfloor T \kappa \rfloor , \,\, i=1,\ldots ,N \} ~.
\end{equation}
 In order to define the estimator $\hat S_N (\kappa )$
 from the data \eqref{temp_model}, we begin with 
 the corresponding slope estimators and consider for $\kappa \in I$ the restricted design matrix 
 \begin{equation} \label{X-matrix}
 X_i(\kappa):=[x_{i,1},...,x_{i, \lfloor \kappa T \rfloor, },\mathbf{0},...,\mathbf{0}]' \in \R^{T \times K}  
 \end{equation}
 of the $i$th individual. Note that $ X_i(\kappa)$
 is obtained from the  design  matrix $X_i$ 
 in model \eqref{model_2} by keeping the first $\lfloor \kappa T \rfloor$ rows of $X_i$ and then replacing 
 the remaining $T- \lfloor \kappa T \rfloor$
subsequent rows by  $K$-dimensional vectors with $0$-entries. Similarly, we define the vector
\begin{equation}
    \label{det3}
 y_i(\kappa):=(y_{i,1},...,y_{i,\lfloor \kappa T \rfloor}, 0,...0)'\in \R^{T}~.
\end{equation}
Notice that the additional $ T - \lfloor T\kappa \rfloor   $
zero  rows in \eqref{X-matrix} and  zero entries \eqref{det3} do not serve any practical purpose. They will, however, facilitate a tidy theory, where all matrices and vectors are of the same dimensions,  regardless of $\kappa \in I $.

The sequential OLS estimator of the parameter $\beta_i$ 
in model \eqref{model_1} based on the  data  \eqref{temp_model} can 
then be expressed as
$$
\hat \beta_i(\kappa):= [X_i(\kappa)'X_i(\kappa)]^{-1} X_i(\kappa)' y_i(\kappa) \in \R^K ~.
$$
Note that this definition requires $\lfloor \kappa T \rfloor  > K $ for all $\kappa \in I$, where again $K$ is the dimension of the slope parameter. Therefore, 
we assume throughout this paper that $pT>K$, which of course is asymptotically true for any  $p>0$ as $T \to \infty$.

The corresponding \textit{mean group estimator} for $\bar \beta:= \frac{1}{N} \sum_{i=1}^{N} \beta_i$ is obtained by 
$$
\hat \beta(\kappa):= \frac{1}{N}\sum_{i=1}^{N} \hat \beta_i(\kappa) ~ ,
$$
and  the time-sequential estimator of the slope heterogeneity $S_N$ is defined by 
$$
\hat{S}_N( \kappa):= \frac{1}{ N  } \sum_{i=1}^{N} \|\hat \beta_i( \kappa) - \hat \beta(\kappa)\|^2.
$$
Notice that  this estimator changes with $\kappa$ and thereby tracks information about the temporal structure of the panel. %
Finally,  we point out that  for  $\kappa =1$
we obtain the estimators  $\hat \beta_i  = \hat \beta_i (1) $
($i=1, \ldots , N)$ and  $\hat S_N = \hat S_N (1)$
from the full sample considered in Section \ref{sec_21}
and in Section \ref{sec31}, respectively.

\subsection{Assumptions}\label{sec32}

We now the present the necessary assumptions for our approach. This section could be skipped during a first reading.

In order to state the dependency of our data, we define the notion of $\alpha$-mixing for multidimensional arrays of random variables (see for a more general definition \cite{Doukhan}).

\begin{defi} \label{def_mixing}
Let $\mathcal{M}\subset \Z^d$ for some $d \in \N$ be endowed with the maximum norm $\|\cdot \|_\infty$ and
 $(\Xi_{z})_{z\in \mathcal{M}  }$ be an array of random variables indexed in $\mathcal{M}$.
 For index sets $\mathcal{I}$ and $\mathcal{J}$ we define the distance w.r.t. the maximum norm as
 $$
 \mbox{dist}(\mathcal{I}, \mathcal{J}) = \min(\|i-j\|_\infty: i \in \mathcal{I}, j \in \mathcal{J})
 $$
 and use the notations $\mathcal{F}_\mathcal{I} =\sigma\big(\Xi_{z}: z \in \mathcal{I}\big)$ for the sigma field generated by the random variables $\{ \Xi_{z}: z \in \mathcal{I}\} $.
For $r \in \N_0$ the $r$th $\alpha$-mixing coefficient  is
defined by 
\begin{align*}
\alpha(r)=\sup &\Big\{ |\mathbb{P}(A \cap B)- \mathbb{P}(A) \mathbb{P}(B)|:  A \in \mathcal{F}_\mathcal{I}, ~
B \in \mathcal{F}_\mathcal{J},~
dist(\mathcal{I}, \mathcal{J})\ge r,~ \mathcal{I}, \mathcal{J} \subset \mathcal{M}\Big\}.
\end{align*}

The  array is called $\alpha$-mixing if $\alpha(r) \to 0$ as  $r \to \infty$. 
\end{defi}

In the following, $C$ denotes a generic constant, independent of $N, T$ or any particular individual or time point. Furthermore,  $ x\lor y $ denotes the maximum of two real numbers $x$ and $y$.

\begin{assumption} \label{assumption1} ~~
{\rm
\begin{itemize}
\item[$(N)$] \quad\quad\,\,\, (\textit{Intersect-time relation}) There exists a constant $\eta$, such that   $N/T^\eta \to 0$ and
\begin{equation} \label{Eq_def_eta}
0 \le \eta < 2 - \frac{2}{1+a(M-1)/M}.
\end{equation}
 Here $M$ and $a$ are defined in the assumptions below.
\item[($\beta$)]
 \quad\quad \,\,\,(\textit{Bounded slopes}) \quad$
	\exists C>0: 
	\max_{i} \| \beta_i \| \le C.
$
\item[($\varepsilon$)]
(1) 
\quad (\textit{Exogeneity}) \quad$ \mathbb{E}[\varepsilon_i | \mathbf{X}]=0$ for all $i$.
\\ (2) 
\quad
(\textit{Separable covariance}) \quad For all $i,j$ it holds that 
$
\mathbb{E}[\varepsilon_i \varepsilon_j' | \mathbf{X}]=\sigma_{i,j}\Sigma_T, 
$
where $\Sigma_T \in \mathbb{R}^{T \times T}$ is a temporal covariance matrix and $\sigma_{i,j}$ a spatial factor, with $\sigma_{i,i}>0\, \forall i$. 
\\ (3) 
\quad (\textit{Conditional moments})  \quad
For some $M >2 \,\exists C>0: \max_{i,t }\mathbb{E}[|\varepsilon_{i,t}|^{2M}|\mathbf{X}]\le C.$
\\
(4) \quad (\textit{Stationarity}) 
\quad
Conditionally on the data $\mathbf{X}$, the errors $(\varepsilon_{i,t})_{t=1,...,T} $ are weakly time stationary, with autocovariance function $
\tau( k):=\mathbb{E}[\varepsilon_{i,t}\varepsilon_{i,t+k}|\mathbf{X}]/\sigma_{i,i}$ for $k \in \mathbb{Z}.
$
(5) 
\quad 
(\textit{Mixing})  The array $(x_{i,t}, \varepsilon_{i,t})_{i,t}$ is $\alpha$-mixing in the sense of Definition \ref{def_mixing}, with
$\alpha(r)\le C r^{-a} $
for some 
\begin{equation} \label{hd14}
a>\frac{4(M+2 \lor 2 \eta)}{M-2 \lor 2 \eta}
\end{equation}
with $M$ and $\eta$ as defined above.  
\item[($X$)]
(1)
\quad (\textit{Regressor convergence})  \quad 
There exist positive-definite, non-stochastic matrices $Q_i$, $i=1,...,N$ and a positive constant $C>0$, such that
$$
\max_i  \|\E [X_i'X_i/T]-Q_i \| \le C/\sqrt{T}.
$$
There exist   non-stochastic matrices $U_{i,j}$, $i,j=1,...,N$, such that 
$$
\max_{i,j}  \|\E [X_i' \Sigma_T X_j/T]-U_{i,j} \| =o(1).
$$
(2)
\quad (\textit{Moments}) \quad
There exists a constant  $ C>0$ such that  $\max_{i,t}\mathbb{E}|x_{i,t}|^{2M} \le C$.  
\\ (3)
\quad (\textit{Bounded limits}) 
\quad There exists a constant $ C>0$ such that $  \max_i \|Q_i^{-1}\| \le C $.
\end{itemize}
\item[$(\LR)$] $\quad \,\,\,\,  $ (\textit{Existence of the
long-run variance (LRV)}) The limit
\begin{equation} \label{long_run_variance}
    \LR^2 := \lim_{N \to \infty} \frac{4 }{N} \sum_{i,j=1}^N \sigma_{i,j}[\beta_i-\bar \beta]' Q_i^{-1} U_{i,j} Q_j^{-1} [\beta_j-\bar \beta]
\end{equation}
$\quad \quad \quad \,$  
exists and is strictly positive.
}
\end{assumption}

\noindent
\begin{rem} 
{\rm  We briefly comment on these assumptions:
\begin{itemize} 
\item[$(N)$] With this assumption, we quantify the relation between $N$ and $T$ and ensure that the number of individuals is not too large compared to the temporal observations. $T$ can be smaller compared to $N$ (larger $\eta$) if more moments exist (larger $M$) and if dependence is weak (larger $a$).  Substituting $a$ in \eqref{Eq_def_eta} by its lower bound in \eqref{hd14} shows that a choice $\eta>1$ is always possible. In the case of 
independent data or
exponentially decaying mixing coefficients, we can choose the constant $a$ arbitrarily large and condition \eqref{hd14} boils down to $\eta \in [0,2)$. In particular, our condition $N/T^\eta \to 0$ is then arbitrarily close to the common assumption 
$N/T^2 \to 0$ in Swamy's test for i.i.d. data \citep{Pesaran}.
\item[$(\beta)$] The assumption of bounded slopes is standard in the literature (see, for example, \cite{Pesaran}). It implies that a violation of the hypothesis is asymptotically not caused by just a few outliers in the slopes.
\item[$(\varepsilon)$] Assumption $(1)$ states strong exogeneity of the residuals and is common in the analysis of panel data (see \cite{Pesaran} or \cite{Campello}). A close examination of our proofs in the supplement demonstrates that this assumption can be relaxed to uncorrelatedness of low degree rational functions of the regressors and residuals  (see e.g. Proposition \ref{prop_all_moments}). However, to derive weak invariance principles it seems necessary to consider a stronger exogeneity condition than merely \textit{weak exogeneity} as for instance used in \cite{Blomquist} (e.g. to guarantee that the term $E_1$ defined in Proposition \ref{theorem1} is centered). Assumption $(2)$
permits the existence of complex error structures. 
 The covariance matrix of the  errors $(\varepsilon_1^\prime, \ldots ,\varepsilon_N^\prime)^\prime$  is the Kronecker product $\Sigma_N \otimes \Sigma_T$, where the first factor captures intersectional and the second factor temporal dependence.  This condition is much weaker than most assumptions made 
in the literature, which  are special cases of this setting.  For example, the traditional assumption of homoscedasticity and independence, as needed for the $F$-test is captured by $\sigma^2 \cdot I_N \otimes I_T$, where $\sigma^2>0$ is the variance and  $I_N, I_T$ are the identity matrices of dimension $N$ and $T$, respectively. Independent heteroscedastic errors  (see \cite{Swamy}, \cite{Pesaran}) are similarly captured by $D \otimes I_T$, where $D=diag(\sigma_1^2,...,\sigma_N^2)$ represents the  heteroscedasticity-matrix. The case of intersectional dependence only, as represented in Zellner's SURE approach, is incorporated by $\Sigma_N \otimes I_T$ \citep{Zellner} (this is also closely related to \cite{ando2015}). Allowing both factors to differ from the identity matrix (simultaneous temporal and intersectional dependence) obviously encompasses much richer models than have been treated before, particularly in high dimensional panels. 
In condition $(3)$ we assume existence of some moments, but we do not need Gaussian errors, as common in traditional models. In the time component, we assume conditional (weak) stationarity, which in a well specified model is sensible. In contrast, no specific covariance structure is assumed intersectionally, where it is not easy to justify. Obviously, 
 conditions on the strength of the 
dependence are required, which is done by $(5)$.  The exponent  $a$ 
 in the polynomial decay 
 of the mixing coefficients  can be  smaller (stronger dependence) if $M$ is larger (stronger moment assumption) and to some degree, if $T$ is larger compared to $N$ (smaller $\eta$). 
Notice that, in principle, the mixing condition in $(5)$  implies a one dimensional ordering of the intersectional data. However, this 
is only made for the sake of a transparent presentation.
Our theory is not confined to one dimensional orderings  and can be adapted to multidimensional intersects by considering $\alpha$-mixing for higher dimensions (recall Definition \ref{def_mixing}). 
Importantly, the test statistics defined below  do not have to be changed in this case and all subsequent results remain correct.
\item[$(X)$] Assumption $(1)$ (in a weaker form) is widely prevalent in the literature (see among others \cite{Pesaran, Blomquist, breitung2016, Campello}), where it is assumed that each individual product $X_i'X_i/T$ has a positive definite limit. We need a convergence rate specifically in the proofs of our weak invariance principles to control certain bias terms. It should be noted that such conditions (even of linear bias) are common in statistical theory. The assumption on the cross term $X_j' \Sigma_T X_i$ is new and can be dropped if - as in most traditional works - the errors are intersectionally uncorrelated ($\sigma_{i,j}=0$ if $i \neq j$). Otherwise, it is important for the convergence of the long-run variance of the statistic (see also Assumption $(\LR)$). Condition $(3)$ ensures that asymptotically (as $N \to \infty$) it is not arbitrarily difficult for an OLS estimator $\hat \beta_i$ to recover the true slope $\beta_i$. 
This assumption is not necessary for traditional test statistics, such as $\hat S_{Swam}$ \eqref{Swamy}, because these heterogeneity measures consider the  differences $\|X_i(\hat \beta_i - \hat \beta_{pool})\|^2$ instead of the raw comparison $\|\hat \beta_i - \hat \beta\|^2$.
\item[$(\LR)$] 
The assumption of a positive  long-run variance  $\LR^2$  is standard in time series analysis and has also been used for slope homogeneity tests in the presence of time dependence (see e.g. Assumption  ERR. iii) in \cite{Blomquist}). In classical scenarios of  temporal and cross-sectional independence,  it  follows from regularity assumptions on the matrices $Q_i$ (see, for example, Assumptions 1, (i) and 2 in \cite{Pesaran}).
\end{itemize}
}
\end{rem}

\subsection{Weak Convergence of the sequential process}\label{sec33} 

In this section, we derive a weak invariance principle for an adequately standardized version of the difference $\hat{S}_N(\kappa)-S_N$ (see \eqref{Eq_weak_convergence}). We prepare this result by three technical propositions dealing with stochastic linearization, asymptotic de-biasing and finally weak convergence. 
We begin by decomposing 
   (for a proof see Section \ref{seca1} in the supplement)
\begin{align} \label{linearization} 
     & \kappa \sqrt{ N   T  }(\hat S_N(\kappa)- S_N) = E_1(\kappa)+E_2(\kappa)+E_3(\kappa), \end{align}
   where
     \begin{align}
     \label{e1}
E_1(\kappa):=&\kappa\sqrt{\frac{T}{N}}\sum_{i=1}^{N}  2\varepsilon_{i}(\kappa)' X_i(\kappa) \big[ X_i(\kappa)'X_i(\kappa)\big]^{-1} [\beta_i-\bar \beta]~, 
\\
     \label{e2}
E_2(\kappa):=&\kappa  \sqrt{\frac{T}{N}} \sum_{i=1}^{ N  }  \varepsilon_i(\kappa)' X_i(\kappa) \big[ X_i(\kappa)'X_i(\kappa)\big]^{-2} X_i(\kappa)' \varepsilon_i(\kappa)~,\\
     \label{e3}
E_3(\kappa):=& -\kappa  \sqrt{\frac{T}{N}} \frac{1}{N} \sum_{i,j=1}^{N} \varepsilon_{i}'(\kappa)X_i(\kappa) \big[ X_i(\kappa)'X_i(\kappa)\big]^{-1} \big[ X_j(\kappa)'X_j(\kappa)\big]^{-1} X_j(\kappa)' \varepsilon_{j}(\kappa). 
\end{align}
In the above expressions $X_i(\kappa)$ is the regressor matrix given in  \eqref{X-matrix} and the  sequential error vector $\varepsilon_i(\kappa)$ is defined as
$
\varepsilon_i(\kappa):=( \varepsilon_{i,1},..., \varepsilon_{i,\lfloor \kappa T \rfloor},0,...,0)' \in \R^T.
$
In the following we show that, after subtracting a bias, $E_1$ is the leading term of
the decomposition \eqref{linearization}.

\begin{prop} \label{theorem1}
Under the Assumptions $(N), (\beta), (\varepsilon), (X)$ it follows  that 
\begin{align} \label{h1} 
& \sup_{\kappa \in I}| E_2(\kappa)-\E[E_2(\kappa)|\mathbf{X}] | =o_P(1)~, \\
& \sup_{\kappa \in I}| E_3(\kappa) | =o_P(1) ~.
\label{h1a} 
\end{align}
\end{prop}

The three terms in the decomposition \eqref{linearization} can be understood as follows:  $E_1$ is a sum of centered random variables and if dependence is not too strong, can be shown to be asymptotically normal. $E_2$ is not centered and constitutes the main source of bias in our statistic. It corresponds to the leading term in classical slope homogeneity tests (see e.g. \cite{Pesaran}). However, if centered $E_2$ is asymptotically negligible, which motivates the bias correction in the next step. Finally, the term $E_3$ converges to the squared norm of a centered process divided by $\sqrt{NT}$ and is thus negligible. 

We begin with  the problem of centering $E_2$. In Proposition \ref{theorem1}, we have seen that $E_2-\E[E_2|\mathbf{X}]$ is of order $o_P(1)$. Furthermore, it can be shown that 
$
\E[E_2|\mathbf{X}]=\mathcal{O}(\sqrt{N/T}),
$
which if $\eta<1$ implies $\E[E_2|\mathbf{X}]=o_P(1)$. In the case $\eta \ge 1$ however $\E[E_2|\mathbf{X}]$ is non-negligible and hence a bias correction for
the statistic  $\kappa \sqrt{ N   T  }(\hat S_N(\kappa)- S_N)$ is necessary. Notice that even for $\eta<1$ such a correction may be advisable if $T$ is not large compared to $N$ (which is rarely the case in economic panels). 
Since the true bias $\E[E_2|\mathbf{X}]$ is unknown, we have to estimate it.  
For this purpose, we denote for  $1\leq A \leq  B \leq T$  
by  
$v_{A:B}=(v_A,v_{A+1},...,v_B)'$ the sub-vector 
of $v=(v_1, v_2,...,v_T)'$. Therewith we 
define the $T \times T$ matrix
\begin{equation} \label{def_hat_Sigma_i}
\hat \Sigma_i(\kappa, b) := (\hat \xi_i(|s-t|, \kappa) \mathbbm{1}\{|s-t| < b\})_{1\le s,t\le T },
\end{equation}
where  the entry
\begin{eqnarray} \label{def_hat_xi}
\hat \xi_i(h, \kappa) &:=& \frac{(y_i(\kappa)-X_i(\kappa)'\hat \beta_i )_{1:\lfloor T\kappa\rfloor-h}' (y_i(\kappa)-X_i(\kappa)'\hat \beta_i )_{h+1:\lfloor T\kappa\rfloor}}{\lfloor T\kappa\rfloor -h-K} \\
\nonumber &=& {1 \over \lfloor T\kappa\rfloor -h-K} \sum_{t=1}^{\lfloor T\kappa\rfloor -h} (y_i(\kappa)-X_i(\kappa)'\hat \beta_i )_t (y_i(\kappa)-X_i(\kappa)'\hat \beta_i )_{t+h}.
\end{eqnarray}
is  the sequential  estimator  of the autocovariance of lag $h$.
Here  $b$ is a regularization  parameter (all $k$-diagonals with $k>b$ are set equal to $0$).
Note that banded estimates (with an increasing size of the band)
of high-dimensional autocovariance matrices are quite common in this context
and motivated by the fact that the autocovariances decay, when moving away from the diagonal (see, for example, \cite{wu2009banding,Mcmurpol2010} among others). We also point out that $b$ has to be chosen such that the denominator is non-degenerate for all $h<b$ (which is asymptotically always true given our subsequent assumptions on $b$).   
 We can now define the estimate for $\E[E_2|\mathbf{X}]$ by  
\begin{equation} \label{bias_estimate}
\hat B_N(\kappa):=\kappa \sqrt{\frac{T}{N}}\sum_{i=1}^N\Tr[\hat \Sigma_i(\kappa, b) X_i(\kappa) \big[ X_i(\kappa)'X_i(\kappa)\big]^{-2} X_i(\kappa)' ],
\end{equation}
and the following result shows its  (uniform) consistency.

\begin{prop} \label{theorem4} 
Assume that the conditions $(N), (\beta), (\varepsilon), (X)$ hold, and that the parameter $b$  in \eqref{def_hat_Sigma_i}  satisfies
 $b=\mathcal{O}(T^{\gamma})$ with 
$$
\gamma \in \Big(\frac{\eta-1}{2(a\frac{M-1}{M}-1)} ,  \frac{1-\eta/2}{2}\Big).\\[2ex]
$$
Then $\hat B_N(\kappa)-\E[E_2(\kappa)|\mathbf{X}]=o_P(1)$ uniformly in $\kappa \in I $.
\end{prop}

Notice that the interval for $\gamma$ is well-defined, if  Assumption $(N)$ holds.
Finally, we consider the leading term $E_1$ in Proposition \ref{theorem1} and show its convergence to a Gaussian process. 

\begin{prop} \label{theorem5}
Suppose that the Assumptions $(N), (\beta), (\varepsilon), (X)$ and $(\lambda)$ hold.
Then the weak convergence
$
\{E_1(\kappa)\}_{\kappa \in I} \to\LR \{\mathbb{B}(\kappa)\}_{\kappa \in I} 
$
holds, where $\mathbb{B}$ is a standard Brownian motion and $\lambda^2$ is defined in \eqref{long_run_variance}.
\end{prop}

We can now combine Propositions \ref{theorem1} - \ref{theorem5}, to 
obtain a weak invariance principle for the 
the bias corrected version 
\begin{equation} \label{bias_correction}
\tilde{S}_N(\kappa):= \hat S_N(\kappa)-\hat B_N (\kappa),
\end{equation} 
of $\hat S_N(\kappa)$ 
(here $\hat B_N$ is  defined in \eqref{bias_estimate}). To be precise, it follows under the assumptions of Proposition \ref{theorem5}, that
\begin{equation} \label{Eq_weak_convergence}
\{ \kappa \sqrt{ N   T  }(\tilde S_N(\kappa)- S_N) \}_{\kappa \in I} \to  \{\LR \mathbb{B}(\kappa)\}_{\kappa \in I},
\end{equation}
where $\lambda^2$ is the long-run variance defined in \eqref{long_run_variance}.
 In the next 
  section, we use  this result  to 
  construct a pivotal test
  statistic for the hypothesis \eqref{hypothesis}.

\subsection{A  pivotal Test for Approximate Slope Homogeneity} \label{sec34} 

We now present the statistical consequences of the results derived in the previous section. In particular,  we construct a test  statistic, that asymptotically cancels  the long-run variance
$\LR^2$ out. For this purpose, we employ the principle of self-normalization and adapt techniques introduced in  \cite{DetteKokotVolgushev} to the high dimensional scenario. We begin by defining the self-normalizing factor 
\begin{equation}
\label{def_denominator}
    \hat V_N := \Big\{\int_I\kappa^4 (\tilde S_N(\kappa)- \tilde S_N)^2d \nu (\kappa) \Big\}^{1/2},
\end{equation}
 where $I=[p,1]$ is the indexing interval and $\nu $ some probability measure on $I$, which does not concentrate at the point $1$ (i.e $\nu (\{1\}) < 1 $).
  $\nu$ is user determined, and typical choices include the uniform measure on a finite number of points, such as $\{1/K, 2/K,...,(K-1)/K, 1\} \cap I$. The size of $K$ has little influence on the performance of the subsequent statistic, but small choices evidently yield computational advantages (as $K$ determines the number of evaluations of $\tilde S_N(\kappa)$). In particular, a small $K$ makes self-normalization computationally much more parsimonious than bootstrap alternatives.\\
With $\hat V_N$ as denominator, we now define the self-normalized statistic 
\begin{align} \label{hat_W}
\hat W_N:= & \frac{\tilde S_N- \Delta}{\hat V_N}.
\end{align}
By inspection of the numerator, we see that large (positive) values of $\hat W_N $ suggest that $S_N >\Delta$, i.e. that the hypothesis \eqref{hypothesis} does not hold,  whereas small values suggest that $S_N \le \Delta$. To understand the asymptotic
properties of $\hat W_N$ more properly, we note that
the weak convergence in \eqref{Eq_weak_convergence}, together with the continuous mapping theorem, implies 
 \begin{equation} \label{W}
 \frac{\tilde S_N- S_N}{\hat V_N}
  \to W := \frac{\mathbb{B}(1)}{\big\{ \int_I \kappa^2 (\mathbb{B}(\kappa)-\kappa\mathbb{B}(1))^2 d\nu(\kappa) \big\}^{1/2}},
\end{equation}
where we have used that the left-hand side of \eqref{W} is a continuous transformation of  the process $\{\kappa \sqrt{ N   T  }(\tilde S_N(\kappa)- S_N)\}_{\kappa \in I}$. Note that the distribution of $W$ is non-normal, but does not contain any nuisance parameters. In particular, its quantiles are truly pivotal, easy to simulate and only depend on the user-determined measure $\nu$. We now consider the decomposition
\begin{equation} \label{decomposition}
\hat W_N = \frac{\tilde S_N- S_N}{\hat V_N} +\frac{S_N- \Delta}{\hat V_N}.
\end{equation}
The first term on the right converges weakly to $W$ (according to \eqref{W}), while the second one is non-positive under the hypothesis ($S_N \le \Delta$) and positive under the alternative ($S_N >\Delta$). This motivates the test decision, to reject the hypothesis \eqref{hypothesis}, if 
\begin{equation} \label{test_decision}
\hat W_N >q_{1-\alpha},
\end{equation}
where $q_{1-\alpha}$ is the upper $\alpha$ quantile of 
the distribution of $W$. Combining \eqref{W} and  \eqref{decomposition} implies that 
$
\limsup_{N \to \infty}  \mathbb{P}(\hat W_N >q_{1-\alpha}) \le \alpha
$
if $S_N \le \Delta$ for all $N$ 
(level $\alpha$) and also
$
\lim_{N \to \infty}  \mathbb{P}(\hat W_N >q_{1-\alpha}) =1
$ for fixed $\epsilon>0$ 
if $S_N >\Delta+\epsilon$ for all $N$ (consistency). 
It is however possible to sharpen these results substantially by considering uniform level $\alpha$ and consistency against classes of local alternatives. We pursue this generalization in the remainder of this section.

Let us define for two fixed but arbitrary constants $C, c>0$   the class
\begin{equation} \label{Eq_def_class_T}
    \mathcal{T}:= \big\{\bsec =(\beta_n)_{n \in \N}: \sup_n \|\beta_n \| \le C, \inf_N\lambda_N^2(\bsec )>c \big\},
\end{equation}
which consists of all sequences of bounded slopes, for which the variance
\begin{equation} \label{ Eq_def_lambdauni}
\LR^2_N(\bsec ) :=  \frac{4 }{N} \sum_{i,j=1}^N \sigma_{i,j}[\beta_i-\bar \beta]' Q_i^{-1} U_{i,j} Q_j^{-1} [\beta_j-\bar \beta]
\end{equation}
is bounded away from $0$ (these conditions are uniform versions of Assumptions $(\beta)$ and $(\lambda)$). Denoting $S_N$ by $S_N(\bsec )$ (see \eqref{hypothesis}) to make dependence on the slopes explicit, we define the sets of hypotheses and local alternatives as
$$
\mathcal{H}:= \mathcal{T} \cap \{\bsec : S_N(\bsec )\le \Delta \} \quad \textnormal{and} \quad \mathcal{A}_N(x):= \mathcal{T} \cap \{\bsec : S_N(\bsec )-x/\sqrt{NT}\ge\Delta \}
$$
respectively. Notice that in the definition of $\mathcal{A}_N(x)$ the number $x>0$ quantifies the distance to the hypothesis. We can then state the following result.

 \begin{theo} \label{theorem_test}
Suppose that the Assumptions $(N), (\varepsilon)$ and $(X)$  hold and that the bandwidth $b$ is chosen according to Proposition \ref{theorem4}. Then for any $\Delta>0$ and $\alpha \in (0,1)$ it holds that
 $$
 \limsup_{N \to \infty} \sup_{\bsec  \in \mathcal{H}}\mathbb{P}\big(\hat W_N >q_{1-\alpha}\big)=\alpha
 \quad \quad (asymptotic \,\, level \,\, \alpha).
 $$
 Furthermore there exists a non-decreasing function $f: \mathbb{R}_{>0} \to \mathbb{R}_{>0}$, with $f(x)> \alpha$ for all $x >0$ and $\lim_{x \to \infty}f(x)=1$, such that
 $$
 \liminf_{N \to \infty} \inf_{\bsec  \in \mathcal{A}_N(x)}\mathbb{P}\big(\hat W_N >q_{1-\alpha}\big) =f(x)\quad \quad (asymptotic \,\, consistency).
 $$
 \end{theo}

Note that Theorem \ref{theorem_test} is a stronger result than those usually derived in the related literature, as it proves consistency of the test decision \eqref{test_decision} against a whole class of local alternatives. In contrast, the works cited in Section \ref{sec_22} consider either fixed alternatives or fixed local alternatives (such as \cite{Pesaran} or \cite{breitung2016}). We conclude this section with a small remark regarding applications.

\begin{rem} \label{rem_intercepts}
{\rm
 For parsimony of presentation, we have throughout this work assumed that the model \eqref{model_2} does not include individual specific intercepts. However, in most applications  a model with intercepts  is more realistic, that is 
 \begin{equation} \label{model_3}
y_{i} =\alpha_i+X_{i} \beta_i+\varepsilon_{i} \quad \quad  i=1,...,N
\end{equation}
where $\alpha_1,...,\alpha_N \in \R$. In this situation, we can define the matrix $M(\kappa) \in \R^{T \times T}$, which is given for $s, t \in \{1,...,T\}$ as $(M(\kappa))_{s,t}=1-1/\lfloor T \kappa \rfloor$ if $t=s$ and $(M(\kappa))_{s,t}=-1/\lfloor T \kappa \rfloor$ else. 
We can then 
rewrite the least squares  estimator of $\beta_i$ in model \eqref{model_3} as
$$
\hat \beta_i(\kappa):= [X_i(\kappa)'M(\kappa) X_i(\kappa)]^{-1} X_i(\kappa)'M(\kappa) y_i(\kappa).
$$
If we now re-define the bias-correction by
$$
\hat B_N(\kappa):=\kappa \sqrt{\frac{T}{N}}\sum_{i=1}^N\Tr[\hat \Sigma_i(\kappa, b) M(\kappa)X_i(\kappa) \big[ X_i(\kappa)'M(\kappa) X_i(\kappa)\big]^{-2} X_i(\kappa)'M(\kappa) ],
$$
and  the entries of $\hat \Sigma_i(\kappa, b)$  by
$$
\hat \xi_i(h, \kappa) = \frac{\big(M(\kappa)(y_i(\kappa)-X_i(\kappa)'\hat \beta_i )\big)_{1:\lfloor T\kappa\rfloor-h}' \big(M(\kappa)(y_i(\kappa)-X_i(\kappa)'\hat \beta_i )\big)_{h+1:\lfloor T\kappa\rfloor}}{\lfloor T\kappa\rfloor -h-K},
$$
then all results presented in this section remain correct for the model \eqref{model_3}. }

\end{rem}
 
\subsection{A Prediction Measure for Slope Homogeneity} \label{sec35}

So far we have considered $S_N$, the empirical variance of the slopes $\beta_1,...,\beta_N$, to assess slope homogeneity. However, other dispersion measures may be of equal interest. For example, instead of comparing $\beta_i$ directly to some $\beta$ in the norm, it can be meaningful to see how their influence differs in terms of prediction, that is comparing $X_i \beta_i$ to $X_i \beta$. In particular, it may happen that $\|\beta_i-\beta\|^2$ is not very small but that $\|X_i(\beta_i-\beta)\|^2$ is. This means that $\beta_i$ and $\beta$ differ only in ways that exert little influence on the result $y_i$.

To make this more precise, recall the definition of the matrix $Q_i$ in Assumption (X), which is the limit of $X_i'X_i/T \in \mathbb{R}^{K \times K}$. We  define the hypothesis of approximate homogeneity in prediction as
\begin{align} \label{hypothesis_pred}
H_0^{pred}: S_N^{pred}:=\frac{1}{N} \sum_{i=1}^N \|Q_i^{1/2}(\beta_i - \beta^{pred})\|^2 \le \Delta \quad vs. \quad H_1^{pred}:S_N^{pred}> \Delta,
\end{align}
where
$
\beta^{pred}:= \big(\sum_{i=1}^N Q_i \big)^{-1} \sum_{i=1}^N Q_i \beta_i,
$
is the minimizer of $S_N^{pred}$. Notice that if the regressors are stationary, we have $\mathbb{E} \|X_i(\beta_i - \beta^{pred})\|^2= \|Q_i^{1/2}(\beta_i - \beta^{pred})\|^2$. We point out that $S_N^{pred}$ is proportional to the non-centrality parameter of the distribution of 
the $F$-test for slope homogeneity  under the classical alternative.
With only a few adaptions to our methods from the previous section, we can construct a self-normalized test for the hypothesis \eqref{hypothesis_pred} as well. 
To be precise we define a sequential version of the fixed effect estimator 
$$
\hat \beta^{pred}(\kappa):= \Big( \sum_{i=1}^N X_i(\kappa)' X_i(\kappa)\Big)^{-1} \sum_{i=1}^N X_i(\kappa)' y_i(\kappa)
$$
and therewith
$$
\hat S_N^{pred}(\kappa):= \frac{1}{NT}\sum_{i=1}^N \|X_i(\kappa)(\hat \beta_i(\kappa)-\hat \beta^{pred}(\kappa) )\|^2.
$$
Notice that for $\kappa=1$ this measure is closely related to Swamy's test statistic defined in \eqref{Swamy}. Finally, we introduce an adapted bias correction
$$
\hat B_N^{pred}(\kappa):=\kappa \sqrt{\frac{T}{N}}\sum_{i=1}^N\Tr[\hat \Sigma_i(\kappa, b)X_i(\kappa) \big[ X_i(\kappa)'X_i(\kappa)\big]^{-1} X_i(\kappa)' ]. 
$$
Under similar conditions as before (only $(X), (1)$ has to be strengthened slightly), we can prove a weak invariance principle in the spirit of 
\eqref{Eq_weak_convergence}, that is 
$$
\{ \kappa \sqrt{ N   T  }(\tilde S_N^{pred}(\kappa)- S_N^{pred}) \}_{\kappa \in I} \to  \{\LR^{pred} \mathbb{B}(\kappa)\}_{\kappa \in I},
$$
where $\tilde{S}_N^{pred}(\kappa):= \hat S_N^{pred}(\kappa)-\hat B_N^{pred} (\kappa)$ and $(\LR^{pred})^2$ is the long-run variance. Then we can adapt the self-normalized statistic $\hat W_N$ to  $\hat W_N^{pred}$ by replacing all instances of $\tilde S_N$ by  $\tilde S_N^{pred}$ and recover Theorem \ref{theorem_test} in this case, with the test decision
\begin{align}\label{test_decision_pred}
 \hat W_N^{pred}>q_{1-\alpha}   .
\end{align} We will use this prediction measure in the analysis of a data example in Section \ref{sec42}, where we find a case of relevant heterogeneity (both w.r.t. the raw slope comparison $S_N$ and to $S_N^{pred}$).

\section{Finite Sample Properties} \label{Section4}

In this section, we investigate the finite sample properties of the self-normalized test-statistic $\hat W_{N}$ defined in \eqref{hat_W} by means of a small simulation study. We  consider normally distributed regressors, generated by random $AR(1)$ processes and two different residual distributions, namely normal and centered chi-squared errors. To assess the performance of $\hat W_{N}$ we simulate samples with different sizes of $N$ and $T$, corresponding to different $\eta$ in our theory and we also investigate the effect of the dimension $K$ of the slopes. 

 We begin by specifying the model. We generate the random variables
\begin{align*}
e_{i,t} \sim \mathcal{N}(\mathbf{0}, R), \,\, i=1,..,N, \, t=1,..,T, \quad \rho_i \sim \mathcal{U}[0.05, 0.95], \,\, i=1,..,N,
\end{align*}
where $R$ is a $K \times K$ dimensional matrix, with entries $R_{i,j}=(1/2)^{|i-j|}$. $\rho_i$ is the AR-parameter of the $i$th individual and thus determines the temporal dependence of the regressors.
More precisely, the regressors are generated by the following procedure
$$x_{i,t}=\rho_i \cdot x_{i,t-1}+  e_{i,t}, $$
where a burn in period of $100$ iterations is used. Next we turn to the model errors defining the temporal and spatial covariance matrices $\Sigma_T$ and $\Sigma_N$. We define $\Sigma_T$ entry-wise as 
 $$(\Sigma_{T})_{t,s}:=\tau(t-s):=(1+|t-s|)^{-2}.$$
 Notice that even though $\tau$  is summable, it is slowly decaying, which translates into relatively strong dependence.  The cross-sectional covariance matrix is defined similarly as
$$ (\Sigma_{N})_{i,j}:=(1+|i-j|)^{-2} z_i z_j,$$
where $z_1,...,z_N$ are i.i.d.  centered normal  random variables  with variance $1/2$, generated in each simulation run. By construction $\Sigma_{N}$ is positive definite but does not imply stationarity. We then define the spatio-temporal covariance matrix $\Sigma := \Sigma_N \otimes \Sigma_T$. Now we generate i.i.d. standard normal  random  variables $g_{i,t}$ for $i=1,...,N, t=1,...,T$ and therewith construct two types of model errors $\varepsilon_{i,t}$: We consider normal model errors by applying the matrix root of $\Sigma$ to the vector $(g_{i,t})_{i,t}$, which yields
$\varepsilon^{norm} \sim \mathcal{N}(\mathbf{0}, \Sigma) $ (conditional on $z_{1}, \ldots , z_{N}$). 
Secondly, we consider dependent chi-squared errors, by taking the normal model errors $\varepsilon_{i,t}^{norm}$ from the previous step and transform them to $\varepsilon_{i,t}^{chi}$ as follows,
$$
\varepsilon_{i,t}^{chi} := {(\varepsilon^{norm}_{i,t})^2-\E (\varepsilon^{norm}_{i,t})^2 \over \sqrt{2}}, \qquad \forall i=1,...,N, \quad t=1,...,T.
$$
It follows by a simple calculation that the residual vector $\varepsilon^{chi}$ has covariance matrix $(\Sigma_{i,t}^2)_{i,t}$. In particular, the Assumption $(\varepsilon), (2)$ of a separable covariance is satisfied. Since the distribution of $\varepsilon^{chi}_{i,t}$ is strongly skewed, more temporal observations are required for the central limit theorem to set in, and we therefore expect a slower convergence of our test statistic. \\
Next we have to select the slopes $\beta_1,...,\beta_N$. These are generated randomly as follows
$$
\beta_i = \mathbf{1}_K+U_i,
$$
where $\mathbf{1}_K$ is the $K$-dimensional vector with entries $1$ and $U_1,...,U_N$ are i.i.d., uniformly distributed on the $K$-dimensional unit-ball. By construction, the heterogeneity measure $S_N$ is then also randomly distributed. However, as a means of understanding its approximate size, we can calculate its expectation given by
$$
\E S_N = \frac{N-1}{N} \E \|U_1\|^2.
$$
In the case of $K=2$ it follows by basic calculations that $\E \|U_1\|^2 =1/2$ and for larger $K$ the expectation grows slowly (converging to $1$ as $K \to \infty$), reflecting the impact of the dimension. \\
Finally, we fix the user-determined parameters in the construction of the statistic: We choose the measure $\nu$, which occurs in the denominator $\hat V_N$ (see \eqref{def_denominator}) of the statistic $\hat W_N$ as the uniform measure on $\{6/10,7/10,8/10,9/10, 1\}$. 
The bandwidth-parameter $b$ \begin{figure}
$$
\begin{matrix}
		 & & & \\
         & \includegraphics[scale=0.3]{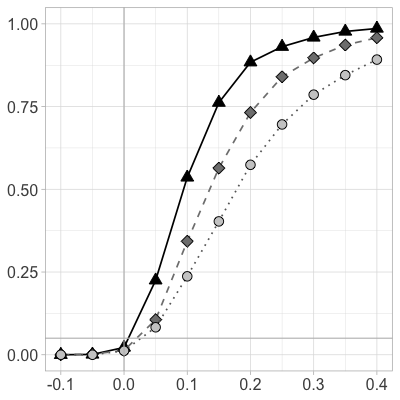}
         & \includegraphics[scale=0.3]{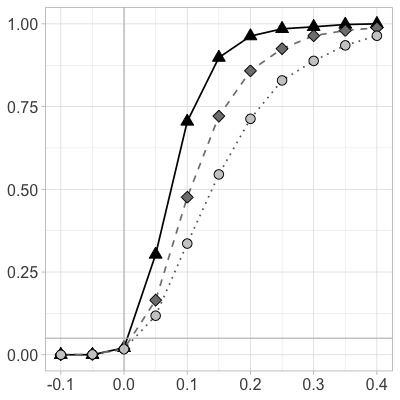}
         & \includegraphics[scale=0.3]{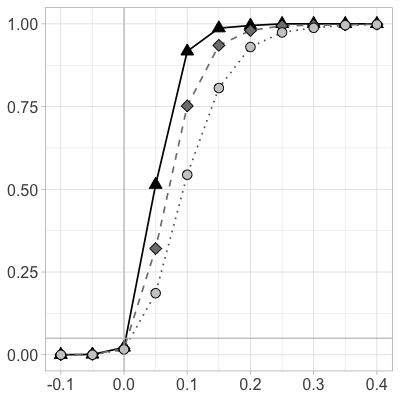}\\		
         & N=10, T=20 & N = 20, T=20 & N=50, T=20  \\[2ex]
         & \includegraphics[scale=0.3]{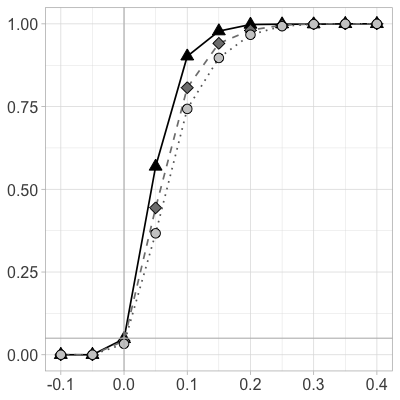}
         & \includegraphics[scale=0.3]{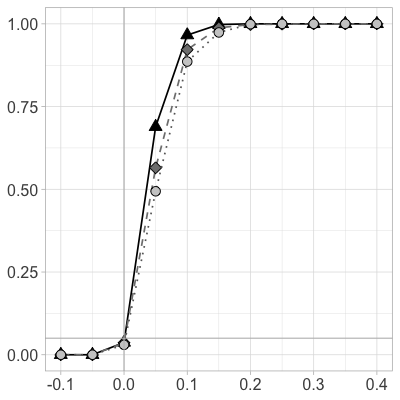}
         & \includegraphics[scale=0.3]{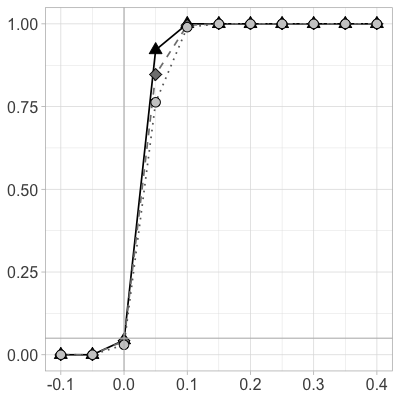}\\
         & N=10, T=50 & N = 20, T=50 & N=50, T=50 \\[2ex]
         & \includegraphics[scale=0.3]{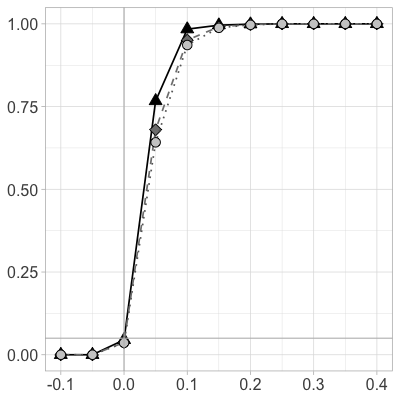}
         & \includegraphics[scale=0.3]{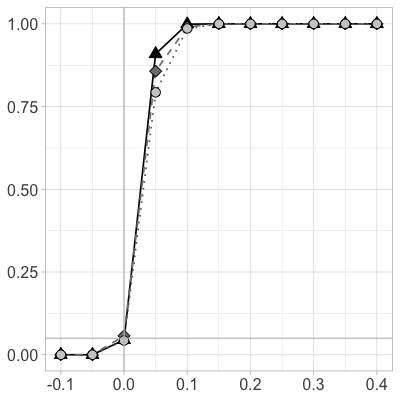}
         & \includegraphics[scale=0.3]{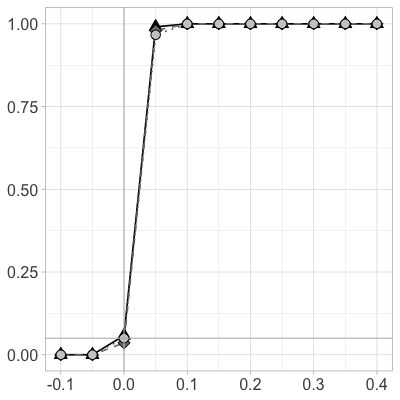}\\
         & N=10, T=100 & N = 20, T=100 & N=50, T=100 \\[2ex]
\end{matrix}
$$
\caption{\doublespacing \label{fig2}\textit{Rejection probabilities for different values of $N$ and $T$ in the case of normal errors. The $y$-axis corresponds to the empirical rejection probability of the test decision \eqref{test_decision} and the $x$-axis depicts different values of  $\epsilon:= S_N-\Delta$, where $\epsilon=0$ corresponds to the maximum rejection probability under $H_0$ and $\epsilon>0$ to the alternative. The different curves stand for $K=2$ (solid, triangles), $K=3$ (dashed, diamonds) and $K=4$ (dotted, circles).}}
     \end{figure}
  
\begin{figure}
$$
\begin{matrix}
		 & & & \\
         & \includegraphics[scale=0.3]{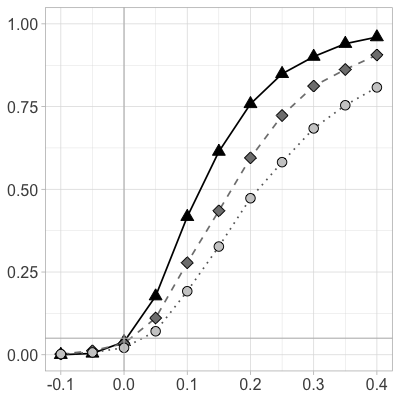}
         & \includegraphics[scale=0.3]{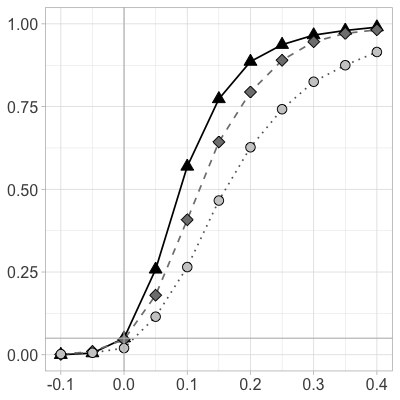}
         & \includegraphics[scale=0.3]{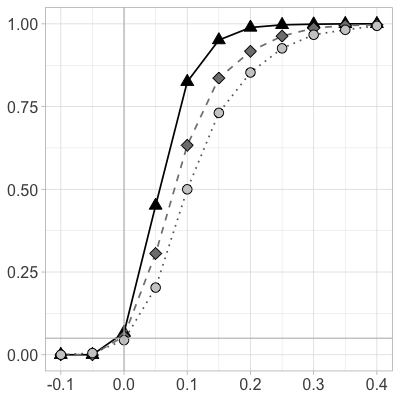}\\		
         & N=10, T=20 & N = 20, T=20 & N=50, T=20  \\[2ex]
         & \includegraphics[scale=0.3]{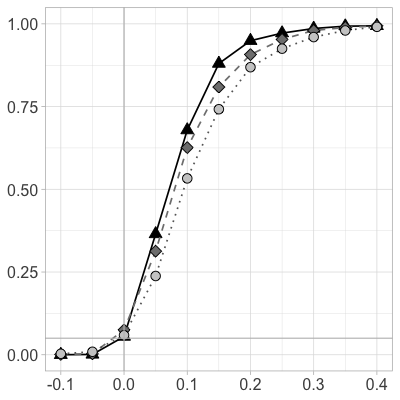}
         & \includegraphics[scale=0.3]{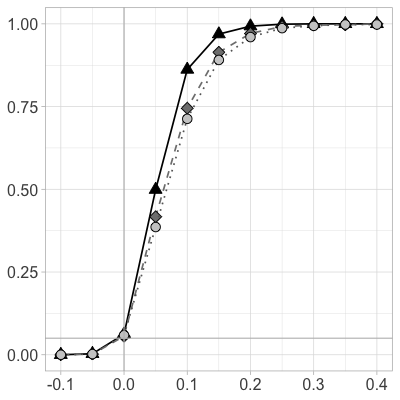}
         & \includegraphics[scale=0.3]{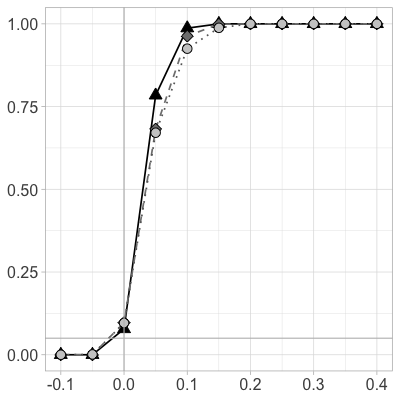}\\
         & N=10, T=50 & N = 20, T=50 & N=50, T=50 \\[2ex]
         & \includegraphics[scale=0.3]{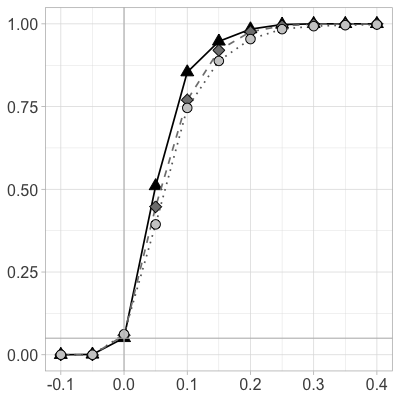}
         & \includegraphics[scale=0.3]{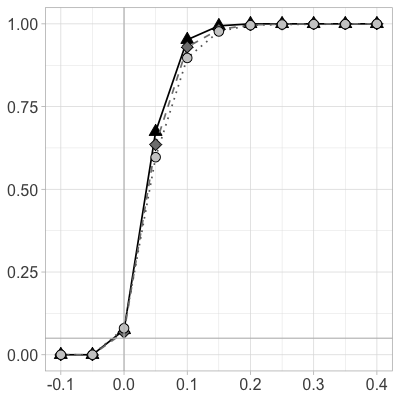}
         & \includegraphics[scale=0.3]{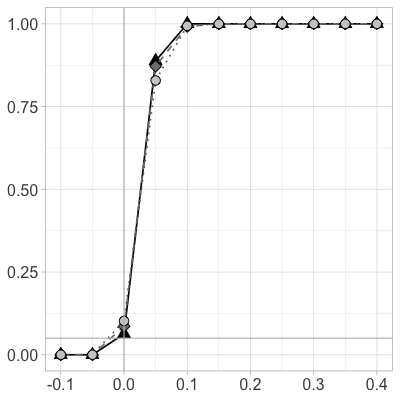}\\
         & N=10, T=100 & N = 20, T=100 & N=50, T=100 \\[2ex]
\end{matrix}
$$
\caption{\doublespacing\label{fig3}\textit{Rejection probabilities for different values of $N$ and $T$ in the case of chi-squared errors. The $y$-axis corresponds to the empirical rejection probability of the test decision \eqref{test_decision} and the $x$-axis depicts different values of  $\epsilon:= S_N-\Delta$, where $\epsilon=0$ corresponds to the maximum rejection probability under $H_0$ and $\epsilon>0$ to the alternative. The different curves stand for $K=2$ (solid, triangles), $K=3$ (dashed, diamonds) and $K=4$ (dotted, circles).}}
     \end{figure}
(see \eqref{def_hat_Sigma_i}) is set to $2$ everywhere and the nominal level $\alpha$ is fixed at $5\%$. 
 In Figures \ref{fig2} and \ref{fig3} we display  the rejection probabilities of the test  \eqref{test_decision},
  first in the case of normal and then of chi-squared errors, where
  $N \in \{10, 20, 50\}$ and $T \in \{20, 50, 100\}$. 
  Each of the subsequent diagrams corresponds to some combination of  values 
  $N$ and $T$ and shows three power-curves for $K=2$ (solid line with triangles),  $K=3$   (dashed line with diamonds)
  and $K=4$  (dotted line with circles).
  On the $y-axis$ we show the rejection probability and on the $x$-axis  the deviation $\epsilon = S_N-\Delta$, where $\epsilon \le 0$ corresponds to the hypothesis and $\epsilon>0$ to the alternative.
  The boundary of the hypothesis $\epsilon=0$ is marked by a vertical gray line, and the nominal level $\alpha =  5\%$ by a gray horizontal line.  
   All simulations are based on $1000$ simulation runs. 

 The results from our simulation study are in line with the asymptotic theory  presented in Sections \ref{sec33} and \ref{sec34}. We observe a good approximation of the nominal level, particularly when $T$ is large compared to $N$. The approximation is slightly more precise for normally distributed data than for the chi-squared case, as may be expected. Larger dimensions of the slope parameter are associated with decreased power for small $T$, whereas the effect is less pronounced for larger $T$. In general, we observe that additional temporal observations increase power faster than additional individual observations (compare $T=20, N=50$ with $T=50, N=20$), which is a well-known effect for data panels. In both cases - for normal and chi-squared errors - we see rapid growth of power for increasing samples, although normal errors lead to higher power than a skewed distributions. In the reported simulations, we have only considered  
  the choice $b=2$ of the bandwidth parameter $b$ (see definition of the bias correction in \eqref{bias_estimate}).  However, non-reported simulations show that the results do not change substantially when changing $b$ to $1,3$ or $4$. No bias reduction ($b=0$) leads to mildly inflated type-1-errors (particularly for non-normal data).

\section{A Data Example: CO${}_2$ Emissions in the G20}
\label{sec42}


Anthropomorphic $\textnormal{CO}_2$ emissions are a key contributor to climate change. According to a recent report of the   \cite{IPCC2021} $\textnormal{CO}_2$ emissions alone have contributed about  $0.7^{\circ}C$ temperature increase over the course of the last century (with some estimates even higher). The increase in global temperatures entails among other costs elevated risks of droughts, heatwaves and heavy precipitation. Against this background, it is important to identify factors that contribute (positively or negatively) to emissions and quantify their influence. 
Recently numerous studies have explored the nexus of carbon emissions and economic factors, such as GDP,  energy consumption, agricultural output, financial development, technology and many others (see e.g. \cite{Chang2015, SHUAI2017310, DONG2018180,BEKUN2019,  SHEN2021}). 
The data panels under investigation comprise regions of individual countries, as well as groups of countries (such as OECD members or the G20) over a moderate time frame. Some of these studies include tests to detect slope heterogeneity and cross-sectional dependence, because ignoring either can result in a distorted analysis. While standard homogeneity tests, like those in \cite{Pesaran} commonly reject the hypothesis of slope homogeneity in these works, practitioners tend to employ them even in the presence of (detected) cross-sectional dependence (see e.g. \cite{Chang2015, DONG2018180}). However, if cross-sectional dependence is not adjusted for, test statistics for homogeneity converge to non-standard limits, implying again a distorted analysis (see also \cite{Blomquist}). In contrast, the self-normalized statistics, presented in Section \ref{sec34} and \ref{sec35} are resilient to cross-sectional dependence (even for large intersections) and therefore offer a convenient alternative for users. We illustrate this point by investigating the $19$ G20 countries (excluding the EU) over a time frame of $18$ years (from $1998-2015$). Our data consists of annual measurements of  $\textnormal{CO}_2$ emissions ($CO$), renewable energy  consumption ($REN$), gross domestic product ($GDP$) and added agricultural output ($AGR$) all per capita, which is closely related to the panel considered in \cite{QIAO2019}. $\textnormal{CO}_2$ emissions are measured in metric tons, GDP and agricultural output added in 2015 dollars and renewable energy in terajoules. The measurements of renewable energy are based on the \textit{Sustainable Energy for All} database from the World Bank and the remaining variables are drawn from the \textit{World Development Indicators} also from the World Bank (\url{https://databank.worldbank.org/home}).  We now consider the panel regression model
$$
\log(CO_{i,t})=\alpha_i + \beta_{i,1} \log(GDP_{i,t})+\beta_{i,2} \log(REN_{i,t})+\beta_{i,3} \log(AGR_{i,t})+\varepsilon_{i,t}
$$
for $i=1,...,19$, $t=1,...,18$ and $K=3$. Notice that we include an individual specific intercept $\alpha_i$ for each country, which is eliminated by subtracting the temporal average for each individual (see Remark \ref{rem_intercepts}). The intercepts model idiosyncratic factors, such as the availability of natural resources, which are different for each country. In Figure \ref{Figure_Boxplots} we display box plots of the logarithmized variables.

\begin{figure} 
$$
\begin{matrix}
		 & & & & \\
         & \includegraphics[scale=0.4]{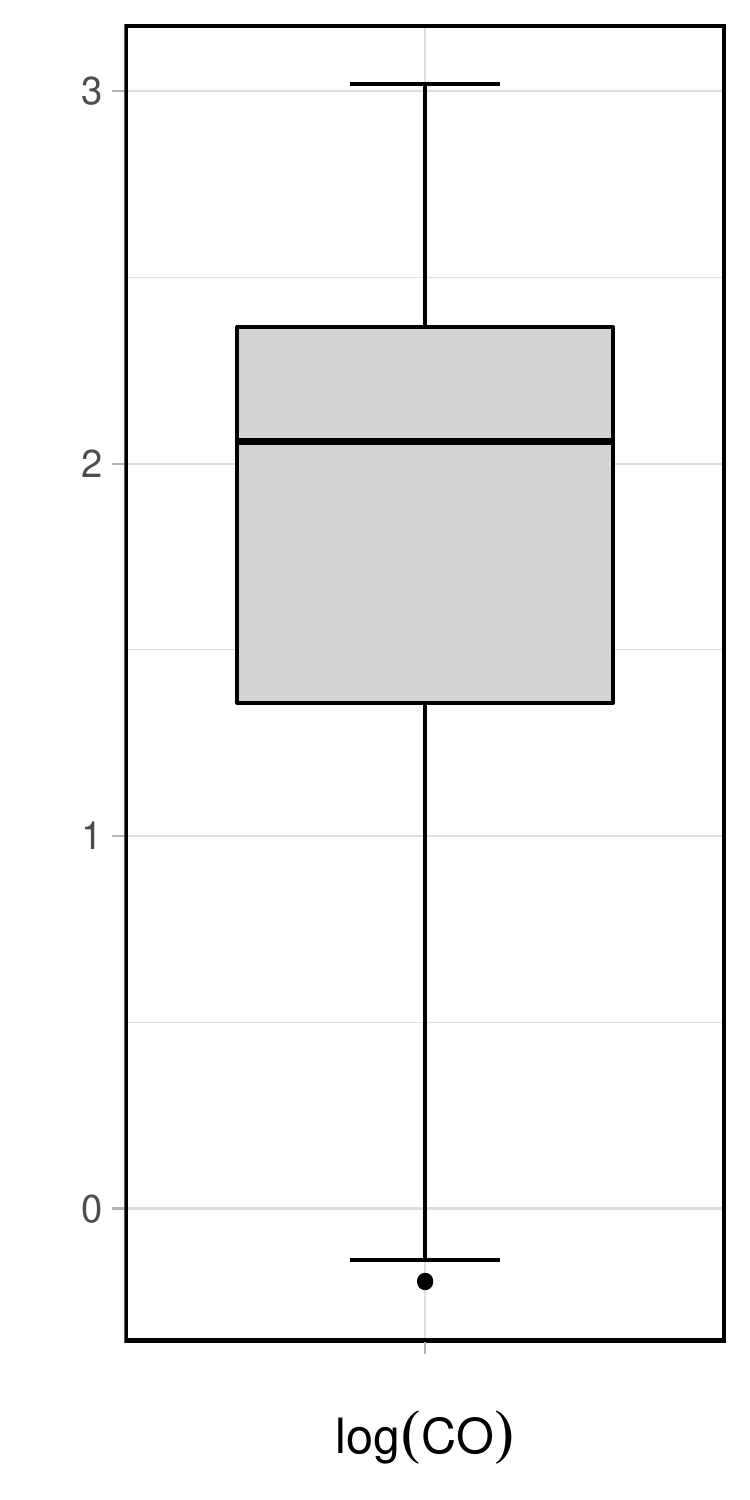}
         & \includegraphics[scale=0.4]{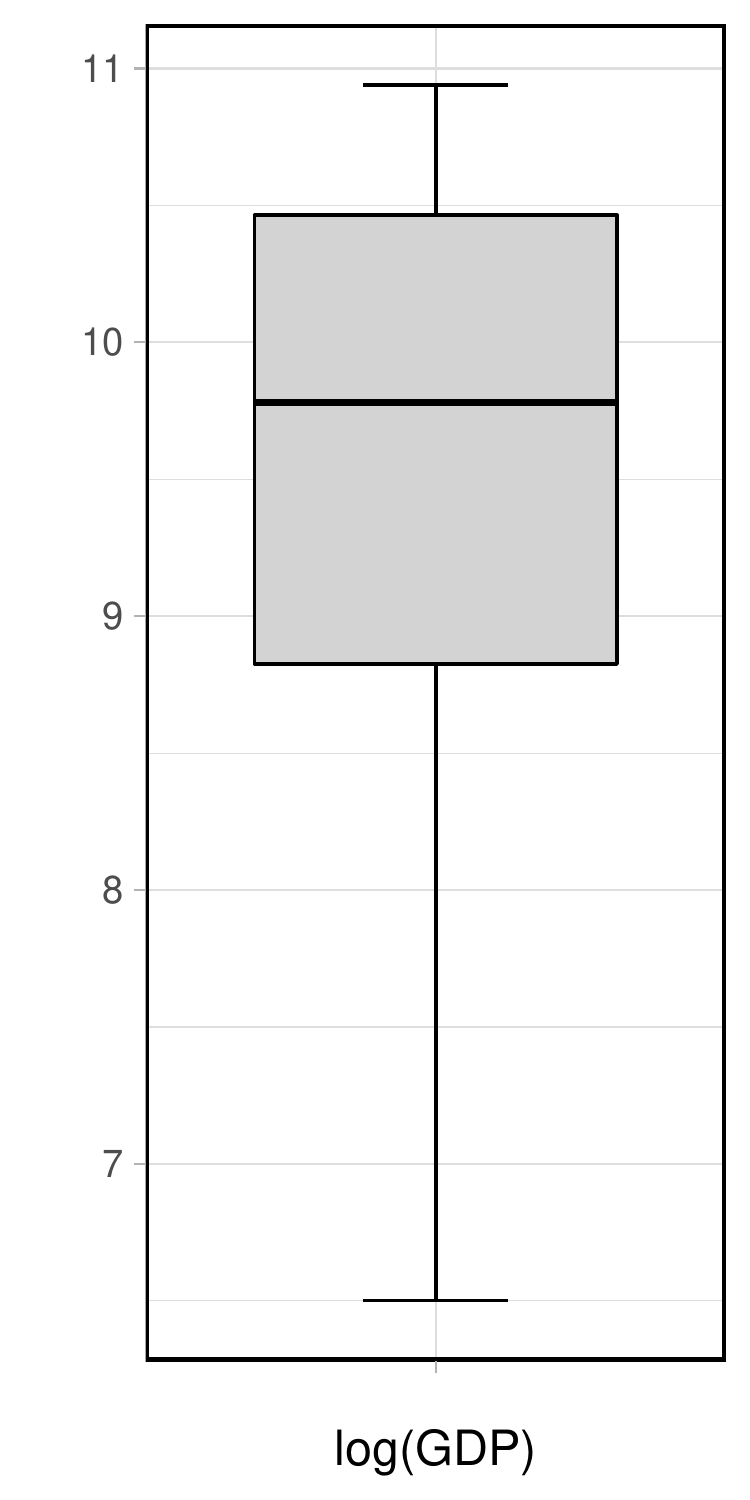}
         & \includegraphics[scale=0.4]{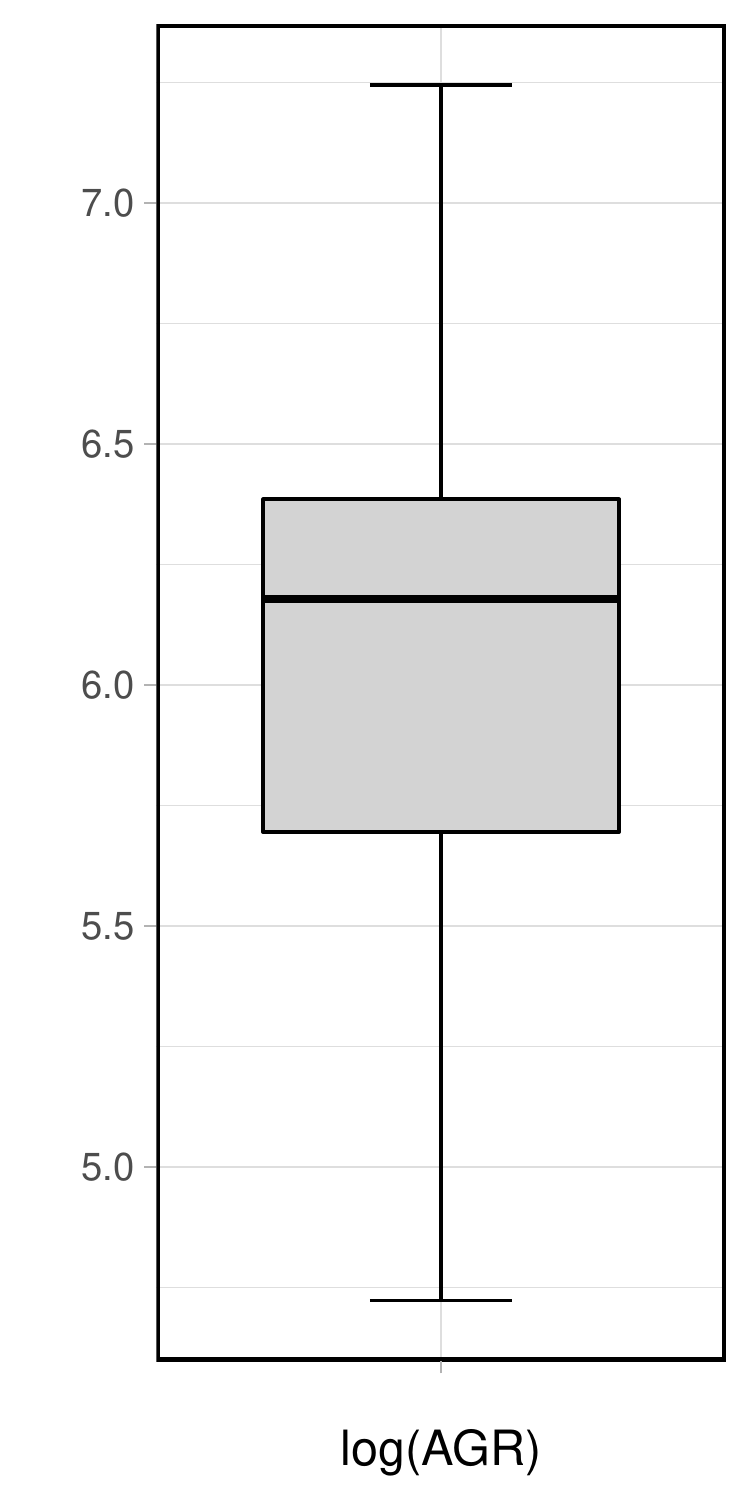}		
         & \includegraphics[scale=0.4]{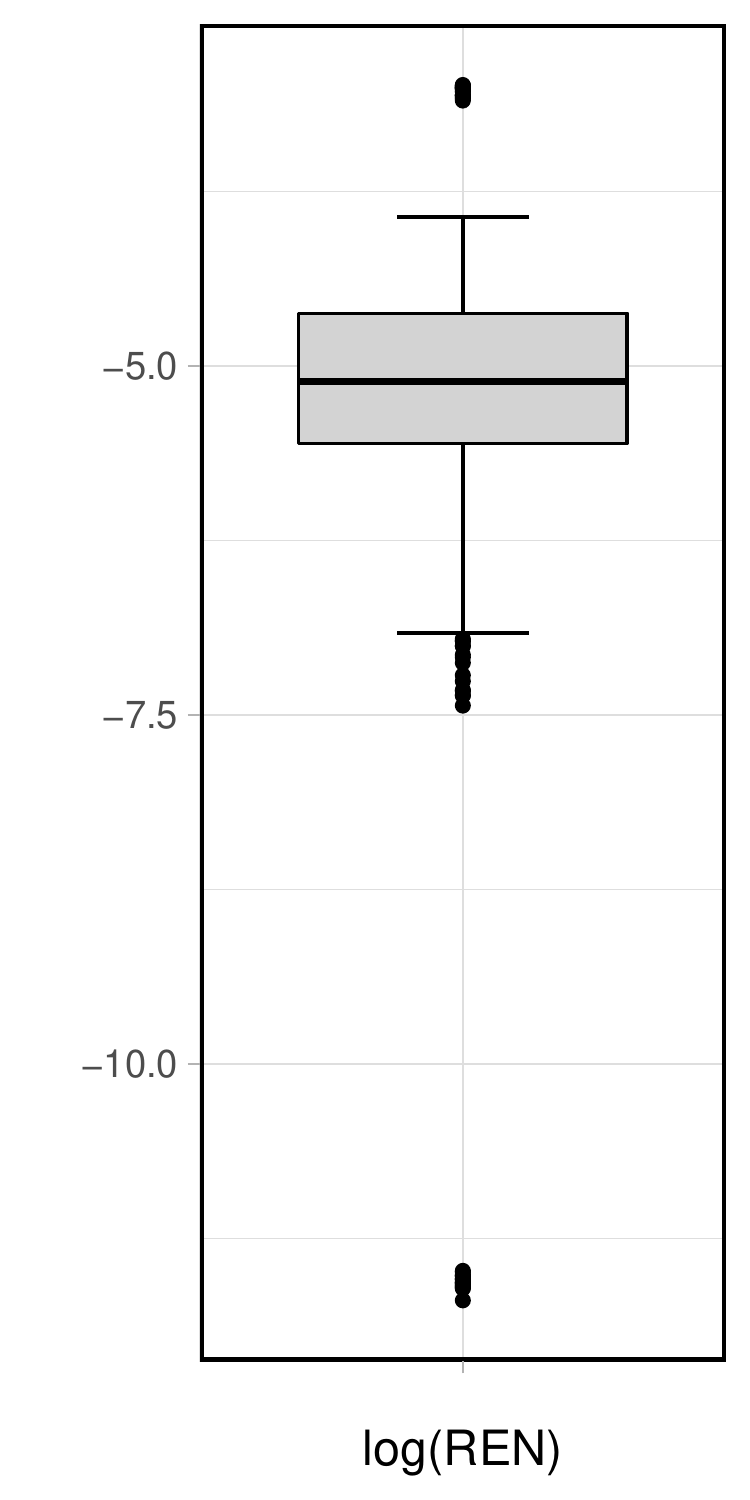}
\end{matrix}
$$
\caption{\doublespacing\textit{Box plots of logarithmized CO${}_2$ emissions, GDP, agricultural value added and renewable energy consumption (all per capita).}
\label{Figure_Boxplots}}
\end{figure}

We can now calculate the mean group estimator $\hat \beta$ and the fixed effect estimator $\hat \beta^{pred}$. Both estimators indicate the positive impact of GDP on emissions ($\approx 0.5$ and $0.72$, respectively) and a negative impact of renewables ($\approx -0.05$ and $-0.18$ respectively). In the case of agricultural output we see entries of equal magnitude but with different signs ($\approx 0.15$ and $-0.11$ respectively). Our results are in line with existing findings, which stress the positive effect of GDP and the negative effect of renewable energies on carbon emissions (see \cite{DONG2018180} and references therein). The case of agricultural output is known to be more ambiguous, as agriculture can be a source as well as a sink for $\textnormal{CO}_2$ (see \cite{JOHNSON2007}). Similarly, \cite{QIAO2019} find different signs for developed and developing economies, with a positive net effect. \\
The fact that mean group and fixed effect estimator differ noticeably is indicative of underlying slope heterogeneity (since $\hat \beta^{pred}$ is basically a weighted average of $\hat \beta_1,...,\hat \beta_N$). This impression is bolstered if we, ignoring intersectional dependence, apply the homogeneity tests of Swamy and Pesaran discussed in Section \ref{sec_21} (for a rigorous definition see Theorem $2$ in \cite{Pesaran}). Both tests reject exact homogeneity at a level of $1\%$. We retrieve strong heterogeneity when applying the self-normalized tests from Sections \ref{sec34} and \ref{sec35}. The test decision \eqref{test_decision} rejects the hypothesis of approximate slope homogeneity \eqref{hypothesis} at a level of $10\%$ for $\Delta = 0.042$, which is a rather large value, relative to the size of the slopes. For instance, the mean group estimator has norm $\|\hat \beta\|^2 \approx 0.267$ s.t. $\Delta/\|\hat \beta\|^2 \approx 0.16$, indicating strong heterogeneity compared to average effects. Similarly, the test decision \eqref{test_decision_pred} rejects the hypothesis \eqref{hypothesis_pred}  for  $\Delta \le 0.017$. In both cases we have employed 
a choice of $b=1$ for the bandwidth parameter (which is used in the bias estimate \eqref{bias_estimate}) 
 but moving to $b=2$ or to $b=0$ has little impact on the results. In light of the diverse nature of the G20 members as well as findings of earlier studies, these rejections are highly plausible. The fact that both tests reject indicates that the findings are stable.  

{\bf Acknowledgements}
This research has been supported by the German Research Foundation (DFG), project number 45723897.
 The authors report that there are no competing interests to declare.

\bibliographystyle{apalike}
\setlength{\bibsep}{1pt}
\begin{small}

\end{small}

\newpage

\appendix

\setcounter{page}{1}
\section*{\LARGE $\quad \,$ Supplements}

\section{Proofs of the main Results} \label{App_A}

The supplement consists of four parts. In Section \ref{App_A} we present the proofs of our main results from Section \ref{sec33}. In Section \ref{App_B} we have gathered additional technical results concerning the convergence of the sequential regressor matrices. These are applied in Section \ref{App_C} to bound the remainder terms occurring in Proposition \ref{theorem1}. Finally, in  Section \ref{App_D} we establish various uniform bounds, which are used throughout Sections \ref{App_A} and \ref{App_B}. 

\subsection{Notations}

\begin{itemize}
\item[1)] From now on $C$ always denotes a generic constant (which is  independent of $T, N$ and any individual or time point $i, t$) and can differ from line to line. 
\item[2)] $C_1, C_2,...$ denote fixed, positive constants, that do not change from line to line (but they can differ from  section to section).
\item[3)] Throughout our discussion, we will use different matrix norms, particularly on the finite dimensional space $\mathbb{R}^{K \times K}$. Usually we will not specify the norm $\|\cdot\|$ (as they are all equivalent on finite dimensional spaces). However, if necessary we will make the notation more precise by referring to the trace-norm $\|\cdot\|_1$, the Frobenius norm $\|\cdot\|_2$, the maximum absolute row sum norm  $\|\cdot\|_{row}$ or the spectral norm $\|\cdot\|_\infty$ of a matrix.
\end{itemize}

 
 \subsection{Proof of the Decomposition (\ref{linearization})}
 \label{seca1}
 
Recalling the notation of Section \ref{Section3} we have
\begin{align*}
&  \kappa\sqrt{N T}(\hat S_N( \kappa)- S_N) \\
= &   \kappa\sqrt{\frac{T}{N}}\sum_{i=1}^{N}  \big(\hat \beta_i(\kappa) -\beta_i -\hat \beta(\kappa) +\bar \beta\big)'  \big(\hat \beta_i(\kappa) +\beta_i -\hat \beta(\kappa) -\bar \beta\big)\\
=& \kappa  \sqrt{\frac{T}{N}} \sum_{i=1}^{N}  \Big(
  \big[ X_i(\kappa)'X_i(\kappa)\big]^{-1} X_i(\kappa)' \varepsilon_{i}(\kappa) - \frac{1}{N} \sum_{j=1}^{N}\big[ X_j(\kappa)'X_j(\kappa)\big]^{-1} X_j(\kappa)' \varepsilon_{j}(\kappa) \Big)' \\
 & \Big(2[\beta_i-\bar \beta]+\big[ X_i(\kappa)'X_i(\kappa)\big]^{-1} X_i(\kappa)' \varepsilon_{i}(\kappa) - \frac{1}{N} \sum_{j=1}^{N}\big[ X_j(\kappa)'X_j(\kappa)\big]^{-1} X_j(\kappa)' \varepsilon_{j}(\kappa) \Big)\\
=&2\kappa\sqrt{\frac{T}{N}}\sum_{i=1}^{N}  \varepsilon_{i}(\kappa)' X_i(\kappa) \big[ X_i(\kappa)'X_i(\kappa)\big]^{-1} [\beta_i-\bar \beta]\\
+&\kappa  \sqrt{\frac{T}{N}} \sum_{i=1}^{ N  }  \Big(
  \big[ X_i(\kappa)'X_i(\kappa)\big]^{-1} X_i(\kappa)' \varepsilon_{i}(\kappa) - \frac{1}{N} \sum_{j=1}^{N}\big[ X_j(\kappa)'X_j(\kappa)\big]^{-1} X_j(\kappa)' \varepsilon_{j}(\kappa) \Big)' \\
 & \Big(\big[ X_i(\kappa)'X_i(\kappa)\big]^{-1} X_i(\kappa)' \varepsilon_{i}(\kappa) - \frac{1}{N} \sum_{j=1}^{N}\big[ X_j(\kappa)'X_j(\kappa)\big]^{-1} X_j(\kappa)' \varepsilon_{j}(\kappa) \Big)\\
 =&2\kappa\sqrt{\frac{T}{N}}\sum_{i=1}^{N}  \varepsilon_{i}(\kappa)' X_i(\kappa) \big[ X_i(\kappa)'X_i(\kappa)\big]^{-1} [\beta_i-\bar \beta]\\
 +&\kappa  \sqrt{\frac{T}{N}} \sum_{i=1}^{ N  }  \varepsilon_i(\kappa)' X_i(\kappa) \big[ X_i(\kappa)'X_i(\kappa)\big]^{-2} X_i(\kappa)' \varepsilon_i(\kappa)\\
-& \kappa  \sqrt{\frac{T}{N}} \frac{1}{N} \sum_{i,j=1}^{N} \varepsilon_{i}(\kappa)X_i(\kappa) \big[ X_i(\kappa)'X_i(\kappa)\big]^{-1} \big[ X_j(\kappa)'X_j(\kappa)\big]^{-1} X_j(\kappa)' \varepsilon_{j}(\kappa),
\end{align*}
where we have used the identity $\sum_{i=1}^N ( \beta_i -\bar \beta ) =0$ in the second equality. This completes the proof of the representation \eqref{linearization}.  \\

\subsection{Proof of Proposition \ref{theorem1}}

\textbf{Proof of \eqref{h1a}: } 
We first define  
\begin{equation} \label{def_tilde_E_3}
\tilde E_3(\kappa) :=- \frac{1}{\sqrt{TN}}  \|\breve E_3(\kappa)\|^2,
\end{equation}
where
\begin{equation} \label{def_breve_E_3}
   \breve E_3(\kappa) := \Big(\frac{1}{\sqrt{NT\kappa} }\sum_{i=1}^{N} \sum_{t=1}^{\lfloor T \kappa \rfloor} \varepsilon_{i,t} (x_{i,t}'Q_i^{-1})_{k} \Big)_{k=1,...,K}~. 
\end{equation}

Recall that $(x_{i,t}'Q_i^{-1})_{k} $ here refers to the $k$th entry of the $K$-dimensional vector $x_{i,t}'Q_i^{-1}$.
Now 
$
E_3(\kappa) = \tilde E_3(\kappa) +o_P(1), 
$
holds uniformly in $\kappa$ according to Proposition \ref{prop_replacement_E}.  Thus, it suffices to show that $\,\,\sup_{\kappa \in I}|\tilde E_3(\kappa)|=o_P(1)$ . For this purpose we investigate the process $\breve E_3(\kappa)$ more closely:
Proposition \ref{prop_all_moments}, part $iv)$ implies for a $\zeta>0$, which can be made arbitrarily small, that
$$
    \E \Big[ \sup_{\kappa \in I}| (\breve  E_3(\kappa))_k|^2\Big]=\mathcal{O}(T^{\zeta}),
$$
for any $k=1,...,K$. Since $K$ is finite and non-increasing, this yields 
\begin{equation} \label{rate_r2}
    \E \Big[\sup_{\kappa \in I}\|\breve  E_3(\kappa) \|^2\Big]=\mathcal{O}(T^{\zeta}).
\end{equation}
Together \eqref{rate_r2} (with $\zeta$ sufficiently small) and the definition of $\tilde E_3$ in \eqref{def_tilde_E_3} entail that \\$\sup_{\kappa \in I}|\tilde E_3(\kappa)|=o_P(1)$. \\

\textbf{Proof of \eqref{h1}: } 
Proposition \ref{prop_replacement_E} implies
$$
E_2(\kappa)-\E[E_2(\kappa)|\mathbf{X}] = \tilde E_2(\kappa)-\E[\tilde E_2(\kappa)|\mathbf{X}] + o_P(1),
$$
where the term $\tilde E_2(\kappa)$ is defined by 
\begin{equation} \label{def_tilde_E_2}
\tilde E_2 (\kappa)= 
 \frac{1}{(T\kappa)\sqrt{TN}} \sum_{i=1}^{ N  }  \varepsilon_i(\kappa)' X_i(\kappa) Q_i^{-2} X_i(\kappa)' \varepsilon_i(\kappa) ~.
\end{equation}
Consequently, we only have to show
the desired order $o_P(1)$ for the  term 
\begin{align} \label{def_E_2_breve}
\breve E_2(\kappa) := & \tilde E_2 (\kappa)-\E[\tilde E_2(\kappa) |\mathbf{X}]  \\
=  &   \frac{1}{(T\kappa)\sqrt{TN}} \sum_{i=1}^{ N  } \Tr[Q_i^{-1}X_i(\kappa)' \{\varepsilon_i(\kappa)\varepsilon_i(\kappa)'-\E[\varepsilon_i(\kappa)\varepsilon_i(\kappa)'|\mathbf{X}] \} X_i(\kappa) Q_i^{-1}] \nonumber
\\
=& \frac{1}{(T\kappa)\sqrt{TN}}\sum_{i=1}^N \sum_{k=1}^K \sum_{s,t=1}^{\lfloor T \kappa \rfloor} ( x_{i,t}' Q_i^{-1})_k (\varepsilon_{i,t} \varepsilon_{i,s}-\sigma_{i,i} \tau(|t-s|)) ( x_{i,s}' Q_i^{-1})_k. \nonumber
\end{align}
Here we have used the  fact that $\E [ \varepsilon_{i,t}\varepsilon_{j,s}] = \sigma_{i,j} \tau(|t-s|)$  in the last step (see Assumption $(\varepsilon)$, $(4)$).
Proposition \ref{prop_all_moments}, part $vii)$ implies that for some $\zeta>0$ which can be made arbitrarily small the inequality
$$
    \E \Big [\sup_{\kappa \in I}  \breve E_2(\kappa)^2 \Big ] \le C T^{\zeta-1}
$$
    holds, where $C>0$ only depends on $\zeta$. 
Now, choosing $\zeta<1$ completes the proof of Proposition \ref{theorem1}.

\subsection{Proof of Proposition \ref{theorem4}}

The proof consists of four steps: First, we derive a decomposition of the difference $\hat B_N(\kappa)-\E[E_2(\kappa)|\mathbf{X}]$ (where $\hat B_N(\kappa)$ is defined in \eqref{bias_estimate}) into three remainder terms. Subsequently, we prove that each of the remainders is uniformly of order $o_P(1)$. \\

\textbf{Step 1: Decomposition of } $\hat B_N(\kappa)-\E[E_2(\kappa)|\mathbf{X}]$.\\
We begin by defining the “ideal” estimator for the residual covariance of the $i$th individual with time lag $h$ as 
\begin{equation} \label{def_tilde_xi} 
    \tilde \xi_i(h, \kappa) := \frac{(\varepsilon_i(\kappa))_{1:\lfloor T \kappa \rfloor -h}' (\varepsilon_i(\kappa))_{h+1:\lfloor T \kappa \rfloor}}{\lfloor T \kappa \rfloor-h-K} = {1 \over \lfloor T \kappa \rfloor-h-K} \sum_{i=1}^{\lfloor T \kappa \rfloor-h} \varepsilon_{i,t}\varepsilon_{i,t+h}.
\end{equation}
Therewith we define the “ideal estimated temporal matrix” as
$$
\tilde \Sigma_i(\kappa, b) := (\tilde \xi_i(|s-t|, \kappa) \mathbbm{1}\{|s-t| <b\}))_{1\le s,t\le T },
$$
which corresponds to the matrix  $\hat \Sigma_i(\kappa)$ defined in \eqref{def_hat_Sigma_i}, where $\hat \xi_i(h, \kappa)$ is replaced by $\tilde \xi_i(h, \kappa)$ everywhere. 
We  now use these matrices to derive the following decomposition
\begin{align*}
\hat B_N(\kappa)-\E[E_2(\kappa)|\mathbf{X}] = & \kappa\sqrt{\frac{T}{N}}\sum_{i=1}^N\Tr[\{\sigma_{i,i} \Sigma_{T }-\hat \Sigma_i(\kappa)\} X_i(\kappa) \big[ X_i(\kappa)'X_i(\kappa)\big]^{-2} X_i(\kappa)' ]\\
= & R ^{(1)}(\kappa)+R ^{(2)}(\kappa)+R ^{(3)}(\kappa), 
\end{align*}
where 
\begin{align*}
 R ^{(1)}(\kappa) := & \kappa\sqrt{\frac{T}{N}}\sum_{i=1}^N\Tr[\{\sigma_{i,i} \Sigma_{T }^{(b)}-\tilde \Sigma_i(\kappa))\} X_i(\kappa) \big[ X_i(\kappa)'X_i(\kappa)\big]^{-2} X_i(\kappa)' ] \\
 R ^{(2)}(\kappa):= &\kappa\sqrt{\frac{T}{N}}\sum_{i=1}^N\Tr[\{\sigma_{i,i} \Sigma_{T }^{(b^c)}\} X_i(\kappa) \big[ X_i(\kappa)'X_i(\kappa)\big]^{-2} X_i(\kappa)' ]\\
 R ^{(3)}(\kappa) := & \kappa\sqrt{\frac{T}{N}}\sum_{i=1}^N\Tr[\{\tilde \Sigma_i(\kappa))- \hat \Sigma_i(\kappa))\} X_i(\kappa) \big[ X_i(\kappa)'X_i(\kappa)\big]^{-2} X_i(\kappa)' ].
\end{align*}
Here $$
\Sigma_{T }^{(b)}:=(\tau(|s-t|) \mathbbm{1}\{|s-t| < b\})_{1 \le s,t \le T},
$$
 i.e. it equals the matrix $ \Sigma_{T }$, except that it is truncated to a bandwidth of $b$. The matrix  $ \Sigma_{T }^{(b^c)}$ is defined as the remainder $ \Sigma_{T }^{(b^c)}:= \Sigma_{T }- \Sigma_{T }^{(b)}$. We now investigate the three terms separately and show that they converge to $0$ uniformly.\\

\textbf{Step 2: Proof of} $ \sup_{\kappa \in I}|R ^{(1)}(\kappa)| = o_P(1)$.\\
 Proposition \ref{proposition_B} 
implies that uniformly in $\kappa$
\begin{equation}
\label{hd6} 
R ^{(1)}(\kappa) = \mathbbm{1}_{\mathscr{B}^c} \cdot R ^{(1)}(\kappa)  + o_P(1)
\end{equation}
where the event $\mathscr{B}$ 
is defined by 
\begin{equation} \label{bad_event}
\mathscr{B}:=\bigcup_{i=1}^N \{\sup_{\kappa \in I} \|X_i(\kappa)'X_i(\kappa)/\lfloor T \kappa \rfloor-Q_i\|_\infty \ge   C_2/2  \Big \}
\end{equation}
and  $ C_2:=(\max_{i=1,...,N} \|Q_i^{-1}\|_\infty)^{-1} $.
It thus suffices to show the desired rate for $\mathbbm{1}_{\mathscr{B}^c} \cdot R ^{(1)}(\kappa)$. We therefore consider the  expectation $E 
\Big [ \sup_{\kappa \in I} |\mathbbm{1}_{\mathscr{B}^c} \cdot R ^{(1)}(\kappa) | \Big ]$, which can be upper bounded as follows:
\begin{align} \label{R_1_bound}
& \E 
\Big [ \sup_{\kappa \in I} |\mathbbm{1}_{\mathscr{B}^c} \cdot R ^{(1)}(\kappa)| \Big ]  \\
~~~~~~~& \le \kappa  \sqrt{\frac{T}{N}}\sum_{i=1}^N
 \E  \Big [  \sup_{\kappa \in I}\Big \|\{\sigma_{i,i} \Sigma_{T }^{(b)}-\tilde \Sigma_i(\kappa))\Big\|_\infty \Big\|X_i(\kappa) \big[ X_i(\kappa)'X_i(\kappa)\big]^{-2} X_i(\kappa)' ] \mathbbm{1}_{\mathscr{B}^c} \Big\|_1 \Big] \nonumber\\
 ~~~~~~~&  \le  C  \frac{1}{\sqrt{NT}}\sum_{i=1}^N\E 
\Big [ \sup_{\kappa \in I}\Big \|\sigma_{i,i} \Sigma_{T }^{(b)}-\tilde \Sigma_i(\kappa)\Big\|_\infty  \Big ]. \nonumber
\end{align}
Here $\|\cdot\|_\infty$ denotes the spectral norm. In the first step, we have used the triangle inequality, as well as the identity $\Tr(AB)  \le \|A \|_\infty \| B \|_1$. In the second inequality, we have used $\kappa \le 1$ and Corollary \ref{regressor_stuff}$(i)$,
which is proved at the end of this section.
Turning to the right side of \eqref{R_1_bound}, we upper bound each expectation in the sum, using the inequality
$$
 \|\sigma_{i,i} \Sigma_{T }^{(b)}-\tilde \Sigma_i(\kappa)\|_\infty \le \|\sigma_{i,i} \Sigma_{T }^{(b)}-\tilde \Sigma_i(\kappa)\|_{row}~, 
$$
 where $\|\cdot\|_{row}$ is the maximum absolute row sum norm. Here we have used that the spectral norm of a symmetric matrix is upper bounded by the absolute row sum norm, which follows by a simple calculation.
Plugging this in on the right side of \eqref{R_1_bound} (and calculating the maximum absolute row sum norm) yields the estimate
\begin{equation}
\label{hd7}
\E 
\Big [ \sup_{\kappa \in I} | \mathbbm{1}_{\mathscr{B}^c} \cdot R^{(1)}(\kappa)| \Big ]  
\le 
C  \frac{1}{\sqrt{NT}}\sum_{i=1}^N
\sum_{h=0}^{b-1} \E  \Big [ \sup_{\kappa \in I} |\tilde \xi_i (h, \kappa) -\sigma_{i,i}\tau(h)|\Big ]~.
\end{equation}
Now it holds for any fixed $h$ that
\begin{align*}
\tilde \xi_i (h, \kappa) -\sigma_{i,i} \tau(h) & =\frac{(\varepsilon_i(\kappa))_{1:\lfloor T \kappa \rfloor-h}' (\varepsilon_i(\kappa))_{h+1:{\lfloor T \kappa \rfloor}}}{\lfloor T \kappa \rfloor-h-K}
-\frac{\E(\varepsilon_i(\kappa))_{1:{\lfloor T \kappa \rfloor}-h}' (\varepsilon_i(\kappa))_{h+1:{\lfloor T \kappa \rfloor}}}{{\lfloor T \kappa \rfloor}-h}, \\
&= (1+\mathcal{O}(1/T)) R^{(1,1)}(h, \kappa)+\mathcal{O}(1/T)
~,
\end{align*}
where
$$ 
 R^{(1,1)}(h, \kappa):= \frac{1}{{\lfloor T \kappa \rfloor}} \sum_{t=1}^{\lfloor T \kappa \rfloor -h}
\big \{ 
\varepsilon_{i,t}\varepsilon_{i,t+h}-\E[ \varepsilon_{i,t}\varepsilon_{i,t+h}] \big \}.
$$
Now by Proposition \ref{prop_all_moments}, part $vi)$ it follows that 
 \begin{equation}\label{hd25}
\E  \Big[ \sup_{\kappa \in I} | R^{(1,1)}(h, \kappa)|^2 
\Big ]
\le C b^2 T^{-1+\zeta}
      \end{equation}
      for some $C>0$ independent of $h$ and some  $\zeta>0$, which can be made arbitrarily small.  
Combining \eqref{hd25} and \eqref{hd7}
      yields
\begin{equation} \label{reference_for_b}
\E \Big [ \sup_{\kappa \in I} | \mathbbm{1}_{\mathscr{B}^c} \cdot R ^{(1)}(\kappa) | \Big ]  = \mathcal{O}(b^2 \sqrt{N} T^{1-\zeta})
=o( \sqrt{N/T^{\eta}})= o(1) ,
\end{equation}
where the last estimate follows from our choice  $b=\mathcal{O}(T^\gamma)$ with  $\gamma < (1-\eta/2)/2$) 
(note that   $N/T^\eta =o(1)$ by Assumption $(N)$ and 
that $\zeta$ can be chosen sufficiently small). 
Recalling that \eqref{hd6} holds uniformly with respect to $ \kappa \in I $
yields  
$$
\sup_{\kappa \in I} |R^{(1)}(\kappa)| = o_P(1)
$$
as claimed.\\

\textbf{Step 3: Proof of} $ \sup_{\kappa \in I}|R ^{(2)}(\kappa)| = o_P(1)$.\\
As in the second step, we use Proposition \ref{proposition_B}  from Supplement \ref{App_B}
to see that uniformly in $\kappa$
\begin{equation}
\label{hd6b} 
R ^{(2)}(\kappa) = \mathbbm{1}_{\mathscr{B}^c} \cdot R ^{(2)}(\kappa)  + o_P(1)
\end{equation}
where the event $\mathscr{B}$ is defined in \eqref{bad_event}.
By  simple application of the triangle inequality we can upper bound $\E \big[\sup_{\kappa \in I} |  \mathbbm{1}_{\mathscr{B}^c} \cdot R ^{(2)}(\kappa) | \big]$ by
\begin{align*}
& \frac{C \sqrt{T}}{\sqrt{ N}}\sum_{i=1}^N\|\{\sigma_{i,i} \Sigma_{T }^{(b^c)}\|_\infty \E \Big[ \sup_{\kappa \in I} \mathbbm{1}_{\mathscr{B}^c}\cdot \|X_i(\kappa) \big[ X_i(\kappa)'X_i(\kappa) \big]^{-2} X_i(\kappa)' \|_1 \Big]  \\
= &  \mathcal{O}(\sqrt{N}\| \Sigma_{T }^{(b^c)}\|_{\infty}/\sqrt{T}),
\end{align*}
where  we have used the boundedness of the regressor matrices, proved in Corollary \ref{regressor_stuff}, part $i)$.
Moreover,  as the maximum absolute row sum norm  provides an upper bound for the spectral norm of a  symmetric matrix (see step 2), we get the following estimate
\begin{equation}\label{cov_sum}
  \| \Sigma_{T }^{(b^c)}\|_\infty \le   \| \Sigma_{T }^{(b^c)}\|_{row}\le \sum_{|h|\ge b} |\tau(|h|)| \le \sum_{|h|\ge b} C  |h|^{-a\frac{M-1}{M}} \le C   b^{-a\frac{M-1}{M}+1 }.
\end{equation}
 In \eqref{cov_sum} we have used the well-known covariance inequality  (Lemma 3.11 in \cite{dehling})
$$
|\tau(|h|)|= |\E \varepsilon_{1,1}\varepsilon_{1,|h|+1}|/\sigma_{1,1} \le  C (\E( \varepsilon_{1,1})^{2M}\E( \varepsilon_{1,1})^{2M})^{1/(2M)} \alpha(|h|)^{1-1/M}
$$
for $\alpha$-mixing random variables
(the existence of moments of order $2M$ 
is guaranteed by Assumption $(\varepsilon), (3)$ and mixing by $(\varepsilon), (5)$)).
 Consequently, for the choice $b=T^\gamma$ with 
$$
\gamma > \frac{\eta-1}{2(a(M-1)/M-1)}
$$
 we get  
$$
\E \Big[ \sup_{\kappa \in I} |\mathbbm{1}_{\mathscr{B}^c} \cdot R ^{(2)}(\kappa)|\Big ] =\mathcal{O}(\sqrt{N/T}b^{-a\frac{M-1}{M}+1})=o(1).
$$
In the last step, we have used that $\sqrt{N/T^\eta} = o(1)$ by Assumption $(N)$. Now \eqref{hd6b} implies
$$
    \sup_{\kappa \in I} | R ^{(2)}(\kappa)| = o_P(1).\\
$$

\textbf{Step 4: Proof of} $ \sup_{\kappa \in I}|R ^{(3)}(\kappa)| = o_P(1)$.\\
We use again Proposition \ref{proposition_B} 
to see that uniformly in $\kappa$
\begin{equation}
\label{hd6c} 
R ^{(3)}(\kappa) = \mathbbm{1}_{\mathscr{B}^c} \cdot R ^{(3)}(\kappa)  + o_P(1).
\end{equation}
We now investigate $\mathbbm{1}_{\mathscr{B}^c} \cdot R ^{(3)}(\kappa)$. 
Using the same techniques as in the second step, we see that 
\begin{equation}
    \label{hd31}
\E \big [  \sup_{\kappa \in I} | \mathbbm{1}_{\mathscr{B}^c} \cdot R ^{(3)}(\kappa)| \Big ]
= \mathcal{O}_P \Big (\frac{\sqrt{N}}{\sqrt{T}}\max_{i=1,...,N} \E  \Big [ 
\sup_{\kappa \in I} \|\tilde \Sigma_i(\kappa))- \hat \Sigma_i(\kappa))\|_{row}\Big] \Big ).
\end{equation}
The expectation in the maximum  can be bounded independently of $i$. To see this, we further analyze it as follows:
\begin{align}
\label{hd32}
     \E
    \Big [ 
      \sup_{\kappa \in I} \|\tilde \Sigma_i(\kappa))- \hat \Sigma_i(\kappa))\|_{row} \Big ]  \le &  \sum_{h=0}^{b-1} \E  \Big [  \sup_{\kappa \in I} |\tilde \xi_i (h, \kappa) -\hat \xi_i (h, \kappa)|\Big] \\
    \le & b \max_{0 \le h <b } \E  \Big [  \sup_{\kappa \in I} |\tilde \xi_i (h, \kappa) -\hat \xi_i (h, \kappa)| \Big ].
    \nonumber
\end{align}
We now show for some fixed but arbitrary $0 \le h <b$, that the expectation on the right is bounded by
$C T^{-1/2+\zeta}$ (with $C$ also independent of  $h \in \{ 1,  \ldots, b\}$ and $\zeta>0$ arbitrarily small). Recalling that
$$
y_i(\kappa)-X_i(\kappa)'\hat \beta_i =\varepsilon_i(\kappa)-X_i(\kappa)'[X_i(\kappa)'X_i(\kappa)]^{-1}X_i(\kappa) \varepsilon_i(\kappa)~,
$$
as well as the definitions of $\tilde \xi_i$ (in \eqref{def_tilde_xi}) and $\hat \xi_i $ (in \eqref{def_hat_xi}),
we obtain the decomposition
\begin{equation} \label{hd30}
     |\tilde \xi_i (h, \kappa) -\hat \xi_i (h, \kappa)| \le R^{(3,1)}(\kappa)+R^{(3,2)}(\kappa)+R^{(3,3)}(\kappa),
\end{equation}
where 
\begin{align*}
    R^{(3,1)}(\kappa):= &  \Big |\frac{(\varepsilon_i(\kappa) )_{1:\lfloor T\kappa\rfloor-h}' (X_i(\kappa)'[X_i(\kappa)'X_i(\kappa)]^{-1}X_i(\kappa) \varepsilon_i(\kappa))_{h+1:\lfloor T\kappa\rfloor}}{\lfloor T\kappa\rfloor -h-K} \Big | ~, \\
    R^{(3,2)}(\kappa) := &  \Big|\frac{(X_i(\kappa)'[X_i(\kappa)'X_i(\kappa)]^{-1}X_i(\kappa) \varepsilon_i(\kappa) )_{1:\lfloor T\kappa\rfloor-h}'}{\lfloor T\kappa\rfloor -h-K}\\
    &\times \frac{(X_i(\kappa)'[X_i(\kappa)'X_i(\kappa)]^{-1}X_i(\kappa) \varepsilon_i(\kappa))_{h+1:\lfloor T\kappa\rfloor}}{\lfloor T\kappa\rfloor -h-K}  \Big|~,  \\
    R^{(3,3)}(\kappa) := &  \Big |\frac{ (X_i(\kappa)'[X_i(\kappa)'X_i(\kappa)]^{-1}X_i(\kappa) \varepsilon_i(\kappa))_{1:\lfloor T\kappa\rfloor-h}'( \varepsilon_i(\kappa) )_{h+1:\lfloor T\kappa\rfloor}}{\lfloor T\kappa\rfloor -h-K} \Big  | ~.
\end{align*}
We can show the desired bound for each of these remainders individually, i.e.
$$
\E \Big [ \sup_{\kappa \in I} |R^{(3,j)}(\kappa) | 
\Big ]
\le C T^{-1/2+\zeta}~,~~ j=1,2,3~, 
$$
where we focus on $R^{(3,1)}(\kappa)$ (the other terms can be treated similarly). 
The Cauchy-Schwarz inequality (applied twice, first for vectors, then for expectations) yields  
\begin{equation} \label{hd29}
    \E  \Big [ \sup_{\kappa \in I} | R^{(3,1)}(\kappa)|
    \Big ] 
    \le C R^{(3,1,1)} \cdot R^{(3,1,2)}T^{-1/2},
    \end{equation}
    where
\begin{align*}
    R^{(3,1,1)} := & \Big \{\E  \Big  [  \sup_{\kappa \in I}\|(\varepsilon_i(\kappa)/\sqrt{T} )_{1:\lfloor T\kappa\rfloor-h}\|^2 \Big  ] \Big  \}^{1/2}~,  \\
    R^{(3,1,2)} := & \Big  \{ \E \Big  [  \sup_{\kappa \in I} \|(X_i(\kappa)'[X_i(\kappa)'X_i(\kappa)]^{-1}X_i(\kappa) \varepsilon_i(\kappa))_{h+1:\lfloor T\kappa\rfloor}\|^2 \cdot  \mathbbm{1}_{\mathscr{B}^c}\Big  ]  \Big  \}^{1/2}~.
\end{align*}
Proposition \ref{prop_all_moments}, part $iii)$ implies that  $( R^{(3,1,1)} ) ^2 \le C T^{\zeta}$. 
Next we turn to $ R^{(3,1,2)}$. Notice that according to Corollary \ref{regressor_stuff}, part $ii)$
\begin{align*}
& \mathbb{E}[\|X_i(\kappa)'[X_i(\kappa)'X_i(\kappa)]^{-1}X_i(\kappa) \varepsilon_i(\kappa)\|^2\mathbbm{1}_{\mathscr{B}^c}|\mathbf{X}]  \le C,
\end{align*}
which gives  
\begin{align*}
(R^{(3,1,2)})^2 \le 
    & (\breve R^{(3,1,2)} )^2+ C ~,
    \end{align*}
{where} 
\begin{align*}
     & (\breve R^{(3,1,2)} )^2: =    \E \Big[  \sup_{\kappa \in I} \|X_i(\kappa)'[X_i(\kappa)'X_i(\kappa)]^{-1}X_i(\kappa) \varepsilon_i(\kappa)\|^2\cdot \mathbbm{1}_{\mathscr{B}^c}\\
     & \quad \quad\quad \quad\quad -\E[\|X_i(\kappa)'[X_i(\kappa)'X_i(\kappa)]^{-1}X_i(\kappa) \varepsilon_i(\kappa)\|^2\cdot \mathbbm{1}_{\mathscr{B}^c}|\mathbf{X}] \Big] \\
  & \quad \quad
   =  \E \Big[ \sup_{\kappa \in I}\Tr\big\{X_i(\kappa) [\varepsilon_i(\kappa) \varepsilon_i(\kappa)'-\E[\varepsilon_i(\kappa) \varepsilon_i(\kappa)'|\mathbf{X}] \big] X_i(\kappa)'[X_i(\kappa)'X_i(\kappa)]^{-1} \big\} \cdot \mathbbm{1}_{\mathscr{B}^c} \Big] \\
& \quad \quad   \le  C \E \Big[ \sup_{\kappa \in I} \|X_i(\kappa) \big[\varepsilon_i(\kappa) \varepsilon_i(\kappa)'-\E[\varepsilon_i(\kappa) \varepsilon_i(\kappa)'|\mathbf{X}] \big] X_i(\kappa)' \|/T \cdot \mathbbm{1}_{\mathscr{B}^c}\Big]~.
\end{align*}
Here we have used the fact hat the matrix  $ X_i(\kappa)'[ X_i(\kappa)'X_i(\kappa)]^{-1}X_i(\kappa)$ is a projection of rank $K$
and that on the event  $\mathscr{B}^c$ the matrix $[X_i(\kappa)'X_i(\kappa)/( T \kappa )]^{-1}$ is bounded
(recall the definition of the event $\mathscr{B}$ in \eqref{bad_event}). We can now define the matrix $M(\kappa) :=X_i(\kappa) \big[\varepsilon_i(\kappa) \varepsilon_i(\kappa)'-\E[\varepsilon_i(\kappa) \varepsilon_i(\kappa)'|\mathbf{X}] \big] X_i(\kappa)' $, which is entry-wise given by the partial sum process 
$$
(M(\kappa))_{l,k} =\sum_{1 \le s,t \le \lfloor T \kappa \rfloor} (x_{i,t})_k (x_{i,s})_l (\varepsilon_{i,t}\varepsilon_{i,s} - \E[ \varepsilon_{i,t}\varepsilon_{i,s} ]).
$$
Proposition \ref{prop_all_moments}, part $v)$ applied to each of the $K \times K$ entries implies for some arbitrarily small $\zeta>0$, that
$$
 \E \Big [  \sup_{\kappa \in I} \|M(\kappa)\|^2/T
 \Big ] 
 \le C T^{\zeta}.
 $$
These considerations show that $R^{(3,1,1)}$ and  $R^{(3,1,2)}$ are each of order $\mathcal{O}(T^\zeta)$  and \eqref{hd29} 
implies that $R^{(3,1)}(\kappa)$ is of order $\mathcal{O}(T^{-1/2+2\zeta})$ uniformly with respect to
$ \kappa \in I$.
The same rate can be shown for  the terms $R^{(3,2)}(\kappa)$ and $R^{(3,3)}(\kappa)$ in \eqref{hd30}. 
Observing  \eqref{hd31} and \eqref{hd32}
this implies that 
$$
\E  \Big [ \sup_{\kappa \in I} | \mathbbm{1}_{\mathscr{B}^c} \cdot R ^{(3)}(\kappa) | \Big ]\le C b \sqrt{N}/T^{1-\zeta} = o(1)
,
$$
 by the choice of $b$ (we have seen in \eqref{reference_for_b} already that even $b^2 \sqrt{N}/T^{1-\zeta} = o(1)$).
Finally, the assertion of step 4 follows by \eqref{hd6c}. This concludes the proof.

\hfill $\square$

\begin{cor} \label{regressor_stuff}
Under the assumptions of Proposition \ref{theorem4} it holds for some $C>0$  independent of $i, N$ and $T$, that
\begin{itemize}
     \item[i)] $\qquad \sup_{\kappa \in I}
\mathbbm{1}_{\mathscr{B}^c} \cdot \|X_i(\kappa) \big[ X_i(\kappa)'X_i(\kappa)\big]^{-2} X_i(\kappa)' \rfloor\| \le C/T,
$
\item[ii)] $ \qquad
 \sup_{\kappa \in I} \mathbb{E}[ \|X_i(\kappa)'[X_i(\kappa)'X_i(\kappa)]^{-1}X_i(\kappa) \varepsilon_i(\kappa)\|^2|\mathbf{X}] \le C
$
\end{itemize}

\end{cor}

\begin{proof}
We begin the proof of $i)$ with the following calculation:
\begin{align}
& \|X_i(\kappa) \big[ X_i(\kappa)'X_i(\kappa)\big]^{-2} X_i(\kappa)' \|_1 \\
 = &  \Tr \big\{X_i(\kappa) \big[ X_i(\kappa)'X_i(\kappa)\big]^{-2} X_i(\kappa)' \big\} = \Tr \big\{ \big[ X_i(\kappa)'X_i(\kappa)/(T\kappa)\big]^{-1} \big\}/(T\kappa). \nonumber
\end{align}
Here we have used that for  a symmetric matrix with non-negative eigenvalues 
the  trace norm  $\| \cdot \|_1 $ is identical to the trace.  \\
Now Proposition \ref{proposition_B} from Supplement \ref{App_B}
implies uniform boundedness of the product $\mathbbm{1}_{\mathscr{B}^c} \cdot \Tr \big\{ \big[ X_i(\kappa)'X_i(\kappa)/(T\kappa)\big]^{-1} \big\}$, which completes the proof.  \\
For part $ii)$, we can similarly rewrite the trace norm via the trace, which yields:
\begin{align*}
& \mathbb{E}[ \|X_i(\kappa)'[X_i(\kappa)'X_i(\kappa)]^{-1}X_i(\kappa) \varepsilon_i(\kappa)\|^2|\mathbf{X}] \\
& ~~~~~~~ =\mathbb{E} \big[ \Tr\{\varepsilon_i(\kappa)'X_i(\kappa)'[ X_i(\kappa)'X_i(\kappa)]^{-1}X_i(\kappa) \varepsilon_i(\kappa)\}|\mathbf{X}\big]\\
& ~~~~~~~  = \Tr \big \{ \sigma_{i,i} \Sigma_T \mathbb{E}\big[   X_i(\kappa)'[ X_i(\kappa)'X_i(\kappa)]^{-1}X_i(\kappa) |\mathbf{X}\big] \big \} \\
& ~~~~~~~   \le \sigma_{i,i}\|\Sigma_T\|_\infty
   \big\|  X_i(\kappa)'[ X_i(\kappa)'X_i(\kappa)]^{-1}X_i(\kappa) \big \|_1   \le C.
\end{align*}
In the last step we have used that $X_i(\kappa)'[ X_i(\kappa)'X_i(\kappa)]^{-1}X_i(\kappa)$ is a projection matrix of rank $\le K$, and thus the trace norm is bounded by $K$. 

\end{proof}

\subsection{Proof of Proposition \ref{theorem5}}
Proposition \ref{prop_replacement_E} implies that uniformly in $\kappa$
\begin{equation} \label{rep_e_1}
    E_1(\kappa) = \tilde E_1 (\kappa) +o_P(1),
\end{equation}
where 
\begin{equation} \label{def_tilde_E_1} 
\tilde E_1(\kappa) := \frac{2 }{\sqrt{NT}}\sum_{i=1}^{N}  \varepsilon_{i}(\kappa)' X_i(\kappa)  Q_i^{-1}  [\beta_i-\bar \beta].
\end{equation}
Therefore the weak convergence of
$\{ E_1(\kappa)  \}_{\kappa \in I} $ follows 
from the weak convergence of the process 
$\{ \tilde  E_1(\kappa)  \}_{\kappa \in I} $. 
Notice that we can write $\tilde E_1$ as a partial sum process as follows:
$$
\tilde E_1(\kappa) =\frac{2}{\sqrt{NT}} \sum_{t=1}^{\lfloor T \kappa \rfloor } \sum_{i=1}^N \varepsilon_{i,t} x_{i,t}'  Q_i^{-1} [\bar \beta_i-\beta]
.$$
We now conclude a weak invariance principle for the stochastic process $\tilde E_1(\kappa)/\tilde \LR_N$, where $\tilde \LR_N^2:=\mathbb{E}[(\tilde E_1(1))^2]$ is the variance.
We therefore employ Theorem 2.7 in \cite{hafouta2021}. 
The assumptions of this theorem can be checked as follows: Assumption 2.1 is met with $q=4$, $p>2$ but sufficiently close to $2$, $B_n=1$ and $\gamma(j) := Cj^{-a(1/2-1/q)}$ (where $C, a$ are chosen in our Assumption \ref{assumption1} $(\epsilon), 5.$). Notice that the $q$th moment of our random variables $\frac{1}{\sqrt{N}}\sum_{i=1}^N \varepsilon_{i,t} x_{i,t}$ are uniformly bounded according to Proposition \ref{moments_for_mixing} in Supplement \ref{App_D}, where we choose $\phi=c=4$, $\chi = 2M-4$, $d=1$ and summability is satisfied by our Assumption $(\varepsilon), (5)$. Accordingly, $A_n$ (defined in their eq. (2.5)) is uniformly bounded and Assumption 2.3 follows directly ($\beta_n \le C$ for some fixed $C$ and all $n$ - where $\beta_n$ here corresponds to the notation in \cite{hafouta2021}). Then Theorem 2.7 in \cite{hafouta2021} entails 
\begin{equation}\label{Eq_unif_invariance}
d_P(\tilde E_1/\tilde \LR_N, \mathbb{B}) \le C \delta_N,
\end{equation}
for some sequence $\delta_N \to 0$ and $C$ independent of $\tilde \LR_N$. Here $d_P(\cdot, \cdot)$, denotes the Prokhorov metric on the Skorohod space $D[0,1]$, which is equipped with the uniform norm and thus a subspace of the bounded functions on the unit interval. We recall the definition of the Prokhorov metric for two probability distributions $P_1, P_2$ as
\begin{align}\label{Eq_def_prokhorov}
d_P(P_1, P_2) := & \inf \big\{r>0: P_1(B) \le P_2(B^r) +r,\\
& \qquad  \qquad \,\,\quad P_2(B) \le P_1(B^r) +r,\forall \,\, B \,\, \textnormal{Borel set}  \big\}, \nonumber
\end{align}
where $B^r$ is the open $r$-environment of $B$. Defining for random variables $X \sim P_1, Y \sim P_2$ the distance  $d_P(X, Y):=d_P(P_1, P_2)$ explains the expression $d_P(\tilde E_1/\tilde \LR_N, \mathbb{B})$. We refer the reader for details on the Skorohod space to \cite{pollard} and for the Prokhorov metric to \cite{prokhorov}.

Finally, we notice that $\tilde \lambda_N^2$ converges to the long run variance $\lambda^2$ (defined in \eqref{long_run_variance}). To see this notice that  according to Assumptions $(X)$, $(1)$ and $(\LR)$
\begin{align} \label{Eq_conv_LR}
\mathbb{E} \big[(\tilde E_1(1))^2 \big]&= \frac{4 }{N} \sum_{i,j=1}^N \sigma_{i,j}[\beta_i-\bar \beta]' Q_i^{-1}\mathbb{E}\Big[\frac{X_i ' \Sigma_{T} X_j }{ T } \Big]Q_j^{-1} [\beta_j-\bar \beta]\\
 = & \frac{4 }{N} \sum_{i,j=1}^N \sigma_{i,j}[\beta_i-\bar \beta]' Q_i^{-1} U_{i,j} Q_j^{-1} [\beta_j-\bar \beta] +o(1) \to  \LR^2. \nonumber
\end{align}

\subsection{Proof of Theorem \ref{theorem_test}} \label{proof_of_test}

Recall the definition of the class $\mathcal{T}$ in \eqref{Eq_def_class_T}, of the variance $\lambda^2_N(\bsec )$ in \eqref{ Eq_def_lambdauni}, of the process $\tilde E_1$ in \eqref{def_tilde_E_1} and of $\tilde \LR_N^2(\bsec ):=\mathbb{E}[(\tilde E_1(1))^2]$. Here $\tilde E_1$ depends on the sequence of slopes $\bsec =(\beta_n)_{n \in \N}$, even though this is not explicit in the notation. Notice  that the variance $\LR^2_N(\bsec )$ is asymptotically identical to $\tilde \LR_N^2(\bsec )$ according to \eqref{Eq_conv_LR}. More precisely we have 
$$
\sup_{\bsec  \in \mathcal{T}} \frac{|\LR^2_N(\bsec )-\tilde \LR^2_N(\bsec )|}{\LR^2_N(\bsec )} =o(1),
$$
where we have used the boundedness of $\LR^2_N(\bsec )$ from below. This implies, by definition of the class $\mathcal{T}$, that $\tilde \LR^2_N(\bsec ) \ge c/2>0$ for sufficiently large $N$, simultaneously for all $\bsec  \in \mathcal{T}$. 
Finally, recall the definition of the Prokhorov metric in \eqref{Eq_def_prokhorov}, which we denote by $d_P$, no matter on which metric space the probability measures are defined (this will always be clear from context). In the case of real valued random variables, $d_P$ is an upper bound for the better known L\'evy metric on the real line.\\
Now consider the following map $\Psi$ defined for a bounded function $f$ on the unit interval as
$$
\Psi(f) :=  \frac{f(1)}{\Big\{\int_I \kappa^2 (f(\kappa)-\kappa f(1))^2 d\nu(\kappa)\Big\}^{1/2}}.
$$
To avoid ill-defined cases, we set $\Psi(f)=0$, if the denominator is $0$. We also define a subspace $D^\rho[0,1]$ of $D[0,1]$ as those functions, which satisfy $\sup_{\kappa \in [0,1]} |f(\kappa)| \le 1/\rho$ and $\big\{\int_I \kappa^2 (f(\kappa)-\kappa f(1))^2 d\nu(\kappa)\big\}^{1/2}\ge \rho$, where $\rho >0$ is fixed. Notice that restricted on $D^\rho[0,1]$ the map $\Psi$ is Lipschitz continuous, where the Lipschitz constant $C_\rho$ grows as $\rho$ decreases. It is clear that $\lim_{\rho \downarrow 0} \mathbb{P}(\mathbb{B} \in D^\rho[0,1])=1$. Similarly, it holds that 
\begin{equation} \label{Eq_E_process_rho_space}
\lim_{\rho \downarrow 0}\lim_{N \to \infty }\inf_{\bsec  \in \mathcal{T}} \mathbb{P}(\tilde E_1/\tilde \lambda_N(\bsec ) \in D^\rho[0,1])=1.
\end{equation}
To see this fact, we check that 
\begin{equation} \label{Eq_E_process_rho}
\lim_{\rho \downarrow 0}\lim_{N \to \infty }\inf_{\bsec  \in \mathcal{T}} \mathbb{P}(\sup_{\kappa \in [0,1]} |\tilde E_1(\kappa)|/\tilde \lambda_N(\bsec ) \le 1/\rho)=1.
\end{equation}
We first notice that
\begin{align*}
& \mathbb{P}(\sup_{\kappa \in [0,1]} |\tilde E_1(\kappa)|/\tilde \lambda_N(\bsec ) \le 1/\rho)-\mathbb{P}(\sup_{\kappa \in [0,1]} |\mathbb{B}(\kappa)| \le 1/\rho)\\
\le & C d_P\big(\sup_{\kappa \in [0,1]} |\tilde E_1(\kappa)|/\tilde \lambda_N(\bsec ), \sup_{\kappa \in [0,1]} |\mathbb{B}(\kappa)|\big).
\end{align*}
Here we have used that the (uniform) distance between distribution functions is bounded by some constant times the L\'evy metric, if either of them  is Lipschitz continuous (in our case that of $\sup_{\kappa \in [0,1]} |\mathbb{B}(\kappa)|$) and the L\'evy metric is bounded by the Prokhorov metric. Now, as
the map $f \mapsto \sup_{\kappa \in [0,1]} |f(\kappa)|$ is Lipschitz continuous with constant $1$,  Lemma \ref{lemma_prokhorov} implies that $
d_P(\sup_{\kappa \in [0,1]} |\tilde E_1(\kappa)|/\tilde \lambda_N(\bsec ), \sup_{\kappa \in [0,1]} |\mathbb{B}(\kappa)|) \le d_P(\tilde E_1, \mathbb{B})$, which in turn is bounded by $C \delta_N$ according to \eqref{Eq_unif_invariance}.
As a consequence the left side of \eqref{Eq_E_process_rho} equals
$$
\lim_{\rho \downarrow 0}\lim_{N \to \infty } \mathbb{P}\big(\sup_{\kappa \in [0,1]} |\mathbb{B}(\kappa)| \le 1/\rho\big) + \mathcal{O}(\delta_N) = \lim_{\rho \downarrow 0} \mathbb{P}\big(\sup_{\kappa \in [0,1]} |\mathbb{B}(\kappa)| \le 1/\rho\big)=1.
$$
Using similar arguments for the integral condition shows  \eqref{Eq_E_process_rho_space}. 

We now investigate the self-normalized statistic $\hat W_N = \Psi(\{\kappa \sqrt{ N   T  }(\tilde S_N(\kappa)- \Delta)\}_{\kappa \in I})$ (which as we recall also depends on the slope sequence $\bsec =(\beta_n)_{n \in \mathbb{N}}$).  Introducing the conditional probability measure
$$
\mathbb{P}_\rho(\cdot):=\mathbb{P}(\cdot|\{\tilde E_1/\tilde \lambda_N(\bsec )\in D^\rho[0,1]\} \, \forall \bsec  \in \mathcal{T}, \mathbb{B} \in D^\rho[0,1]\})
$$
gives for any $y \in \R$
\begin{align} \label{Eq_mop_up_terms_1}
    &\mathbb{P}_\rho\big(\hat W_N \le y\big) = \mathbb{P}_\rho\big(\Psi(\{\kappa \sqrt{ N   T  }(\tilde S_N(\kappa)- \Delta)\}_{\kappa \in I}) \le y\big)\nonumber \\
    = & \mathbb{P}_\rho\big(\Psi(\{[\tilde E_1(\kappa) +\sqrt{NT}(S_N-\Delta)]/\tilde \lambda_N(\bsec )\}_{\kappa \in I}) \le y\big)+o(1).
\end{align}
Here we have used Propositions \ref{theorem1} and \ref{theorem4}, together with the fact \eqref{rep_e_1} and the Lipschitz continuity of $\Psi$ restricted on $D^\rho[0,1]$ to get the second equality. Note that the remainder does not depend on $\bsec $, because all vanishing terms in the cited propositions are uniformly bounded in $\bsec $. Going from a vanishing term inside the distribution function to an $o(1)$ term requires the distribution function of $\Psi(\{[\tilde E_1(\kappa) +\sqrt{NT}(S_N-\Delta)]/\tilde \lambda_N(\bsec )\}_{\kappa \in I})$ to be (asymptotically) continuous, which it is, as we will see below. Finally, we have also employed the fact that for any real number $x \neq 0$ and any function $f$ we have $\Psi(f)=\Psi(f/x)$, to divide by the standard deviation $\tilde \lambda_N(\bsec )$, which is a consequence of the self-normalizing structure of the functional $\Psi$.

Now we turn to the convergence of the distribution function of the real valued, random variable $\Psi(\{[\tilde E_1(\kappa) +\sqrt{NT}(S_N-\Delta)]/\tilde \lambda_N(\bsec )\}_{\kappa \in I})$. For the moment, let us set $S_N-\Delta=x/\sqrt{NT}$. The random variable
$\Psi(\{\mathbb{B}(\kappa)+x/\lambda(\bsec )\}_{\kappa \in I})$
has a continuous density and thus a uniformly continuous distribution function. In the following, we denote by $\mathbb{P}^X$ the measure induced by a random variable $X$ and by $\mathbb{P}^X_\rho$ the conditional image measure. It follows that  
\begin{align}  \label{Eq_mop_up_terms_2}
    & \sup_{y \in \R}\big|\mathbb{P}_\rho\big(\Psi(\{[\tilde E_1(\kappa) +x]/\tilde \lambda_N(\bsec )\}_{\kappa \in I}) \le y\big)-\mathbb{P}_\rho\big(\Psi(\{\mathbb{B}(\kappa)+x/ \lambda(\bsec )\}_{\kappa \in I})  \le y\big)\big|\nonumber\\
    =& \sup_{y \in \R}\big|\mathbb{P}\big(\Psi(\{[\tilde E_1(\kappa) +x]/\tilde \lambda_N(\bsec )\}_{\kappa \in I}) \le y\big)-\mathbb{P}\big(\Psi(\{\mathbb{B}(\kappa)+x/ \lambda(\bsec )\}_{\kappa \in I})  \le y\big)\big| \nonumber\\
    & +Rem^{(1)}(\rho)\nonumber\\
    \le & C  d_P(\mathbb{P}^{\Psi(\{[\tilde E_1(\kappa) +x]/\tilde \lambda_N(\bsec )\}_{\kappa \in I})},\mathbb{P}^{ \Psi(\{\mathbb{B}(\kappa)+x/ \lambda(\bsec )\}_{\kappa \in I})}+Rem^{(1)}(\rho)\nonumber\\
    \le & C  d_P(\mathbb{P}_\rho^{\Psi(\{[\tilde E_1(\kappa) +x]/\tilde \lambda_N(\bsec )\}_{\kappa \in I})},\mathbb{P}_\rho^{ \Psi(\{\mathbb{B}(\kappa)+x/ \lambda(\bsec )\}_{\kappa \in I})}+Rem^{(1)}(\rho)\nonumber\\
    \le & C C_\rho d_P(\mathbb{P}_\rho^{\tilde E_1/\tilde \LR_N(\bsec )}, \mathbb{P}_\rho^{\mathbb{B}}) +Rem^{(1)}(\rho)\nonumber\\
    \le & C C_\rho d_P(\mathbb{P}^{\tilde E_1/\tilde \LR_N(\bsec )}, \mathbb{P}^{\mathbb{B}})+Rem^{(1)}(\rho) \le C \delta_N C_\rho+Rem^{(1)}(\rho).
\end{align}
Here $Rem^{(1)}(\rho)$ is a remainder that depends only on $\rho$, which for parsimony of notation, we have allowed to change from one line to the next. Importantly $Rem^{(1)}(\rho)\to 0$ as $\rho \to 0$. In the first equality, we have used proximity of the true and conditional measures. In the first inequality, we have used that the Prokhorov metric upper bounds the uniform distance of the distribution functions (except for some factor), if either of them is Lipschitz continuous (in our case the distribution function involving $\mathbb{B}$).  
From the third to the fourth line we have again used proximity of conditional and unconditional measures (increasing $Rem^{(1)}(\rho)$ in the process) and from the fourth to the fifth line Lipschitz continuity of $\Psi$, together with Lemma \ref{lemma_prokhorov}. We then switch back to the conditional measures (increasing $Rem^{(1)}(\rho)$) and finally use the inequality \eqref{Eq_unif_invariance}, which holds for the unconditional measure. \\
Let us now consider the implications of our above arguments: Recalling the definition of the class of local alternatives as $\mathcal{A}_N(x):= \mathcal{T} \cap \{\bsec : S_N(\bsec )-x/\sqrt{NT}\ge\Delta \}$ we have for any $\bsec  \in \mathcal{A}_N(x)$ the lower bound
$$
\mathbb{P}\big(\hat W_N >q_{1-\alpha}\big) \ge  \mathbb{P}\big(\Psi(\{\kappa \sqrt{ N   T  }(\tilde S_N(\kappa)- S_N + x/\sqrt{NT})\}_{\kappa \in I})  >q_{1-\alpha}\big).
$$
By our above derivations, it follows that
\begin{align} \label{Eq_f_lower_bound}
f(x) := & \liminf_{N \to \infty} \inf_{\bsec  \in \mathcal{A}_N(x)}\mathbb{P}\big(\hat W_N >q_{1-\alpha}\big) \\
\ge &\lim_{\rho \downarrow 0 }\liminf_{N \to \infty} \inf_{\bsec  \in \mathcal{A}_N(x)}\mathbb{P}_\rho\big(\Psi(\mathbb{B}+x/ \lambda(\bsec ))>q_{1-\alpha}\big) \nonumber\\
&+\lim_{\rho \downarrow 0 }\liminf_{N \to \infty} \inf_{\bsec  \in \mathcal{A}_N(x)}\big\{ Rem^{(1)}( \rho)+Rem^{(2)}(N, \rho)\big\}. \nonumber
\end{align}
Here the remainder $Rem^{(1)}( \rho)$ has been established above, where we have noticed that it vanishes as $\rho \to 0$. The remainder $Rem^{(2)}(N, \rho)$ consists of those remainder terms, which are not included in $Rem^{(1)}( \rho)$, specifically the $o(1)$ term on the right of \eqref{Eq_mop_up_terms_1} and the $CC_\rho \delta_N$ in \eqref{Eq_mop_up_terms_2}, both of which vanish as $N \to \infty$ for any fixed $\rho >0$. Now, we can make the probability of the event $\{\Psi(\mathbb{B}+x/\tilde \lambda(\bsec ))>q_{1-\alpha}\}$ smallest by maximizing $\tilde \lambda(\bsec )$, which yields a finite value $\lambda_{max}$ given the uniform boundedness of slopes in $\mathcal{T}$. As $\rho \to 0$ we have $$
\mathbb{P}_\rho\big(\Psi(\mathbb{B}+x/\lambda_{max})>q_{1-\alpha}\big) \to \mathbb{P}\big(\Psi(\mathbb{B}+x/\lambda_{max})>q_{1-\alpha}\big). 
$$
The probability on the right is $>\alpha$ and
goes to $1$ as $x \to \infty$. Consequently, $f$ inherits these two properties according to the lower bound \eqref{Eq_f_lower_bound}. The fact that $f$ is monotonically increasing follows directly, since $\mathcal{A}_N(x) \supset \mathcal{A}_N(y)$ if $x<y$ by definition of the class of local alternatives. This completes the proof of asymptotic consistency. Analogous arguments show asymptotic level $\alpha$ which we omit to avoid redundancy. We conclude  by proving a technical Lemma, which we have repeatedly used in our proof. 

\begin{lem} \label{lemma_prokhorov}
Let $(\mathcal{X}, d_\mathcal{X})$ and $(\mathcal{Y}, d_\mathcal{Y})$ be two metric spaces and $g:\mathcal{X} \to \mathcal{Y}$ be Lipschitz continuous with constant $C$. Then it holds for any two probability measures $P_1, P_2$ on $(\mathcal{X}, d_\mathcal{X})$ (defined on the Borel-$\sigma$-algebra), that
$$
d_P(P_1^g, P_2^g) \le (C \lor 1) d_P(P_1, P_2),
$$
where $P_1^g, P_2^g$ denote the image measures and $d_P$  the Prokhorov metric.
\end{lem}

The proof of the Lemma is fundamental. For simplicity, we assume that $C\ge 1$. We then have by set inclusion \begin{align*}
    &\{r>0: P_i^g(B) \le P_j^g(B^r) +r, i,j \in \{1,2\},  \forall \,\, B \,\, \textnormal{Borel set}  \big\}\\
    \supset &\big\{r>0: P_i(g^{-1}(B)) \le P_j(g^{-1}(B)^{r/C}) +r/C, i,j \in \{1,2\},  \forall \,\, B \,\, \textnormal{Borel set}  \big\}\\
    \supset & \big\{r/C>0: P_i(B) \le P_j(B^{r}) +r, i,j \in \{1,2\},  \forall \,\, B \,\, \textnormal{Borel set}  \big\}
\end{align*}

In the first inclusion, we have used the fact that $g^{-1}(B)^{r/C} \subset g^{-1}(B^r)$, which follows by Lipschitz continuity. Using the definition of the Prokhorov metric in \eqref{Eq_def_prokhorov} completes the proof.

\section{Technical Details }\label{App_B}

The first part of this section is concerned with the convergence of the regressor-matrices. In Proposition \ref{proposition_regressors} we show that the sequential estimates  $(X_i( \kappa)'X_i( \kappa)/\lfloor T \kappa\rfloor )$ are uniformly close to their limits $Q_i$. In Proposition \ref{proposition_B} we strengthen this result by showing closeness with probability converging to $1$ and finally  Proposition \ref{proposition_inverse} establishes proximity of the inverse matrices. 

\subsection{Bounds for the Regressors}

\begin{prop}  \label{proposition_regressors}
If Assumptions $(X)$ $(1), (2)$ and $(\varepsilon)$, $(5)$ hold, it follows for some  constant $C_1$ independent of $i, N, T$, that
$$\mathbb{E}\Big[\sup_{\kappa \in I}\|X_i( \kappa)'X_i( \kappa)/\lfloor T \kappa\rfloor -Q_i\|^{2 \eta \lor 2} \Big] \le C_1 /T^{\eta \lor 1}.$$
\end{prop}

\begin{proof} 
We confine ourselves to the case $\eta > 1$. The case $\eta\le 1$ follows by the same method.

First, notice that we can replace $Q_i$ by the expectation of $X_i( \kappa)'X_i( \kappa)/\lfloor T \kappa\rfloor$ due to the bias assumption ($(X), (1)$):
$$
\mathbb{E}\Big[\sup_{\kappa \in I}\|X_i( \kappa)'X_i( \kappa)/\lfloor T \kappa\rfloor -Q_i\|^{2 \eta } \Big]\le F_1+ F_2~,
 $$ 
 where
   \begin{align}
    F_1 := & C \mathbb{E}\Big[\sup_{\kappa \in I}\|(X_i( \kappa)'X_i( \kappa)/\lfloor T \kappa\rfloor )-\E \big[X_i( \kappa)'X_i( \kappa)/\lfloor T \kappa\rfloor \big]\|^{2 \eta }\Big] ~, \nonumber  \\
    F_2 := & C \sup_{\kappa \in I}\big\|Q_i-\E \big[ X_i( \kappa)'X_i( \kappa)/\lfloor T \kappa\rfloor \big] \big\|^{2 \eta } .\nonumber 
\end{align}
Here $C$ only depends on $\eta$. The bias assumption implies that 
$$
\|Q_i-\E (X_i( \kappa)'X_i( \kappa)/\lfloor T \kappa\rfloor )\| \le C/\sqrt{\lfloor T \kappa\rfloor }
$$
for some universal constant $C$ and thus (as $\kappa$ is bounded away from $0$) $F_2 = \mathcal{O}(1/T^{ \eta })$.  Next we consider $F_1$.
In order to upper bound it, we want to apply Theorem 3.1 from \cite{Moricz}.
For this purpose we have to show for arbitrary $1 \le L < U \le T$ that
$$
\E \Big[ \Big\| \sum_{t=L}^{U} x_{i,t}x_{i,t}'-\E x_{i,t}x_{i,t}'\Big\|^{2\eta} \Big] \le C |U-L|^{\eta}.
$$
Notice that in the norm on the left we have a $K \times K$ matrix for fixed $K$. We can show the desired bound entry-wise (and subsequently it follows for the whole matrix with possibly larger constant depending on $K$). For simplicity, we confine ourselves to the first entry $(1,1)$. Denoting the first entry from the vector $x_{i,t}$ by $(x_{i,t})_1$, we thus have to show
\begin{equation} \label{moricz_for_covariance}
\E \Big( \sum_{t=L}^{U} (x_{i,t})_1^2-\E \Big[ (x_{i,t})_1^2\Big)^{2\eta}\Big] \le C |U-L|^{\eta}.
\end{equation}
The sequence $(x_{i,1})_1^2-\E(x_{i,1})_1^2, (x_{i,2})_1^2-\E(x_{i,2})_1^2,...$ is by assumption one-dimensional $\alpha$-mixing (Assumption $(\varepsilon), (5)$) with existing moments of order $M$ (Assumption $(X)$, $(2)$). We can apply Proposition \ref{moments_for_mixing} with $d=1, c=4, \phi=2\eta, \chi=M-2\eta$ (notice that the summability condition  \eqref{hd3} is fulfilled due to  Assumption $(\varepsilon), (5)$), which yields \eqref{moricz_for_covariance} with some universal constant. By Theorem 3.1 from \cite{Moricz} the assertion follows.\\
\end{proof}

\begin{prop} \label{proposition_B}
Suppose that  Assumptions  $(N), (X)$ $(1)- (3)$ and $(\varepsilon)$, $(5)$ hold and let $ C_2:=(\max_{i=1,...,N} \|Q_i^{-1}\|_\infty)^{-1} $ (see Assumption $(X), (3)$).
Define
\begin{equation*}
\mathscr{B}:=\bigcup_{i=1}^N \{\sup_{\kappa \in I} \|X_i(\kappa)'X_i(\kappa)/\lfloor T \kappa \rfloor-Q_i\|_\infty \ge   C_2/2  \Big \}
\end{equation*}
then
$
\mathbb{P}(\mathscr{B}) \to 0  
$
as $N,T \to \infty$.

\end{prop}
This proposition in particular implies that 
with probability converging to $1$ all matrices $X_i(\kappa)'X_i(\kappa)/\lfloor T \kappa \rfloor$ (for all $i=1,...,N$ and $\kappa \in I$) are positive definite. 

\begin{proof}
The probability, that the event $\mathscr{B}$  occurs can be upper bounded by the union bound and subsequently by the Markov inequality, which yields 
\begin{align*}
 \mathbb{P}( \mathscr{B})\le& N \max_{i=1,...,N} \mathbb{P}(\sup_{\kappa \in I}\|X_i(\kappa)'X_i(\kappa)/\lfloor T \kappa \rfloor-Q_i\|_\infty\ge C_2/ 2)\\
\le & \frac{N   \max_{i=1,...,N} \mathbb{E}\big[\sup_{\kappa \in I}\|(X_i( \kappa)'X_i( \kappa)/\lfloor T \kappa\rfloor )-Q_i\|_\infty^{2\eta}\big]}{(C_2/2)^{2\eta}}=\mathcal{O}(N/T^{\eta}),
\end{align*}
which is of order $o(1)$ by Assumption $(N)$. 
Notice that for the last equality, we have used Proposition \ref{proposition_regressors}.\\
\end{proof}

\begin{prop} \label{proposition_inverse}
Under the assumptions of Proposition \ref{proposition_B}
$$
\mathbb{E} \Big [ 
\mathbbm{1}_{\mathscr{B}^c} \cdot \sup_{\kappa \in I} \big\| \big[ X_i(\kappa)'X_i(\kappa)/(T\kappa) \big]^{-1}-Q_i^{-1}\big\|^2 
\Big  ] \le C/T
$$
for some constant $C$ independent of $i$. In particular 
$$
\sup_{\kappa \in I} \big\| \big[ X_i(\kappa)'X_i(\kappa)/(T\kappa) \big]^{-1}-Q_i^{-1}\big\| = \mathcal{O}_P(1/\sqrt{T}).
$$
\end{prop}

\begin{proof}
We introduce the notation $H_i (\kappa) :=Q_i-X_i(\kappa)'X_i(\kappa)/\lfloor T \kappa \rfloor$ and note that we have on the event $\mathscr{B}^c$ (see  \eqref{bad_event})
$$
\|Q_i^{-1} H_i(\kappa)\|_\infty \le \|Q_i^{-1}\|_\infty\| H_i(\kappa)\|_\infty\le \max_{i=1,...,N} \|Q_i^{-1}\|_\infty C_2/2 \le 1/2~.
$$
Therefore Neumann's expansion gives 
$$
[X_i(\kappa)'X_i(\kappa)/\lfloor T \kappa \rfloor]^{-1}= [Id-Q_i^{-1}H_i(\kappa)]^{-1}Q_i^{-1}= \sum_{\ell \ge 0}(-1)^\ell (Q_i^{-1} H_i(\kappa))^\ell  Q_i^{-1}~,
$$
where the tail for $\ell \ge 2$ can be bounded as follows
\begin{align*}
\Big\|\sum_{\ell \ge 2}(-1)^\ell (Q_i^{-1} H_i(\kappa))^\ell Q_i^{-1} \Big\| & \le \|Q_i^{-1}\|^3 \|H_i(\kappa)\|^2 \sum_{\ell \ge 0} \|Q_i^{-1} H_i(\kappa)\|^\ell \\
& \le \|H_i(\kappa)\|^2  \max_{j=1,...,N}\|Q_j^{-1}\|^3\sum_{\ell \ge 0}(1/2)^\ell\le \|H_i(\kappa)\|^2  C.
\end{align*}
The constant $C$ in this equation is independent of $i$ (as $\max_j \|Q_j^{-1}\|$ is bounded by Assumption $(X), (3)$). It follows that
\begin{align*}
& \mathbb{E} \Big [ 
\mathbbm{1}_{\mathscr{B}^c}\cdot \sup_{\kappa \in I} \big\| \big[ X_i(\kappa)'X_i(\kappa)/(T\kappa) \big]^{-1}-Q_i^{-1}\big\|^2  \Big] 
\\
 \le & \mathbb{E} \Big [ 
\sup_{\kappa \in I} \big\| -Q_i^{-1} H_i(\kappa) Q_i^{-1}\big\|^2
\Big  ] 
+C\mathbb{E} \Big [ \sup_{\kappa \in I}\|H_i(\kappa)\|^2 \Big ]  \le  {C \over T} ~,
\end{align*}
where the last inequality follows from Proposition \ref{proposition_regressors}, which implies
$
\mathbb{E} \big [  \sup_{\kappa \in I} \|H_i(\kappa)\|^2 \big  ] \le C/T.
$
\end{proof}

\section{Bounds for Remainders} \label{App_C}

In this section we employ the results from Propositions \ref{proposition_regressors} and \ref{proposition_inverse}  to show the proximity of the remainders $E_j$ to $\tilde E_j$ ($j=1,2,3$), where $E_1, E_2, E_3$ 
are defined in \eqref{e1} - \eqref{e3}
 and $\tilde E_1, \tilde E_2, \tilde E_3$ in \eqref{def_tilde_E_1}, \eqref{def_tilde_E_2} and \eqref{def_tilde_E_3} respectively.

\begin{prop} \label{prop_replacement_E}
Under the assumptions of Proposition \ref{theorem1} it holds that
$$
\sup_{\kappa \in I} |E_j(\kappa) - \tilde E_j(\kappa)| =o_P(1), \qquad j=1,2,3,
$$
where $E_1, E_2, E_3$ are defined in \eqref{e1} - \eqref{e3} and $\tilde E_1, \tilde E_2, \tilde E_3$ in \eqref{def_tilde_E_1}, \eqref{def_tilde_E_2} and \eqref{def_tilde_E_3} respectively. Moreover 
$$
\sup_{\kappa \in I} |\E [E_2(\kappa)|\mathbf{X}] - \E [\tilde E_2(\kappa)|\mathbf{X}]| =o_P(1).
$$

\end{prop}

\begin{proof}
We show the identity for each term individually, beginning with $E_3$ and ending with $E_1$ (which reflects the actual order in the previous proofs).\\

Recall the definition of $E_3$ in \eqref{h1}, which implies
\begin{equation}\label{R_N_outer}
E_3(\kappa) = - \frac{1}{\sqrt{TN}} \|\bar E_3(\kappa)\|^2,
\end{equation}
where
$$
\bar E_3(\kappa) := \frac{1}{\sqrt{N}}\sum_{i=1}^{N} \varepsilon_{i}(\kappa)X_i(\kappa)/\sqrt{T\kappa} \big[ X_i(\kappa)'X_i(\kappa)/(T\kappa) \big]^{-1}.
$$
Furthermore recall the definitions of $\tilde E_3$ and $\breve E_3$ in \eqref{def_tilde_E_3} and \eqref{def_breve_E_3} respectively. A simple calculation shows that
$$
E_3(\kappa) = - \frac{1}{\sqrt{TN}} \|\breve E_3(\kappa)+( \bar  E_3(\kappa) - \breve E_3(\kappa))\|^2 = \tilde E_3(\kappa) + Rem,
$$
where 
$$
Rem = \mathcal{O}_P \left({\|\breve E_3(\kappa)\| \cdot \|\bar  E_3(\kappa) - \breve E_3(\kappa)\| \over \sqrt{TN}}\right)+ \mathcal{O}_P \left({\|\bar  E_3(\kappa) - \breve E_3(\kappa)\|^2 \over \sqrt{TN}}\right)
$$
We can now complete the proof by establishing uniformly in $\kappa$ that 
 $\|\bar  E_3(\kappa) - \breve E_3(\kappa)\| = o_P(N^{1/4})$ (also recall that $\|\breve E_3(\kappa)\|$ is uniformly of order $\mathcal{O}_P(T^\zeta)$, where $\zeta$
can be made arbitrarily small
(see Proposition \ref{prop_all_moments} $iv)$).
We first notice that
$$
\|\bar  E_3(\kappa) - \breve E_3(\kappa)\| = \|\bar  E_3(\kappa) - \breve E_3(\kappa)\| \cdot \mathbbm{1}_{\mathscr{B}^c} +o_P(1),
$$
according to Proposition \ref{proposition_B}. We further analyze the non-vanishing part on the right as follows
\begin{align*}
& \E \Big[\sup_{\kappa \in I}  \|\bar  E_3(\kappa) - \breve E_3(\kappa)\| \cdot \mathbbm{1}_{\mathscr{B}^c}\Big] \\
\le & \frac{1}{\sqrt{N}}\sum_{i=1}^{N}\E \Big[\sup_{\kappa \in I}  \left\| \varepsilon_{i}(\kappa)X_i(\kappa)/\sqrt{T\kappa}\right\| \Big\|\big[ X_i(\kappa)'X_i(\kappa)/(T\kappa) \big]^{-1}-Q_i^{-1} \Big\| \cdot  \mathbbm{1}_{\mathscr{B}^c}\Big]  \nonumber
\end{align*}
The right side is upper bounded by $\frac{1}{\sqrt{N}}\sum_{i=1}^N E^{(3,1)}_i \cdot E^{(3,2)}$, where
\begin{align*}
E^{(3,1)}_i := & \bigg\{ \E \Big[\sup_{\kappa \in I}  \left\| \varepsilon_{i}(\kappa)X_i(\kappa)/\sqrt{T\kappa}\right\|^2 \Big]\bigg\}^{1/2}\\
E^{(3,2)}_i := & \bigg\{\mathbb{E}\Big[\sup_{\kappa \in I}  \mathbbm{1}_{\mathscr{B}^c}  \cdot  \Big\|\big[ X_i(\kappa)'X_i(\kappa)/(T\kappa) \big]^{-1}-Q_i^{-1} \Big\|^2\Big]\bigg\}^{1/2}. 
\end{align*}
According to Proposition \ref{prop_all_moments}, part $ii)$ $E^{(3,1)}_i \le C T^\zeta$ for a $\zeta>0$, which can be chosen arbitrarily small and a $C$, only depending on $\zeta$ (and in particular not on $i$). Moreover, $E^{(3,2)}_i \le C$ for some $C$ (again independent of $i$) according to Proposition \ref{proposition_inverse}. These observations directly imply
$$
 \E \Big[\sup_{\kappa \in I}  \|\bar  E_3(\kappa) - \breve E_3(\kappa)\| \cdot \mathbbm{1}_{\mathscr{B}^c}\Big] = \mathcal{O}(\sqrt{N/T} T^{\zeta}).
$$ 
Finally, to see that the right side is of order $o(N^{1/4})$ consider the simple calculation
$$
{(\sqrt{N/T} T^{\zeta})^2 \over \sqrt{N}}= \frac{\sqrt{N}}{T^{1-2\zeta}}. 
$$
For $\zeta>0$ sufficiently small, $1-2\zeta>\eta/2$ and the right side is of order $o(1)$   according to Assumption $(N)$.\\

The proof for $E_2$ and its conditional expectation is conducted by similar techniques and is therefore omitted.\\

The proof for $E_1$ is more intricate than for the previous terms (as the remainders in the analysis are of larger magnitude). Let us first define the difference process
\begin{align*}  
 D  (\kappa):= E_1(\kappa) - \tilde E_1(\kappa)=  \frac{2 }{\sqrt{TN}}\sum_{i=1}^{N}  \varepsilon_{i}(\kappa)' X_i(\kappa) \big\{\big[ X_i(\kappa)'X_i(\kappa)/(T \kappa)\big]^{-1}-Q_i^{-1}\big\} [\beta_i-\bar \beta].
\end{align*}
According to Proposition \ref{proposition_B} $ D  (\kappa) =  D  (\kappa) \cdot \mathbbm{1}_{\mathscr{B}^c} + o_P(1)$ uniformly in $\kappa$. 
By the same techniques as in the proof of Proposition \ref{proposition_inverse} (that is Neumann expansion of the difference $[X_i(\kappa)'X_i(\kappa)/(T \kappa)]^{-1}-Q_i^{-1}$, 
 on the set $\mathscr{B}^c$) we can rewrite
\begin{align*} 
D  (\kappa) \cdot \mathbbm{1}_{\mathscr{B}^c} =&  D_1(\kappa) + D_2(\kappa), 
\end{align*}
{where}
\begin{align} \label{def_D_1}
D_1(\kappa) := &\frac{1}{\sqrt{TN}}\sum_{i=1}^{N}  2\varepsilon_{i}(\kappa)' X_i(\kappa) Q_i^{-1} \big\{ Q_i-X_i(\kappa)'X_i(\kappa)/(T\kappa) \big\} Q_i^{-1}[\beta_i-\bar \beta]\cdot \mathbbm{1}_{\mathscr{B}^c}\\
D_2(\kappa) := & \mathcal{O}(\frac{1}{\sqrt{TN}}\sum_{i=1}^{N}   \| Q_i-X_i(\kappa)'X_i(\kappa)/(T\kappa)\|^2) \nonumber
\end{align}
We now have to show the uniform rate of $o_P(1)$ for each of these terms individually.
According to Proposition \ref{proposition_regressors}
$$
\sup_{\kappa \in I} D_2(\kappa)  =  \mathcal{O}_P(\sqrt{N}/T) = o_P(1)
$$
and thus it remains to derive an upper bound for  $\sup_{\kappa \in I}D_1(\kappa)$. It is easy to show that $\E[ D_1(p)^2 ] =o(1)$ (see Proposition \ref{moments_for_mixing}), and we will establish in Section \ref{secC3} that 
\begin{equation} \label{moricz_term}
\mathbb{E}\big[\sup_{\kappa \in I}[D_1(\kappa)-D_1(p)]^{2}\big] \le C T^{-\zeta}
\end{equation}
for some $\zeta>0$. This directly implies 
\begin{equation*}
 \mathbb{E}\big[ \sup_{\kappa \in I} D_1(\kappa)^2\big]  =o(1). \\
\end{equation*}
\end{proof}

\section{Bounds for sequential Processes}\label{App_D}

\subsection{Technical Prerequisites}

In order to conduct our proofs, we frequently consider sums of mixing random variables $(\Xi_z)_{z \in \mathcal{M}}$, where $\mathcal{M} \subset \Z^d$ is equipped with the maximum metric. In the next proposition we present an upper bound for the moments of such a random variable. Recall that $x \lor y$ denotes the maximum of two numbers $x$ and $y$.

\begin{prop} \label{moments_for_mixing}
Suppose that $(\Xi_z)_{z \in \mathcal{M}}$ is an $\alpha$-mixing field of random variables (see Definition \ref{def_mixing}), that $\phi > 1$ and $\chi>0$ such that $\E|\Xi_z|^{\phi+\chi}<\infty \, \forall z \in \mathcal{M}$. Furthermore, assume that $f(c)<\infty$, where 
\begin{equation} \label{hd3}
f(u):= \sum_{r \ge 1}  r^{d(u-1)-1}\alpha(r)^{\frac{\chi}{\chi+u}}<\infty,
\end{equation}
and $c$ is the smallest even integer $c \ge \phi$. Then there exists a constant $C_1 $ 
$$
\E \Big| \sum_{z \in \mathcal{M}} \Xi_z \Big|^\phi \le C_1 \Big( \sum_{z \in \mathcal{M}}  \big\{ \E\big[| \Xi_z |^{\phi+\chi}\big]^{\frac{\phi}{\phi+\chi}}\big\} \lor 1 \Big)^{(\phi/2)\lor 1}.
$$
If $\phi = 2$, both '' $\lor 1$'' can be dropped in the above formula. 
The constant $C_1$ only depends on the mixing coefficients and can be upper bounded as follows: $C_1 \le C_2 C_3$, where $C_2$ only depends on $\phi$ and $\chi$ and $C_3 = C_3(f(1), f(2),...,f(c))$ is monotone in each component.
\end{prop}

The proposition shows that the moments of order $\phi$  of the sum $\sum_{z \in \mathcal{M}} \Xi_z$ are (for $\phi\ge 2$) of size $\mathcal{O}(|\mathcal{M}|^{\phi/2})$, as in the independent case. The result follows from a Rosenthal type inequality for mixing random fields, which was first suggested in this setting by \cite{Doukhan}. However, we have based our result on the slightly weaker version derived by \cite{Fazekas}, in order to control the constant $C_1$ more precisely. The theorem shows that the constant only depends on  the mixing coefficients in a monotone fashion.\\ 

In our proofs, we consider stochastic processes indexed in $\kappa$. Therefore, we need inequalities to establish not only pointwise, but uniform boundedness. The following, simple result plays an important role in this context. Notice that in this proposition we consider a process indexed in a discrete index set, whereas $\kappa$ lives on the continuous interval $I$. However, all processes introduced in Section \ref{sec31} are for finite samples discrete (jumping) processes and only asymptotically continuous.  This makes the following result applicable.

\begin{prop} \label{maximal_inequality}
Suppose a stochastic process $\mathpzc{P}:\{Q, Q+1,...,R\} \times \Omega \to \R$ is given, where $R, Q \in  \N$ with $Q<R$ and $\E  \mathpzc{P}(q)^2 < \infty$ for $q=Q,...,R$. Furthermore, suppose that for any $Q \le L \le U \le R$ it holds that 
\begin{equation} \label{moricz_inequality} 
   \E[ \mathpzc{P}(U)- \mathpzc{P}(L)]^2 \le C_4 |U-L| 
\end{equation}
for some $C_4>0$. Then for any ${\zeta_1}>0$, there exists a constant $ C_5$, only depending on ${\zeta_1}$ such that 
$$
\E \sup_{q=Q, Q+1,...,R} [ \mathpzc{P}(q)- \mathpzc{P}(Q)]^2 \le C_4 C_5 (R-Q)^{1+{\zeta_1}}.
$$
\end{prop}

\begin{proof}
 We notice that one can rewrite $ \mathpzc{P}(U)- \mathpzc{P}(L)$ as a sum of random variables 
$$
 \mathpzc{P}(U)- \mathpzc{P}(L)=\sum_{s=L+1}^U[ \mathpzc{P}(s)- \mathpzc{P}(s-1)].
$$
Moreover it holds by assumption that
$$
\E[ \mathpzc{P}(U)- \mathpzc{P}(L)]^2 \le C_4 |U-L| \le C_4 |U-L|^{1+{\zeta_1}} 
$$
We now apply Theorem 3.1 from \cite{Moricz} with $g(i,j)=|i-j|$, $\alpha =1+{\zeta_1}$, which directly yields the desired result. 
\end{proof}

\subsection{Moment Bounds}

In this section, we gather a few results on the moments of random variables, which are used in the other sections. The proofs are based on applications of Propositions \ref{moments_for_mixing} and \ref{maximal_inequality}.

\begin{prop} \label{prop_all_moments}
Suppose that the Assumptions $(N), (\beta), (\varepsilon), (X)$  hold. Then the following statements are true, where $C>0$ is a fixed constant (possibly depending on $\zeta$) and $\zeta>0$ can be made arbitrarily small. 
\begin{itemize}
    \item[i)] For any fixed but arbitrary $i \in \{1,...,N\}$
$$
 \E  \|\varepsilon_{i}(\kappa)' X_i(\kappa)/\sqrt{T}\|^4 \le C.
 $$
\item[ii)] For any fixed but arbitrary $i \in \{1,...,N\}$
$$
 \E \Big[\sup_{\kappa \in I}  \left\| \varepsilon_{i}(\kappa)X_i(\kappa)/\sqrt{T\kappa}\right\|^4 \Big] \le C T^\zeta.
 $$
 \item[iii)] For any fixed but arbitrary $i \in \{1,...,N\}$
$$
 \E \Big[\sup_{\kappa \in I}  \left\| \varepsilon_{i}(\kappa)/\sqrt{T}\right\|^4 \Big] \le C T^\zeta.
 $$
\item[iv)] For fixed but arbitrary $k \in \{1,...,K\}$ and $i \in \{1,...,N\}$
    $$
    \E \bigg[\sup_{\kappa \in I} \bigg| \frac{1}{\sqrt{NT\kappa} }\sum_{i=1}^{N} \sum_{t=1}^{\lfloor T \kappa \rfloor} \varepsilon_{i,t} (x_{i,t}'Q_i^{-1})_{k}\bigg|^2 \bigg]\le C T^\zeta.
    $$
\item[v)] For any fixed but arbitrary $l,k \in \{1,...,K\}$, $i \in \{1,...,N\}$
$$
 \E \Big[ \sup_{\kappa \in I}\Big | \frac{1}{T} \sum_{1 \le s,t \le \lfloor T \kappa \rfloor}  (x_{i,t})_k (x_{i,s})_l [\varepsilon_{i,t}\varepsilon_{i,s} - \E \varepsilon_{i,t}\varepsilon_{i,s} ]
 \Big|^2\Big]
 \le C T^{\zeta}.
 $$
\item[vi)] For fixed but arbitrary $h \in \{1,...,b\}$ and $i \in \{1,...,N\}$
$$
    \E\bigg[ \sup_{\kappa \in I} \bigg| \frac{1}{{\lfloor T \kappa \rfloor}} \sum_{t=1}^{\lfloor T \kappa \rfloor -h}
    \big \{ 
    \varepsilon_{i,t}\varepsilon_{i,t+h}-\E[ \varepsilon_{i,t}\varepsilon_{i,t+h}] \big \} \bigg|^2 \bigg]\le C b^2 T^{-1+\zeta}.
$$
 \item[vii)]
With $\breve E_2$ defined in \eqref{def_E_2_breve}
$$
    \E \Big[\sup_{\kappa \in I} \breve E_2(\kappa)^2 \Big ] \le C T^{\zeta-1}.
$$
\end{itemize}
\end{prop}

\begin{proof}

\textbf{i):} We begin by noticing that $\varepsilon_{i}(\kappa)' X_i(\kappa)$ can be written as a sum
$$
\varepsilon_{i}(\kappa)' X_i(\kappa) = \sum_{t=1}^{\lfloor T \kappa \rfloor} \varepsilon_{i,t} x_{i,t} = \Big( \sum_{t=1}^{\lfloor T \kappa \rfloor} \varepsilon_{i,t} x_{i,t,k}  \Big)_{k=1,...,K}.
$$
It suffices to bound the fourth moment entry-wise, so let $k \in \{1,...,K\}$ be fixed but arbitrary. Notice that the moments of order $2M$ of $\varepsilon_{i,t} x_{i,t,k}$ are uniformly bounded, as
$$
\E(\varepsilon_{i,t} x_{i,t,k})^{2M} = \E(\E[\varepsilon_{i,t}^{2M}|\mathbf{X}] x_{i,t,k}^{2M}) \le C\E(x_{i,t,k}^{2M}) \le C.
$$
Here we have used exogeneity (Assumption $(\varepsilon), (1)$), existence of conditional moments of order $2M$ ($(\varepsilon), (3)$) and finally existence of $2M$-moments of the regressors (Assumption $(X)$, $(2)$). 
We now apply Proposition \ref{moments_for_mixing} for $\phi=c=4$ (moment), $d=1$ (dimension of the array) and $\chi=2M-4$ (existing moments of  $\varepsilon_{i,t} x_{i,t,k}$) to see that indeed
$$
\E \Big(\sum_{t=1}^{\lfloor T \kappa \rfloor} \varepsilon_{i,t} x_{i,t,k}  \Big)^2 \le C.
$$
Notice that the mixing condition \eqref{hd3} is met, as $a> 3M/(M-2)$ by Assumption $(\varepsilon), (5)$. This proofs the assertion. \\

\textbf{ii):} 
Recall the proof of the previous step: We can write $\varepsilon_{i}(\kappa)' X_i(\kappa)$  as a sequential sum process and in order to show the desired upper bound we can restrict ourselves to proving it for the $k$th entry (with $k \in \{1,...,K\}$ fixed but arbitrary). First notice that for any pair $(L,U)$ of natural numbers which satisfies $1 \le L <U \le T$ the following inequality holds 
$$
 \E \Big|\sum_{t=L}^U \varepsilon_{i,t} x_{i,t,k}' /\sqrt{T} \Big|^4 \le C (|U-L|/T)^2~,
$$
where $C$ is independent of $i, N, T, L, U$. The argument is the same as in the previous step (Application of Proposition \ref{moments_for_mixing}, with the same parameter specifications). This inequality is the sufficient condition \eqref{moricz_inequality} and we can hence apply Proposition \ref{maximal_inequality}, which directly implies: 
$$
\E\Big[ \sup_{p \le \kappa 1} \Big| \sum_{t=1}^{\lfloor T \kappa \rfloor} \varepsilon_{i,t} x_{i,t,k}  \Big|^4 \Big]\le C T^\zeta.
$$

\textbf{iii):} The proof follows in the same way as the previous one and is therefore omitted. \\

\textbf{iv):}
Let $k$ be fixed but arbitrary. Notice that the random variables $(\varepsilon_{i,t} (x_{i,t}'Q_i^{-1})_{k})_{t=1,...,T}$ are $\alpha$-mixing according to Assumption $(\varepsilon), (5)$ and that their $2M$-moments exist and are uniformly bounded. To see this consider the following calculation:
\begin{align*}
   \E( \varepsilon_{i,t} (x_{i,t}'Q_i^{-1})_{k})^{2M}
   =  \E\{ \E [\varepsilon_{i,t}^{2M}|\mathbf{X}] (x_{i,t}'Q_i^{-1})_{k}^{2M}\} \le C \E\{ (x_{i,t}'Q_i^{-1})_{k}^{2M}\} \le  C \E\{ \|x_{i,t}\|^{2M}\} \le C.
\end{align*}
In the first inequality we have used Assumption $(\varepsilon)$, $(3)$, in the second the uniform boundedness of $Q_i^{-1}$ (see Assumption $(X), (3)$) and in the last Assumption $(X), (2)$
Proposition \ref{moments_for_mixing}
(with $c=\phi=d=2$ and $\chi=2M-2$) implies that 
$$
\E \bigg\{ \sum_{i=1}^{N} \sum_{t=L}^{U} \varepsilon_{i,t} (x_{i,t}'Q_i^{-1})_{k}  \bigg\}^2 \le C  N |U-L|,
$$
for any $1 \le L \le U \le T$ and 
some universal constant $C$ (note that assumption \eqref{hd3} indeed holds due to our mixing conditions). 
 This inequality corresponds to the condition \eqref{moricz_inequality} of Proposition \ref{maximal_inequality} 
 which, yields the desired result for some universal constant $C$ depending on $\zeta$ but not on $N$ or $T$.\\
 
\textbf{v):} The proof works by the same techniques as that of the previous step and is therefore omitted.\\

\textbf{vi):} Recall the definition of $b$ in Proposition \ref{theorem4}. The proof again follows by the maximal inequality in Proposition \ref{maximal_inequality}. In order to apply it, we demonstrate that
\begin{align} \label{need_Fazekas_2}
\E\Big[ \sum_{t=L}^{U}\varepsilon_{i,t}\varepsilon_{i,t-h}-\E\varepsilon_{i,t}\varepsilon_{i,t-h} \Big] ^2\le C b^2 |U-L|,
\end{align}
for any  $1 \le L \le U \le T$ and a constant  $C>0$ which is  independent of $h$. 
The inequality
\eqref{need_Fazekas_2} follows by a similar, but stronger result than Proposition \ref{moments_for_mixing}. To see the necessity of a stronger result, notice that while $\varepsilon_{i,t}\varepsilon_{i,t-h}$  for fixed $h$ still satisfies the mixing condition described in Assumption 
$(\varepsilon), (5)$,  the constant in the mixing condition is not uniform in $h$ anymore. We hence use Remark $1$ in \cite{Fazekas}, which shows that the constant satisfies $C(h) \le C h \le C b$. An application of Theorem 1 of \cite{Fazekas}  yields the desired result in \eqref{need_Fazekas_2} (the conditions of the cited  result  match those of Proposition \ref{moments_for_mixing}). From \eqref{need_Fazekas_2} together with Proposition \ref{maximal_inequality} follows the desired result.\\
To get an intuitive understanding for the argument, the reader may consider a pointwise application of Proposition \ref{moments_for_mixing} for $\phi=c=2, \chi=M-2$, which if $C$ \textit{was} uniform gives \eqref{need_Fazekas_2}. \\

\textbf{vii):} We want to apply the maximal inequality from Proposition \ref{maximal_inequality}. 
  For this purpose, we define for $k=1, \ldots, K$
  $$
  \breve E_{2,k}(L) :=\sum_{i=1}^N \sum_{s,t=1}^{L}  ( x_{i,t}' Q_i^{-1})_k (\varepsilon_{i,t} \varepsilon_{i,s}-\sigma_{i,i} \tau(t-s))  ( x_{i,s}' Q_i^{-1})_k.
 $$
 and notice that by construction $\breve E_2(\kappa) = \frac{1}{(T\kappa)\sqrt{TN}} \sum_{k=1}^K \breve E_{2,k} ( \lfloor T \kappa \rfloor ) $. 
We consider each term $\breve E_{2,k}$ separately. More precisely, we show that for any $k \in \{1,...,K\}$ the inequality
\begin{equation} \label{cond_E_2_Prop_2}
\E \big| \breve E_{2,k}(L) - \breve E_{2,k}(U) \big|^2 \le C(|U-L|TN)
\end{equation}
holds, for all $1 \le L \le U \le T$ and some universal constant $C$. To see that \eqref{cond_E_2_Prop_2} holds, we first decompose the difference $ \breve E_{2,k}(L) - \breve E_{2,k}(U) $ as follows: 
\begin{align*}
\big| \breve E_{2,k}(L) - \breve E_{2,k}(U) \big|
\le & 2  \Big| \sum_{i=1}^N \sum_{s=1}^{L}\sum_{t=L}^{U}  ( x_{i,t}' Q_i^{-1})_k (\varepsilon_{i,t} \varepsilon_{i,s}-\sigma_{i,i} \tau(t-s)) ( x_{i,s}' Q_i^{-1})_k \Big|\\
 +&  \Big| \sum_{i=1}^N \sum_{s=L}^{U}\sum_{t=L}^{U}  ( x_{i,t}' Q_i^{-1})_k (\varepsilon_{i,t} \varepsilon_{i,s}-\sigma_{i,i} \tau(t-s)) ( x_{i,s}' Q_i^{-1})_k \Big|
\end{align*} 

Both sums on the right consist  of three-dimensional arrays of strongly mixing random variables with expectation $0$ and existing moments of order $M$. The second moment of each sum is upper bounded in same fashion and thus we only focus on the first one.
We can  apply Proposition \ref{moments_for_mixing} with $d=2, c=2, \chi=M-2$, which gives 
$$
\E \Big| \sum_{i=1}^N \sum_{s=1}^{L}\sum_{t=L}^{U}  ( x_{i,t}' Q_i^{-1})_l (\varepsilon_{i,t} \varepsilon_{i,s}-\sigma_{i,i} \tau(t-s)) ( x_{i,s}' Q_i^{-1})_l  \Big|^2 \le C(|U-L|TN).
$$
(note that  condition \eqref{hd3}
in Proposition \ref{moments_for_mixing} is satisfied because  $a>2M/(M-2)$).  
This implies \eqref{cond_E_2_Prop_2} and hence, as $K$ is finite we
obtain 
\begin{align*}
& \E\Big[\sum_{k=1}^K \breve E_{2,k}(L) -\sum_{k=1}^K  \breve E_{2,k} (U) \Big]^2
  \le C(|U-L|TN).
\end{align*}
Defining $\breve E_{2}(0)=0$ it follows  from 
 Proposition \ref{maximal_inequality}
$$  
\E  \Big [ \sup_{\kappa \in I} |\breve E_{2}(\kappa)| \Big ] ^2 \le \E \Big [ \sup_{0 \le \kappa \le 1} |\breve E_{2}(\kappa)|^2\Big ] \le CT^{\zeta-1}
$$

\end{proof}

\subsection{Proof of
(\ref{moricz_term})}
\label{secC3} 
 
In this section we prove the identity \eqref{moricz_term}. 
Recall the definition of the stochastic process $D_1(\kappa)$ defined in \eqref{def_D_1}. In order to make the subsequent derivations easier, we use a different index notation in the course of this proof. Instead of using a notation  depending on $\kappa$ to indicate that a quantity depends on the first  
$\lfloor T \kappa \rfloor$ data points, we use the notation
$\lfloor T \kappa \rfloor$. For example,
$\varepsilon_i(\kappa)$ will be denoted by $\varepsilon_i(\lfloor T \kappa \rfloor)$, which  also  explains the notation $\varepsilon_i(L)$  meaning that the first $L$ data points are used. Similarly, we use the notations $D_1(L)$
instead of $D_1(\kappa)$.
\\
We want to employ Proposition \ref{maximal_inequality} and therefore consider for  $\lfloor p T \rfloor \le L \le U \le T$ the difference $D_1(U)-D_1(L)$. If we can show 
\begin{equation} \label{hd10}
    \E(D_1(U)-D_1(L))^2 \le CT^{-2\zeta-1} |U-L|
    \end{equation}
  (see condition \eqref{moricz_inequality}) for some $\zeta>0$ and some constant $C$, it follows by  Proposition \ref{maximal_inequality} that 
\begin{equation*} \label{hd11}
\mathbb{E}[\sup_{\kappa \in I}[D_1(T)-D_1(\lfloor p T \rfloor )]^{2}] \le C T^{\zeta-2\zeta}=C T^{-\zeta},
    \end{equation*}
    which is \eqref{moricz_term}. \\
    
For a proof of \eqref{hd10} we use the decomposition 
\begin{align}   \label{hd12}
D_1(U)-D_1(L)= A_1+A_2+A_3, 
\end{align} 
where 
\begin{align*} 
A_1:=& \frac{U/T}{\sqrt{TN}}\sum_{i=1}^{N}  2\{\varepsilon_{i}(U)-\varepsilon_{i}(L)\}' X_i(U) Q_i^{-1} \big\{ Q_i-X_i(U)'X_i(U)/U \big\} Q_i^{-1}[\beta_i-\bar \beta]\mathbbm{1}_{\mathscr{B}^c}\\
A_2:=&\frac{U/T}{\sqrt{TN}}\sum_{i=1}^{N}  2\varepsilon_{i}(L)' \{X_i(U)-X_i(L)\} Q_i^{-1} \big\{ Q_i-X_i(U)'X_i(U)/U \big\} Q_i^{-1}[\beta_i-\bar \beta]\mathbbm{1}_{\mathscr{B}^c}\\
A_3:=&\frac{U/T}{\sqrt{TN}}\sum_{i=1}^{N}  2\varepsilon_{i}(L)' X_i(L) Q_i^{-1} \big\{ X_i(L)'X_i(L)/L-X_i(U)'X_i(U)/U \big\} Q_i^{-1}[\beta_i-\bar \beta]\mathbbm{1}_{\mathscr{B}^c}.
\end{align*}
Obviously, \eqref{hd10} follows from a corresponding estimate 
for each term $A_\ell$ $(\ell =1 , 2 ,3)$, i.e.
\begin{equation}  \label{hd15}
\E A_\ell^2 \le C |U-L| T^{-2\zeta-1}
\end{equation} 
for a sufficiently small $\zeta>0$.
We focus on the terms $A_1$ and $A_3$, the term $A_2$ can be treated similarly. 

For the  term $A_1$ we use  Proposition \ref{moments_for_mixing} with  $\phi=c=2$ and $\chi<2$ but very  close $2$ (we will specify it more precisely below), such that
$$
a> \frac{3(M+2 \lor 2 \eta)}{M-2 \lor 2 \eta} > \frac{\chi +2}{\chi}
$$
(recall  inequality \eqref{hd14} in Assumption $(\varepsilon), (5)$).
Hence, the summability condition \eqref{hd3} in Proposition \ref{moments_for_mixing} applies and we get 
\begin{equation}
 \label{hd16}   
\E A_1^2 \le \frac{C }{N} \sum_{i=1}^N  A_{1,i}
\end{equation}
where 
$$
 A_{1,i} = 
\E[|2\{\varepsilon_{i}(U)-\varepsilon_{i}(L)'\} X_i(U) Q_i^{-1} \big\{ Q_i-X_i(U)'X_i(U)/U \big\} Q_i^{-1}[\beta_i-\bar \beta]\mathbbm{1}_{\mathscr{B}^c}|^{2+\chi}]^{\frac{2}{2+\chi}}
$$
We can now upper bound the term $A_{1,i}$ follows:
\begin{align} 
\nonumber
A_{1,i} &
\le  C \{  \E \| \{\varepsilon_{i}(U)-\varepsilon_{i}(L)\}'X_i(U)/\sqrt{T} \|^{2+\chi} \|     \big\{ Q_i-X_i(U)'X_i(U)/U \big\} \mathbbm{1}_{\mathscr{B}^c}\|^{2+\chi}\}^{\frac{2}{2+\chi}} \\  
 &\le C B_{1,1,i} \cdot B_{1,2,i}, 
 \label{crazy_terms_1}
 \end{align}
 where 
 \begin{align}
 & B_{1,1,i}:=   \{\E[|\{\varepsilon_{i}(U)-\varepsilon_{i}(L)\}'X_i(U)/\sqrt{T}\|^4\}^{1/2} \nonumber\\
 & B_{1,2,i} := \{\E \| \big\{ Q_i-X_i(U)'X_i(U)/U \big\}  \mathbbm{1}_{\mathscr{B}^c} \|^{4(2+\chi)/(2-\chi)}\}^{(2-\chi)/(2(2+\chi))} \nonumber .
\end{align}
and we have used  the Hölder inequality  and the fact that the quantities 
$\| Q_i^{-1} \|$ and $\beta_i - \bar \beta$ are bounded (see Assumptions $(\beta)$ and $(X), (3)$).
Beginning with $B_{1,1,i}$, wee see that
\begin{align}
\label{hd17}
B_{1,1,i}^2 = & \E \|\sum_{t=L}^U \varepsilon_{i,t} x_{i,t}' /\sqrt{T}\|^4 \le C  \max_{k=1,...,K} \E |\sum_{t=L}^U \varepsilon_{i,t} x_{i,t,k}' /\sqrt{T}|^4 \le C (|U-L|/T)^2~,
\end{align}
where the last equality follows from  Proposition \ref{moments_for_mixing} with $\phi=c=4, \chi=2M-4$, $d=1$ 
(note that the summability condition is met as $a$ in particular satisfies $a>3M/(M-2)$). 

To derive an estimate of $B_{1,2,i}$ we note that the random variable $\|\{ Q_i-X_i(U)'X_i(U)/U \} \mathbbm{1}_{\mathscr{B}^c} \|$ 
is by definition of $\mathscr{B}$ (see \eqref{bad_event}) uniformly bounded (with respect to $i$). Consequently, we may upper bound its moment of arbitrary order by the $2\eta$th  moment, times some constant $C$, which depends only on the range of the random variable. Hence
\begin{align*}
B_{1,2,i}
\le C \{\E \| \big\{ Q_i-X_i(U)'X_i(U)/U \big\} \mathbbm{1}_{\mathscr{B}^c} \|^{2\eta}\}^{(2-\chi)/(2(2+\chi))} \le C T^{-\eta(2-\chi)/(2(2+\chi))}~,
\end{align*}
where we have used Proposition \ref{proposition_regressors} in the last step. The result in  \eqref{hd15} now follows by combining 
this estimate with  \eqref{hd16}, \eqref{crazy_terms_1} and \eqref{hd17}. \\

Next we  turn to the term $A_3$ in the decomposition \eqref{hd12}. The first step is similar as for $A_1$.
According to Proposition \ref{moments_for_mixing} we can upper bound the second moment of $
A_3$ as follows: In the proposition let $\phi=c=2$ and $\chi<2$ be sufficiently close to $2$, such that
$$
a>\frac{\chi +2}{\chi}.
$$
Hence the summability condition in Proposition \ref{moments_for_mixing} applies and we get 
\begin{equation}
\E A_3^2 \le \frac{C }{N} \sum_{i=1}^N  A_{3,i}~,
\end{equation}
where
$$
A_{3,i}=
\E[|\varepsilon_{i}(L)' X_i(L) \sqrt{T}Q_i^{-1} \big\{ X_i(U)'X_i(U)/U-X_i(L)'X_i(L)/L\big\} Q_i^{-1}[\beta_i-\bar \beta]\mathbbm{1}_{\mathscr{B}^c}|^{2+\chi}]^{\frac{2}{2+\chi}}.$$
The Hölder inequality yields 
\begin{align} 
\nonumber
A_{3,i} 
\le & C  \E[|\varepsilon_{i}(L)' X_i(L)/\sqrt{T} \|^{2+\chi} \|  X_i(U)'X_i(U)/U-X_i(L)'X_i(L)/L \|^{2+\chi} \mathbbm{1}_{\mathscr{B}^c}]^{\frac{2}{2+\chi}}   \\
\le & C B_{3,1,i} \cdot B_{3,2,i},
\label{crazy_terms }
\end{align}
{where}
\begin{align} 
& B_{3,1,i} :=    \{\E[|\varepsilon_{i}(L)' X_i(L) /\sqrt{T}\|^4\}^{1/2} \nonumber\\
& B_{3,2,i} := \{\E\|  X_i(U)'X_i(U)/U-X_i(L)'X_i(L)/L \| \mathbbm{1}_{\mathscr{B}^c}) ^{4(2+\chi)/(2-\chi)}\}^{(2-\chi)/(2(2+\chi))}. \nonumber
\end{align}

$B_{3,1,i}$  is bounded according to Proposition \ref{prop_all_moments}, part $i)$.
For $B_{3,2,i}$ we note that 
\begin{align*}
    \frac{T}{U-L} (X_i(U)'X_i(U)/U-X_i(L)'X_i(L)/L)  = B_{3,2,1,i}+B_{3,2,2,i}~,
\end{align*}
where
\begin{align*}
 B_{3,2,1,i} & = - \frac{T}{UL}\sum_{t=1}^{L}(x_{i,t} x_{i,t}'-Q_i) \\
 B_{3,2,2,i} &= \frac{T}{U(U-L)}\sum_{t=L+1}^{U}(x_{i,t} x_{i,t}'-Q_i)
\end{align*}
We can now show for each of these terms separately that
\begin{equation} 
\label{hd20}
\{\E( \| B_{3,2,\ell,i} \| \mathbbm{1}_{\mathscr{B}^c})^{4(2+\chi)/(2-\chi)}\}^{(2-\chi)/(2(2+\chi))} \le C T^{-\zeta}~~~(\ell =1,2),
\end{equation}
where we restrict ourselves to a discussion of $ \| B_{3,2,1,i}
\| \mathbbm{1}_{\mathscr{B}^c}$
for the sake of brevity.  Note that  on the event $\mathscr{B}^c$ the sum $ B_{3,2,1,i} $ is bounded. As a consequence we can upper bound the moment of order $4(2+\chi)/(2-\chi)$ by a constant times a moment of order $1+\delta$, for some arbitrarily small $\delta>0$.
Consequently, it is sufficient to show 
\begin{equation} \label{hd20a}
 \E\| B_{3,2,1,i}\|^{1+\delta}\le C T^{-\zeta}.
\end{equation}
To show this estimate we  first note that  that in the definition of $B_{3,2,1,i}$ we can replace the matrix $Q_i$ in the sum by the expectations $\E [ x_{i,t} x_{i,t}'] $ only incurring an error of size $C /\sqrt{T}$  (see Assumption $(X), (1)$), that is
\begin{equation} \label{hd21}
| B_{3,2,1,i} - \tilde B_{3,2,1,i}  |  \leq C /\sqrt{T}~, 
\end{equation}
where
$$
\tilde B_{3,2,1,i} =- \frac{T}{UL}\sum_{t=1}^{L}(x_{i,t} x_{i,t}'-\E [x_{i,t} x_{i,t}'] ) .
$$
 We 
now consider each entry of the $K \times K$-matrix separately, where we focus as example on the first one (with index $(1,1)$) considering 
\begin{align*}
    \E|(\tilde B_{3,2,1,i})_{1,1}^{1+\delta}|= \left(\frac{T}{U}\right)^{1+\delta}\E \Big |\frac{1}{L}\sum_{t=1}^{L}((x_{i,t} x_{i,t}')_{1,1}-\E (x_{i,t} x_{i,t}')_{1,1}) \Big |^{1+\delta}.
\end{align*}
Notice that $T/U$ is bounded by a constant (recall that the indexing interval for $\kappa$ is bounded away from zero). The above moment is bounded by $C L^{-\zeta}$ for some $\zeta>0$, which follows from 
 Proposition \ref{moments_for_mixing} with  $\phi=1+\delta, c=2, d=1, \chi = M-1-\delta$ (note that  the summability condition holds for $a>(M-3)/(M-1)$ if $\delta$ is chosen sufficiently small). 
 Now $L^{-\zeta} \le C T^{- \zeta}$ and consequently, using the same argument for the other entries in the matrix $B_{3,2,1,i}$,  we have shown that 
\begin{equation} \label{hd22}
 \E\|\tilde B_{3,2,1,i}\|^{1+\delta}\le C T^{-\zeta}.
\end{equation}
We can now establish \eqref{hd20a}  by combining \eqref{hd21} 
and \eqref{hd22}.
This proves \eqref{hd15} for  the case $\ell=3$.

Similar arguments yield the corresponding statement in the case 
$\ell =2$ and it follows from \eqref{hd12} that \eqref{hd10} holds.


\end{document}